\documentclass[reqno]{amsart}
\usepackage{amssymb}
\usepackage{enumerate}
\usepackage{hyperref}
\usepackage{doi,url}

\usepackage{graphicx}

\newcommand{\nn}{\notag}

\numberwithin{equation}{section}
\newtheorem{theorem}{Theorem}[section]
\newtheorem{proposition}[theorem]{Proposition}
\newtheorem{lemma}[theorem]{Lemma}
\newtheorem{corollary}[theorem]{Corollary}
\newtheorem{definition}[theorem]{Definition}
\newtheorem{remark}[theorem]{Remark}

\newtheorem*{theorem*}{Theorem}
\newtheorem*{corollary*}{Corollary}

\DeclareMathOperator{\supp}{supp}

\DeclareMathOperator{\dist}{dist}

\DeclareMathOperator{\diam}{diam}

\DeclareMathOperator{\Ima}{Im}

\makeatletter
\DeclareRobustCommand\widecheck[1]{{\mathpalette\@widecheck{#1}}}
\def\@widecheck#1#2{%
    \setbox\z@\hbox{\m@th$#1#2$}%
    \setbox\tw@\hbox{\m@th$#1%
       \widehat{%
          \vrule\@width\z@\@height\ht\z@
          \vrule\@height\z@\@width\wd\z@}$}%
    \dp\tw@-\ht\z@
    \@tempdima\ht\z@ \advance\@tempdima2\ht\tw@ \divide\@tempdima\thr@@
    \setbox\tw@\hbox{%
       \raise\@tempdima\hbox{\scalebox{1}[-1]{\lower\@tempdima\box
\tw@}}}%
    {\ooalign{\box\tw@ \cr \box\z@}}}
\makeatother

\renewcommand\H{\mathcal{H}}

\renewcommand\L{\mathrm{L}}
\newcommand\R{\mathbb R}
\newcommand\N{\mathbb N}
\newcommand\C{\mathbb C}
\newcommand\Z{\mathbb Z}
\newcommand\G{\mathbb{G}}

\newcommand\e{\mathrm{e}}
\newcommand{\la}{\langle}
\newcommand{\ra}{\rangle}
\renewcommand\P{\mathbb P}
\newcommand\E{\mathbb E}
\newcommand\cE{\mathcal{E}}

\newcommand\cL{\mathcal{L}}
\newcommand\cB{\mathcal{B}}

\newcommand\cG{\mathcal{G}}
\newcommand\cV{\mathcal{V}}

\newcommand\cA{\mathcal{A}}

\newcommand\cF{\mathcal{F}}
\newcommand\cS{\mathcal{S}}
\newcommand\cY{\mathcal{Y}}

\newcommand\vphi{\varphi}
\newcommand\vtheta{\vartheta}
\newcommand\vrho{\varrho}

\newcommand{\vs}{\varsigma}

\newcommand{\pr}{\prime}

\newcommand\wtilde{\widetilde}

\newcommand\ttau{\widetilde{\tau}}
\newcommand\tzeta{\widetilde{\zeta}}
\newcommand{\partialin}{\partial_{\mathrm{in}}}
\newcommand{\partialex}{\partial_{\mathrm{ex}}}

\newcommand{\bom}{{\boldsymbol{{\omega}}}}

\newcommand\beq{\begin{equation}}
\newcommand\eeq{\end{equation}}

\newcommand{\abs}[1]{\left\lvert #1 \right\rvert}
\newcommand{\norm}[1]{\left\lVert #1 \right\rVert}
\newcommand{\scal}[1]{\left\langle #1 \right\rangle}
\newcommand{\set}[1]{\left\{ #1 \right\}}
\newcommand{\pa}[1]{\left( #1 \right)}

\newcommand{\fl}[1]{\left\lfloor #1 \right\rfloor}
\newcommand{\br}[1]{\left [ #1 \right]}

\newcommand\La{\Lambda}

\newcommand\Th{\Theta}
\newcommand\Ups{\Upsilon}

\newcommand{\eq}[1]{\eqref{#1}}

\newcommand{\up}[1]{^{(#1)}}

\newcommand{\qtx}[1]{\quad\text{#1}\quad}
\newcommand{\mqtx}[1]{\; \ \text{#1}\; \  }
\newcommand{\sqtx}[1]{\;\text{#1}\;}

\newcommand\Chi{\raisebox{.2ex}{$\chi$}}

\begin{document}

\title[Eigensystem multiscale analysis  in energy intervals]{Eigensystem multiscale analysis for Anderson localization in energy intervals}

\author{Alexander Elgart}
\address[A. Elgart]{Department of Mathematics; Virginia Tech; Blacksburg, VA, 24061, USA}
 \email{aelgart@vt.edu}

\author{Abel Klein}
\address[A. Klein]{University of California, Irvine;
Department of Mathematics;
Irvine, CA 92697-3875,  USA}
 \email{aklein@uci.edu}

\thanks{A.E. was  supported in part by the NSF under grant DMS-1210982.}
\thanks{A.K. was  supported in part by the NSF under grant DMS-1001509.}

\date{Version of \today}

\begin{abstract}
We present an eigensystem multiscale analysis  for proving  localization (pure point spectrum with exponentially decaying eigenfunctions, dynamical localization)    for the Anderson model in an energy interval.  In particular, it yields localization for the Anderson model in a nonempty interval at the bottom of the spectrum. This eigensystem multiscale analysis in an energy interval treats all energies of the finite volume operator at the same time, establishing level spacing and  localization of  eigenfunctions with eigenvalues in the energy interval in a fixed box with high probability.  In contrast to the usual strategy, we do not study  finite volume  Green's functions.   Instead, we perform a multiscale analysis based on finite volume eigensystems (eigenvalues and eigenfunctions).   In any given scale we only have decay for eigenfunctions with eigenvalues in the energy interval, and no information about the other eigenfunctions.  For this reason,  going to a larger scale requires new arguments that were not necessary in our previous  eigensystem multiscale analysis for the Anderson model at high disorder, where in a given scale  we have decay for all eigenfunctions.
\end{abstract}

\maketitle

\tableofcontents

\section*{Introduction}

We present an eigensystem multiscale analysis  for proving  localization (pure point spectrum with exponentially decaying eigenfunctions, dynamical localization)    for the Anderson model in an energy interval.  In particular, it yields localization for the Anderson model in a nonempty interval at the bottom of the spectrum. 

 The well known methods developed for proving localization for random Schr\"odinger operators, the multiscale analysis \cite{FS,FMSS,Dr,DK,Sp,CH,FK,GKboot,Kle,BK,GKber} and the fractional moment method \cite{AM,A,ASFH,AENSS,AW},  are based on the study of   finite volume  Green's functions.  Multiscale analyses based on  Green's functions  are performed  either at a fixed energy in a single box, or for all energies but with two boxes with an `either or' statement for each energy. 
 
In  \cite{EK} we provided an  implementation of  a multiscale analysis for the Anderson model at high disorder based on  finite volume eigensystems (eigenvalues and eigenfunctions). In contrast to the usual strategy, we did  not study  finite volume  Green's functions. 
 Information about eigensystems at a given scale was used to derive information about eigensystems at larger scales. This eigensystem multiscale analysis
treats all energies of the finite volume operator at the same time, giving a complete picture in a fixed box.  For this reason it does not use a Wegner estimate as in a Green's functions multiscale analysis, it uses instead   a probability estimate for level spacing  derived by Klein and Molchanov  from Minami's estimate \cite[Lemma~2]{KlM}.  This eigensystem multiscale analysis for the Anderson model at high disorder has been enhanced in \cite{KlT}  by  a bootstrap argument as in \cite{GKboot,Kle}.   

 The motivation for developing an alternative approach to localization is related to a new focus among the mathematical physics community in disordered systems with an infinite number of particles, for which Green's function methods break down.  The  direct study of the structure of eigenfunctions for such systems has been advocated by Imbrie  \cite{Im1,Im2} in a context of both  single and many-body localization.

 The Green's function methods  allow  for proving localization in 
energy intervals, and hence
 localization has also been proved  at fixed disorder  in an interval at the edge of the spectrum (or, more generally, in the vicinity of a spectral gap), and  for  a fixed interval of energies at the bottom of the spectrum  for sufficiently  high disorder. (See, for example, \cite{HM,KSS,FK1,ASFH,GKfinvol,Ki,GKber,AW}.)  These methods do not differentiate between energy intervals and the whole spectrum; they can be used whenever the initial step can be established. 

The results in \cite{EK} yield localization for the Anderson model in the whole spectrum, which in practice requires high disorder. This eigensystem multiscale analysis treats all energies of the finite volume operator at the same time, at  a given scale  we have decay for all eigenfunctions, and the induction step uses information about all eigenvalues and eigenfunctions.  The method does not have a straightforward extension for proving localization in an energy interval, since at any give scale we would only have information (decay) about eigenfunctions corresponding to eigenvalues in the given interval.   For this reason, when performing an eigenfunction multiscale analysis in an energy interval,   going to a larger scale requires new arguments that were not necessary in our previous  eigensystem multiscale analysis for the Anderson model at high disorder, where in a given scale  we have decay for all eigenfunctions.

In this paper  we develop  a version of the eigensystem multiscale analysis tailored to the   establishment of   localization   for the Anderson model in an energy interval. This version yields  localization at fixed disorder  on an interval at the edge of the spectrum (or in  the vicinity of a spectral gap), and  at  a fixed interval at the bottom of the spectrum  for sufficiently  high disorder.  

The Anderson model is    a random   Schr\"odinger  operator $H_\bom$ on  
 $\ell^2(\Z^d)$ (see Definition~\ref{defineAnd}).  Multiscale analyses prove statements about finite volume operators $H_{\bom,\La}$, the restrictions  of $H_\bom$ to  finite  boxes $\La$.  The eigensystem multiscale analysis developed in this article establishes eigensystem localization in a bounded energy interval with good probability at  large scales, as we will now explain. 
 
An eigensystem $\set{(\vphi_j,\lambda_j)}_{j\in J}$ for $H_{\bom,\La}$ consists of eigenpairs 
$(\vphi_j,\lambda_j)$, where  $\lambda_j$ is an eigenvalue for  $H_{\bom,\La}$  and $\vphi_j$ is a corresponding normalized eigenfunction,  such that $\set{\vphi_j}_{j\in J}$ is an orthonormal basis for the finite dimensional Hilbert space $\ell^2(\La)$.   If   all eigenvalues of  $H_{\bom,\La}$  are simple, we can rewrite
 the eigensystem as  $\set{(\vphi_\lambda,\lambda)}_{\lambda\in \sigma(H_{\bom,\La})}$.

We define  eigensystem localization in a bounded energy interval $I$ in the following way.  We fix appropriate exponents
 $\beta,\tau \in (0,1)$ (see \eq{ttauzeta0}), 
 take $m>0$, and  say that a box $\La$ of side $L$ is $(m,I)$-localizing for $H_\bom$ (see Definition~\ref{defmIloc}) if $\La$ is level spacing
 (i.e., the eigenvalues of  $H_{\bom,\La}$ are simple and separated by at least $ \e^{-{L}^\beta}$), and  eigenfunctions corresponding to eigenvalues in the interval $I$ decay exponentially as follows: 
  if $\lambda \in  \sigma(H_{\bom,\La})\cap I$, then there exists $x_\lambda \in \La$ such that  the corresponding eigenfunction $\vphi_\lambda$ satisfies
 \[  \abs{\vphi_\lambda(y) } \le \e^{-m h_I(\lambda) \norm{y-x_\lambda}}\qtx{for all} y \in \La  \qtx{with} \norm{y-x_\lambda}\ge L^\tau,
 \]
where $h_I$  (defined in \eqref{eq:h_I}) is a concave function on $I$, taking the value  one at the center of the interval and the value zero at the endpoints. The modulation of the decay of the eigenfunctions by the function $h_I$ is a new feature of our  method.   

Our multiscale analysis shows that eigenfunction localization  in an energy interval with good probability at some large enough scale implies eigenfunction localization with good (scale dependent and improving as the scale grows) probability for all sufficiently large scales, in a slightly smaller energy interval.   The key step  shows
that  localization at a large scale $\ell$ yields localization at a much larger scale $L$.  The proof proceeds by covering a box $\La_L$ of side $L$ by boxes of side $\ell$,  which are mostly $(m,I)$-localizing, and showing this implies that $\La_L$ is $(m^\pr,I^\pr)$-localizing.   There are always some losses, $m^\pr < m$ and $I^\pr \subsetneq I$, but this losses are controllable, and continuing this procedure we converge  to some rate of decay $m_\infty>0$ and interval  $I_\infty\neq\emptyset$. 

 The eigensystem multiscale analysis in an energy interval $I$ requires a  new ingredient, absent in the  treatment of the system at  high disorder  given in \cite{EK}, where $I=\R$ and $h_I=1$.   In broad terms, the reason is that our energy interval  multiscale scheme only carries information about eigenfunctions with eigenvalues in the interval $I$, and  contains  no information whatsoever concerning  eigenfunctions with eigenvalues that lie outside the interval $I$.   Given  boxes $\La_\ell \subset \La_L$, with $\ell \ll L$,   a  crucial step in our analysis
 shows that if $(\psi,\lambda)
 $ is an eigenpair for $H_{\bom,\La_L}$,  with  $\lambda \in I$  not too close to the  eigenvalues of $H_{\bom,\La_\ell}$ corresponding to eigenfunctions localized deep inside $\La_\ell $, and  the box $\La_\ell $ is $(m,I)$-localizing for $H_\bom$, then $\psi$ is exponentially small  deep inside $\La_\ell $ (see Lemma~\ref{lemdecay2}(ii)).  This is    proven by expanding the values of $\psi$ in $\La_\ell $ in terms of the $(m,I)$-localizing eigensystem  $\set{(\vphi_\nu,\nu)}_{\nu\in \sigma(H_{\bom,\La_\ell})}$ 
 for $H_{\bom,\La_\ell}$.  The difficulty is that we only have decay for the eigenfunctions $\vphi_\nu$ with $\nu\in I$; we know nothing about $\vphi_\nu$ if  $\nu\notin I$.   We overcame this difficulty by showing that the decay of the term containing the latter eigenfunctions  comes from the distance from the eigenvalue $\lambda$ to  the complement of the interval $I$,  using Lemmas~\ref{lemkey} and \ref{lemkey2}.  As a result, it is natural to expect that the decay rate for the  localization of eigenfunctions goes to zero  as the eigenvalues approach the edges of the interval $I$. The introduction of the modulating function $h_I$ in the decay 
 models this phenomenon. 
 
The same difficulty appears if, given an  $(m,I)$-localizing box  $\La $ for $H_\bom$, we try to recover the decay of  
 the Green's function at an energy  $\lambda \in I$ not too close to the  eigenvalues of $H_{\bom,\La}$.  The simplest approach is to decompose the Green's function in terms of an $(m,I)$-localizing eigensystem $\set{(\vphi_\nu,\nu)}_{\nu\in \sigma(H_{\bom,\La})}$ for  $H_{\bom,\La}$:
  \begin{align}\nn
 \scal{ \delta_{x}, (H_{\bom,\La}  -\lambda)^{-1}\delta_{y}}= \sum_{\nu \in \sigma(H_{\bom,\La})} (\nu-\lambda)^{-1} \overline{\vphi_\nu(x)}\vphi_\nu(y).
 \end{align}
The sum over the eigenvalues inside the interval $ I$ can be estimated using the decay of the corresponding eigenfunctions, but we have a problem estimating the sum over eigenvalues outside $I$ since we have no information  concerning the spatial decay properties of the corresponding eigenfunctions.
 To overcome this difficulty, we use a more delicate argument  (see Lemma~\ref{lemtoreg})  that decomposes the Green's function  into a sum of two analytic functions of $H_{\bom,\La}$ with appropriate decay properties  (see Lemmas~\ref{lemkey} and \ref{lemkey2} for details), obtaining  the desired decay of the Green function:
    \beq
 \nn
\abs{ \scal{ \delta_{x}, (H_{\bom,\La}  -\lambda)^{-1}\delta_{y}}} \le  \e^{-m^\pr h_I(\lambda) \norm{x-y}}.
 \eeq

 Readers familiar with the  Green's function multiscale analysis may notice that the modulation by the function $h_I$ is not required there. This has to do with the fact the Green's function approach essentially considers each  energy value separately, while the eigensystem approach treats the whole energy interval simultaneously. 
 A Green's function multiscale analysis is performed at a fixed energy; the modulation of the decay may appear in the starting condition, but not in the multiscale analysis proper.  (The
 starting condition near an spectral edge is usually obtained from the Combes-Thomas estimate, which modulates the decay rate by the distance to the spectral edge.)

A version of our main result,  Theorem~\ref{thmMSA}, can be  stated as follows.  (The exponents $\zeta, \xi \in (0,1)$ and $\gamma >1$ are as in \eq{ttauzeta0}.  $\La_{L} (x) $ denotes the box in $\Z^d$ of side $L$ centered at $x\in \R^d$ as in \eq{defbox}.)

\begin{theorem*}[Eigensystem multiscale analysis]
Let $H_\bom$ be an Anderson model.  Let  $I_0=(E-{A_0},E+{A_0})\subset \R$,   with $E\in \R$ and $A_0>0$, and
$0<  m_0   \le  \tfrac 1 2 \log \pa{1 + \tfrac {A_0}{4d}}$.  Suppose for some scale 
$L_0 $ we have
  \begin{align}\nn
\inf_{x\in \R^d} \P\set{\La_{L_0} (x) \sqtx{is}  (m_0,I_0) \text{-localizing for} \; H_{\bom}} \ge 1 -  \e^{-L_0^\zeta}.
\end{align}
Then, if $L_0$ is sufficiently large,   there exist $m_\infty=m_\infty(L_0)>0$ and  $A_\infty=A_\infty(L_0)\in(0,A_0)$, with $\lim_{L_0\to \infty} A_\infty (L_0) = A_0 \qtx{and} \lim_{L_0\to \infty} m_\infty (L_0) = m_0$,   such that, setting $I_\infty=(E-A_\infty,E+A_\infty)$,   we have
   \begin{align} \nn
\inf_{x\in \R^d} \P\set{\La_{L} (x) \sqtx{is} ( m_\infty , I_\infty)  \text{-localizing for} \; H_{\bom}} \ge 1 -  \e^{-L^\xi}  ,
\end{align}
for all   $ L\ge L_0^\gamma $.  
\end{theorem*}

The theorem yields all the usual forms of Anderson localization on the interval $I_\infty$.   In particular we obtain the following version of  Corollary~\ref{corloc}.

\begin{corollary*}[Localization in an energy interval]
Suppose the   theorem  holds for  an Anderson model $H_{\bom}$.
Then the following holds with probability one:
  \begin{enumerate}
 \item  {$H_{\bom}$}  has  pure point spectrum in the interval $I_\infty$.

\item  If   {$\psi_\lambda$}  is a normalized {eigenfunction} of $H_{\bom}$
with eigenvalue  {$\lambda\in I_\infty$}, then $\psi_\lambda$ is exponentially localized  with rate of decay $ \frac 1{20} m_\infty  h_{I_\infty}(\lambda)$,   more precisely,
\[
\abs{\psi_\lambda(x)} \le C_{\bom,\lambda}\, e^{-  \frac 1{20}m_\infty  h_{I_\infty}(\lambda)\norm{x}} \qquad \text{for all}\quad  x \in \Z^{d}.
\]
\end{enumerate}
\end{corollary*}

In particular, our results prove localization at the bottom of the spectrum. Let  $H_{\bom}$ be an Anderson model, and  let $E_0$ be  the bottom of the almost sure spectrum of $H_{\bom}$.
 We  consider intervals at the bottom of the spectrum, more precisely,  intervals of the form $J=[E_0 , E_0 + A)$ with $A>0$.  We set $\tilde J= (E_0 -A, E_0 + A)$,
so $ J \cap \Sigma = \tilde J \cap \Sigma$,  call a box $(m,J)$-localizing if it is $(m,\tilde J)$-localizing, etc. The following is a version of Theorem~\ref{themlocbottom}.

\begin{theorem*}[Localization at the bottom of the spectrum] Let  $H_{\bom}$ be an Anderson model,  and fix $ 0<\xi <\zeta <\frac  d {d+2}  $.   Then  there is $\gamma >1$, such that, if $L_0$ is sufficiently large,   there exist $m_{\zeta,\infty} =m_{\zeta,\infty} (L_0)>0$ and  $A_{\zeta,\infty} =A_{\zeta,\infty} (L_0)\in(0,A_0)$, with
\[ \lim_{L_0\to \infty} A_{\zeta,\infty}  (L_0) L_0^{\frac {2\zeta}d}  = C_{d,\mu}  \qtx{and} \lim_{L_0\to \infty} m_\infty (L_0) L_0^{\frac {2\zeta}d} = \frac 1 {9d} C_{d,\mu}, \]  such that, setting $J_{\zeta,\infty}=[E_0,E_0+A_{\zeta,\infty})$,  for all $ L\ge L_0^\gamma$ we have
  \begin{align} \nn
\inf_{x\in \R^d} \P\set{\La_{L} (x) \sqtx{is} ( m_{\zeta,\infty} , J_{\zeta,\infty})   \text{-localizing for} \; H_{\bom}} \ge 1 -  \e^{-L^\xi},
\end{align}
  In particular, the conclusions of  the Corollary  hold in the interval $J_{\zeta,\infty}$.
\end{theorem*}
We also establish localization in a fixed interval at the bottom of the the spectrum, for sufficiently large disorder (Theorem~\ref{thmfixedint}).

Our main results and definitions are stated in Section~\ref{secmain}.  Theorem~\ref{thmMSA} is our main result, which  we prove  in   Section~\ref{secEMSA}.  Theorem~\ref{thmloc}, derived from  Theorem~\ref{thmMSA}, encapsulates localization in an energy interval  for the Anderson model and yields Corollary~\ref{corloc}, which  contains typical statements  of Anderson localization and dynamical localization in an energy interval.  Theorem~\ref{thmloc} and Corollary~\ref{corloc} are proven in Section~\ref{seclocproof}. In Section~\ref{secbottom} we show how to fulfill the starting condition for Theorem~\ref{thmMSA} and establish localization  in an interval at the bottom of the spectrum, for fixed disorder (Theorem~\ref{themlocbottom}) and in a fixed interval for sufficiently large disorder (Theorem~\ref{thmfixedint}). Section~\ref{secprep} contains notations, definitions and lemmas required for the proof of the eigensystem multiscale analysis given in  Section~\ref{secEMSA}. 
The connection with  the Green's functions multiscale analysis is established in Section~\ref{secGreen}.

 \section{Main results}  \label{secmain}

In this article we will use many positive exponents, which will be required to satisfy certain relations.
 We consider  $\xi,\zeta, \beta,\tau \in (0,1)$ and $\gamma >1$ such that
\begin{gather}\label{ttauzeta0}
0<\xi< \zeta<\beta<\frac 1 \gamma <1<\gamma < \sqrt {\tfrac \zeta \xi} \mqtx{and}   \max\set{ \gamma \beta, \tfrac {(\gamma-1)\beta +1}{\gamma}  }      <  \tau <1,
\end{gather}
and note that  
\beq\label{ttauzeta}
0<\xi<\xi\gamma^2< \zeta < \beta <\frac \tau \gamma   <\frac 1 \gamma <\tau <1< \frac {1-\beta}{\tau-\beta} < \gamma <\frac \tau \beta.
\eeq
We also take
\beq\label{ttauzeta2}
\tzeta= \frac {\zeta +\beta}2  \in (\zeta, \beta)  \qtx{and} {\ttau}= \frac {1 +\tau}2  \in (\tau,1),
\eeq
 so
\beq\label{gamtzetabeta}
(\gamma-1)\tzeta  +1 <(\gamma-1)\beta  +1< \gamma \tau.
\eeq
We also consider  $\kappa\in (0,1) $ and   $\kappa^\pr \in [0,1)$ such that
\beq\label{gamtzetabeta2}
 \kappa+\kappa^\pr < \tau - \gamma \beta.  
\eeq
We set 
\beq \label{defvrho}
\vrho= \min\set{\kappa, \tfrac{1- \tau}2, \gamma \tau- (\gamma-1)\tzeta  -1}, \qtx{note} 0<\kappa \le \vrho <1,
\eeq 
and choose
\beq\label{vsdef}
\vs \in (0,1-\vrho] , \qtx{so} \vrho < 1-\vs .
\eeq

We consider these  exponents fixed  and do    not make explicit the dependence  of constants on them.   We write  $\Chi_A$ for  the characteristic function of the set $A$.
By a  constant  we  always mean a finite constant.  We will use  $C_{a,b, \ldots}$, $C^{\pr}_{a,b, \ldots}$,  $C(a,b, \ldots)$, etc., to  denote a constant depending  on the parameters
$a,b, \ldots$. Note that $C_{a,b, \ldots}$ may denote different constants in different equations, and even in the same equation.

Given a scale $L\ge 1$,  we sets
\[
  L=\ell^\gamma \; (\text{i.e.,}\;\ell= L^{\frac 1 \gamma}), \quad  L_{\tau}=\fl{L^{\tau}}, \qtx{and} L_{{\ttau}}= \lfloor{L^{{\ttau}}}\rfloor.
\]

If $x=(x_1,x_2,\ldots, x_d)\in \R^d$, we set $\abs{x}=\abs{x}_2= \pa{\sum_{j=1}^dx_j^2}^{\frac 12}$,   and $\norm{x}=\abs{x}_\infty= \max_{j=1,2,\ldots,d} \abs{x_j}$. If $x\in \R^d$ and   
$\Xi\subset \R^d$, we set $\dist (x,\Xi)= \inf_{y\in \Xi} \norm{y-x}$.  The diameter of a set $\Xi\subset \R^d$ is given by $\diam \Xi= \sup_{x,y \in \Xi} \norm{y-x}$.

$H$ we will always denote a discrete  Schr\"odinger operator, that is,   an operator 
 $H=-\Delta +V$ on $\ell^2(\Z^d)$, where where
 $\Delta$ is the  (centered) discrete  Laplacian:  
 \begin{equation}\label{defDelta}
  (\Delta \varphi)(x):=  \sum_{\substack{y\in\Z^d\\ |y-x|=1}} \varphi(y)  \qtx{for} \varphi\in\ell^2(\Z^d),
\end{equation} 
 and  $V$ is a  bounded potential. Given   $\Phi\subset \Theta\subset \Z^d$, we consider $\ell^2(\Phi)\subset \ell^2(\Theta)$ by extending functions on $\Phi$ to functions on $\Theta$ that are identically $0$ on $\Theta\setminus \Phi$.   If $\Th \subset \Z^d$ and $\vphi \in  \ell^2(\Theta)$,  we let $\norm{\vphi}=\norm{\vphi}_2$ and   $\norm{\vphi}_\infty= \max_{y \in \Theta} \abs{\vphi(y)}$.

Given  $\Theta\subset \Z^d$, we let $H_\Theta$ be the restriction of $ \Chi_\Theta H\Chi_\Theta$ to $\ell^2(\Theta)$.
 We call $(\vphi,\lambda)$ an eigenpair for $H_\Th$ if $\vphi\in \ell^2(\Theta)$ with $\norm{\vphi}=1$,  
 $\lambda \in \R$, and $H_\Theta \vphi=\lambda \vphi$.  (In other words,    $\lambda$ is an eigenvalue for $H_\Th$ and $\vphi$ is a corresponding normalized eigenfunction.)  A collection $\set{(\vphi_j,\lambda_j)}_{j\in J}$ of eigenpairs for $H_\Th$ will be called an eigensystem for $H_\Th$ if   $\set{\vphi_j}_{j\in J}$ is an orthonormal basis for $\ell^2(\Th)$. If   all eigenvalues of  $H_\Th$ are simple, we can rewrite
 the eigensystem as  $\set{(\psi_\lambda,\lambda)}_{\lambda\in \sigma(H_\Th)}$.
 
 Given ${\Theta}\subset \Z^d$, a function  $\psi\colon{\Theta} \to \C$ is called a generalized eigenfunction for $H_{\Theta}$ with generalized eigenvalue $\lambda \in \R$ if $\psi$ is not identically
 zero and 
 \beq \label{pointeig} 
 \scal{(H_\Th -\lambda)\vphi, \psi}=0 \qtx{for all} \vphi \in \ell^2(\Th)\quad \text{with finite support}.
 \eeq
 In this case we call $(\psi,\lambda)$ a generalized eigenpair for $H_{\Theta}$. (Eigenfunctions are generalized eigenfunctions, but we do not require generalized eigenfunctions to be in $\ell^2(\Th)$.)

 For convenience we consider boxes in $\Z^d$ centered at points of $\R^d$. The box in $\Z^d$  of side $L>0$  centered at $x\in \R^{d}$ is given by
\begin{align} \label{defbox}
\La_L(x)&=\La^\R_L(x)\cap \Z^d, \qtx{where}
\La^\R_L(x)= \set{y \in \R^d;\  \norm{y-x} \le  \tfrac{L}{2}}.
\end{align}
By a box  $\La_L$ we will mean a box $\La_L(x)$ for some $x\in \R^d$.  It is easy to see that for all $L\ge 2$ and  $x\in \R^d$ we have $(L-2)^{d}<  \abs  {\La_L(x)}\le  (L+1)^{d}$.

\begin{definition}  Given $R>0$,   a finite set $\Theta \subset \Z^d$ will be called  
$R$-level spacing  for $H$ if   all eigenvalues of $H_\Th$ are simple and  $\abs{\lambda- \lambda^\pr}\ge \e^{-{R}^\beta}$ for all $\lambda, \lambda^\pr\in \sigma(H_{\Theta})$, $\lambda\ne  \lambda^\pr$.

If $\Theta$ is a box $\La_L$  and $R=L$,  we will simply say that
 $\La_L$ is level spacing  for $H$.
\end{definition}

 \begin{definition}\label{defxmloc}  Let $\La_L$ be a box, $x\in \La_L$, and  $m\ge 0$.  Then   $\vphi\in \ell^2(\La_L)$  is said to be $(x,m)$-localized if  $\norm{\vphi}=1$ and\beq\label{hypdec}
\abs{\vphi(y)}\le \e^{-m\norm{y-x}}\qtx{for all} y \in \La_L \qtx{with} \norm{y-x}\ge L_\tau.
\eeq
\end{definition}

Note that  $m=0$ is allowed in Definition~\ref{defxmloc}.

\begin{definition}  \label{defmIloc} 
Let  $ J=(E-B,E+B) \subset I=(E-A,E+ A)$, where $E\in \R$ and  $0<B\le A$,  be bounded open intervals with the same center,
 and let   $m>0$.   A box $\La_L$ will be called $(m,J,I)$-localizing for $H$ if the following holds:

\begin{enumerate}
\item  $\La_L$ is level spacing  for $H$. 
\item  There exists an  $(m,J,I)$-localized eigensystem for $H_{\La_L}$, that is,   an  eigensystem   $\set{(\vphi_\nu, \nu)}_{\nu \in \sigma(H_{\La_L})}$ for $H_{\La_L}$ such that for all $\nu \in \sigma(H_{\La_L})$ there is $x_\nu\in \La_L$ such that  $\vphi_\nu$ is $(x_\nu, m \Chi_{J}(\nu)h_{ I}(\nu))$-localized,
where the modulating function $h_I$ is defined by
\beq\label{eq:h_I}
h_{ I}(t) =h\pa{{\tfrac{{t-E}}{{A}}}} \mqtx{for} t \in \R, \mqtx{where} h(s)=\begin{cases} 
1-s^2&\text{if} \; \; s\in[0,1)\\0 & \text{otherwise}\end{cases}.
\eeq

\end{enumerate}
  We will  say that $\La_L$ is  $(m,I)$-localizing for $H$ if  $\La_L$ is  $(m,I,I)$-localizing for $H$.
 \end{definition}
 
 Note that     $h_{ I}(t)  >0 \iff t\in I$, in particular   $h_I=\Chi_Ih_I$. Since $\Chi_{J}h_{ I} \ge  h_{ J}$ , if  $\La_L$ is $(m,J,I)$-localizing for $H$ it is also
$(m,J)$-localizing for $H$.

\begin{remark} In \cite{EK} we had $I=\R$ and $h_\R=1$, and  called a box $\La_L$ $m$-localizing if it was level spacing for $H$ and  for all $\nu \in \sigma(H_{\La_L})$ there is $x_\nu\in \La_L$ such that  $\vphi_\nu$ is $(x_\nu, m)$-localized.
\end{remark}

  Given an interval $I=(E-{A},E+{A})$ and  scales $\ell,L>1 $, we use the notation
\begin{align}\label{defIell}
I_\ell&= (E-{A}(1-\ell^{-\kappa}),E+{A}(1-\ell^{-\kappa})), \\
\nn  
I^\ell&= (E-{A}(1-\ell^{-\kappa})^{-1},E+{A}(1-\ell^{-\kappa})^{-1}).
\end{align}
  We write $ I_\ell^L = \pa{I_\ell}^L = \pa{I^L}_\ell$, note that $I_\ell^\ell=I$, and  observe that
\beq\label{lowerbdh}
 \Chi_{I_\ell} h_I   \ge   \ell^{-\kappa} \Chi_{I_\ell}, \qtx{i.e.,} h_{ I}(t)\ge 1- (1-\ell^{-\kappa})^2\ge \ell^{-\kappa} \mqtx{for all} t\in I_\ell.
\eeq

 \begin{definition}\label{defineAnd} The Anderson model is the
 random discrete  Schr\"odinger  operator
\beq \label{defAnd}
H_{\bom} :=  -\Delta + V_{\bom} \quad \text{on} \quad  \ell^2(\Z^d), 
\eeq 
where $V_{\bom}$ is a random potential:      $V_{\bom}(x)= \omega_x$ for  $ x \in \Z^d$, where
$\bom=\{ \omega_x \}_{x\in
\Z^d}$ is a family of independent 
identically distributed randoms
variables,  whose  common probability 
distribution $\mu$ is non-degenerate with bounded support.  We assume $\mu$ is H\"older continuous of order $\alpha \in ( \frac 12,1]$: 
 \beq\label{Holdercont}
S_\mu(t) \le K t^\alpha \qtx{for all} t \in [0,1],
\eeq
where $K$ is  a constant and  $S_\mu(t):= \sup_{a\in \R} \mu \set{[a, a+t]} $ is the concentration function of the measure $\mu$.
\end{definition}

 It follows from ergodicity (e.g., \cite[Theorem~3.9]{Ki}) that
\beq\label{Sigma}
 \sigma (H_{\bom})= \Sigma:=\sigma(-\Delta) + \supp \mu= [-2 d,2 d] + \supp \mu \;\;\text{with probability one}.
\eeq

 The eigensystem multiscale analysis in an energy interval yields the following theorem.  
 
 \begin{theorem}\label{thmMSA} Let $H_\bom$ be an Anderson model.  Given $m_- >0$,  there exists a 
 a finite scale   $\cL= \cL(d,m_-) $ and a constant  $C_{d,m_-} >0$  with the following property:  Suppose for some scale 
$L_0 \ge \cL$ we have
  \begin{align}\label{initialconinduc9932}
\inf_{x\in \R^d} \P\set{\La_{L_0} (x) \sqtx{is}  (m_0,I_0) \text{-localizing for} \; H_{\bom}} \ge 1 -  \e^{-L_0^\zeta},
\end{align}
where  $I_0=(E-{A_0},E+{A_0})\subset \R$,   with $E\in \R$ and $A_0>0$, and
 \beq\label{upbm25552}
m_- L_0^{-\kappa^\pr} \le  m_0   \le  \tfrac 1 2 \log \pa{1 + \tfrac {A_0}{4d}}.
 \eeq  
 Then for all   $ L\ge L_0^\gamma $ we have
   \begin{align} \label{MSALnok2}
\inf_{x\in \R^d} \P\set{\La_{L} (x) \sqtx{is} ( m_\infty , I_\infty, I_\infty^{L^\frac 1 \gamma})  \text{-localizing for} \; H_{\bom}} \ge 1 -  \e^{-L^\xi}  ,
\end{align}
where, with $\vrho$ as in \eq{defvrho},
 \begin{align}\label{Aminfty}
 A_\infty&=A_\infty (L_0)= A_0 \prod_{k=0}^\infty \pa{1- L_0^{-\kappa\gamma^k}}, \quad I_\infty= (E-{A_\infty},E+{A_\infty}), \\ \nn
  m_\infty&=m_\infty (L_0)= m_0 \prod_{k=0}^\infty \pa{1- C_{d,m_-} L_0^{-\vrho\gamma^k}} <   \tfrac 1 2 \log \pa{1 + \tfrac {A_\infty}{4d}}.
   \end{align}
In particular,  
$\lim_{L_0\to \infty} A_\infty (L_0) = A_0$ and $\lim_{L_0\to \infty} m_\infty (L_0) = m_0 $.

 \end{theorem}

Theorem~\ref{thmMSA} yields all the usual forms of localization on the interval $I_\infty$.
To state these results,
we fix $\nu > \frac d 2$, and for $a \in \Z^{d}$ we let 
 $T_{a} $  be  the operator on   $\ell^2(\Z^d)$ given by multiplication by the function
$T_{a}(x):= \la x-a\ra^{\nu}$, where  $ \la x\ra=  \sqrt{1 + \norm{x}^2}$. Since   $\langle a +b \rangle \le \sqrt{2}\langle a \rangle
\langle b\rangle$, we have $
\| T_{a} T_{b}^{-1} \| \le 2^{\frac  {\nu}  2} \la  a -b \ra^{\nu}$.
A function
$\psi\colon \Z^d \to \C$ will be called a $\nu$-generalized eigenfunction for  the discrete  Schr\"odinger operator $H$ if $\psi$ is a generalized eigenfunction 
and $\norm{T_0^{-1} \psi}<\infty$.  ($\norm{T_0^{-1} \psi}<\infty$ if and only if  $\norm{T_a^{-1} \psi}<\infty$ for all $a\in \Z^d$.)
We let 
$\cV({\lambda})$ denote the collection of  $\nu$-generalized eigenfunctions for $H$ with generalized eigenvalue ${\lambda} \in \R$.
Given   ${\lambda} \in \R$ and $a,b \in \Z^{d}$, we set
 \begin{align} \label{defGWx}
W_{\lambda}\up{a}({b}):=\begin{cases} 
\sup_{\psi \in\cV({\lambda}) }
\ \frac {\abs{\psi(b)}}
{\|T_{a}^{-1}\psi \|}&
 \text{if $\cV({\lambda})\not=\emptyset$}\\0 & \text{otherwise}\end{cases}.
\end{align}
It is easy to see that for  all $a,b,c \in \Z^d$ we have 
\begin{equation}\label{boundGW}
W\up{a}_{\lambda}({a})\le 1,\quad  W\up{a}_{\lambda}({b})\le\la b-a\ra^\nu,  \qtx{and} W\up{a}_{\lambda}({c})\le 2^{\frac \nu 2} \la b-a\ra^\nu 
W\up{b}_{\lambda}({c}).
\end{equation}
  
 \begin{theorem}\label{thmloc}   Suppose the conclusions of  Theorem~\ref{thmMSA} hold  for an  Anderson model $H_{\bom}$, and  let $I=I_\infty$, $m=m_\infty$.   There exists a finite scale   $\cL=\cL_{d,\nu,m_-} $ such that, given  $\cL\le L \in 2\N$ and  $a\in \Z^d$,    there exists an event  $\cY_{L,a}$ with the following properties:
  
  \begin{enumerate}
\item $\cY_{L,a}$  depends only on the random variables $\set{\omega_{x}}_{x \in \Lambda_{5L}(a)}$,  and  
\beq\label{cUdesiredint}
  \P\set{\cY_{L,a} }\ge  1 - C \e^{-L^\xi}.
  \eeq

\item If  $\bom \in \cY_{L,a}$,   for all  $\lambda \in I$ we have that  ($L=\ell^\gamma$)
\beq \label{locimpl}
\max_{b\in \La_{\frac L 3}(a)} W\up{a}_{\bom,\lambda}(b)>\e^{-\frac 1 4 m h_{ I^\ell} (\lambda) L} \;  \Longrightarrow \; \max_{y\in A_L(a)} W\up{a}_{\bom,\lambda}(y)\le \e^{-\frac 7 {132}m  h_{ I^\ell}  (\lambda) \norm{y-a}},
\eeq
 \beq\label{Aell}
 \text{where} \quad A_L(a):=  \set{y\in \Z^d; \ \tfrac 8 7 L \le   \norm{y-a}\le \tfrac {33}{14} L}.
  \eeq
In particular, for all $\bom \in \cY_{L,a}$ and   $\lambda \in I$ we have
\beq  \label{WW}
W\up{a}_{\bom,\lambda}(a)W\up{a}_{\bom,\lambda}(y)\le 
 \e^{- \frac 7 {132}  m h_{ I^\ell}  (\lambda)  \norm{y-a}}\mqtx{for all} y\in A_L(a).
\eeq

   \end{enumerate}
   \end{theorem}
   
   Theorem~\ref{thmloc}  implies Anderson localization and dynamical localization, and more, as shown in  \cite{GKsudec,GKber,EK}.  In particular, we get the following corollary.

\begin{corollary} \label{corloc}  Suppose the conclusions of  Theorem~\ref{thmMSA} hold  for an  Anderson model $H_{\bom}$, and  let $I=I_\infty$, $m=m_\infty$. Then the following holds with probability one:
  \begin{enumerate}
 \item  {$H_{\bom}$}  has  pure point spectrum in the interval $I$.

\item   If   {$\psi_\lambda$}  is an {eigenfunction} of $H_{\bom}$
with eigenvalue  {$\lambda\in I$}, then $\psi_\lambda$ is exponentially localized  with rate of decay $ \frac 7 {132} m  h_{I}(\lambda)$,   more precisely,
\begin{equation}\label{expdecay222}
\abs{\psi_\lambda(x)} \le C_{\bom,\lambda}\norm{T_0^{-1} \psi}\, e^{-  \frac 7 {132}m  h_{I}(\lambda)\norm{x}} \qquad \text{for all}\quad  x \in \R^{d}.
\end{equation}

 \item If  $\lambda \in I$, then for all   $x,y \in \Z^d$ we have 
 \begin{align}\label{eqWW}
W\up{x}_{\bom,\lambda}(x)W\up{x}_{\bom,\lambda}(y)
\le   C_{m,\bom,\nu} \pa{h_{I} (\lambda)}^{-\nu}\e^{(\frac 4 {33} +\nu)m  h_{I} (\lambda)   (2d\log \scal{x})^{\frac 1 \xi}}  \e^{- \frac 7 {132} m  h_{I} (\lambda) \norm{y-x}} .
\end{align}

\item  If  $\lambda \in I$, then for  $ \psi\in \Chi_{\set{\lambda}}(H_{\bom})$ and all $x,y \in \Z^d$ we have
\begin{align}\label{eqWW2}
 &\abs{\psi(x)}\abs{\psi(y)}\\ \notag
& \;\;
\le  C_{m,\bom,\nu} \pa{h_{I} (\lambda)}^{-\nu}\, \norm{T_x^{-1} \psi}^2\e^{(\frac 4 {33} +\nu)m  h_{I} (\lambda)   (2d\log \scal{x})^{\frac 1 \xi}}  \e^{- \frac 7 {132} m  h_{I} (\lambda) \norm{y-x}}  \\ \notag
&\;  \;\le 2^\nu C_{m,\bom,\nu} \pa{h_{I} (\lambda)}^{-\nu}\, \norm{T_0^{-1} \psi}^2\la x\ra^{2\nu}\e^{(\frac 4 {33} +\nu)m  h_{I} (\lambda)   (2d\log \scal{x})^{\frac 1 \xi}}  \e^{- \frac 7 {132} m  h_{I} (\lambda) \norm{y-x}} .
 \end{align}

\item If  $\lambda \in I$, then there exists $x_\lambda=x_{\bom,\lambda} \in \Z^d$, such that for $ \psi\in \Chi_{\set{\lambda}}(H_{\bom})$ and  all $x \in \Z^d$ we have
\begin{align}\nn
&\abs{\psi(x)} 
 \le C_{m,\bom,\nu} \pa{h_{I} (\lambda)}^{-\nu}\norm{T_{x_\lambda}^{-1} \psi} \e^{(\frac 4 {33} +\nu)m  h_{I} (\lambda)   (2d\log \scal{x_\lambda})^{\frac 1 \xi}}  \e^{- \frac 7 {132} m  h_{I} (\lambda) \norm{x-x_\lambda}} \\ 
  & \   \le  
 2^{\frac \nu 2} C_{m,\bom,\nu} \pa{h_{I} (\lambda)}^{-\nu}\norm{T_0^{-1} \psi}\la x_\lambda\ra^{\nu}\e^{(\frac 4 {33} +\nu)m  h_{I} (\lambda)   (2d\log \scal{x_\lambda})^{\frac 1 \xi}}  \e^{- \frac 7 {132} m  h_{I} (\lambda) \norm{x-x_\lambda}} .
 \end{align}

\end{enumerate}
\end{corollary}

 In Corollary~\ref{corloc}, (i) and (ii)  are statements of Anderson localization, (iii) and (iv) are statements of dynamical localization ((iv) is called   SUDEC (summable uniform decay of eigenfunction correlations) in \cite{GKsudec}), and (v) is SULE (semi-uniformly localized eigenfunctions; see
\cite{DRJLS0,DRJLS}).

 We can also derive statements of localization in expectation, as in \cite{GKsudec,GKber}.

\section{Localization at the bottom of the spectrum}\label{secbottom}

We now discuss how to obtain the initial step for the eigensystem multiscale analysis at the bottom of the spectrum and prove localization.  Let  $H_{\bom}$ be an Anderson model, and  set $E_0= \inf \Sigma$ (see \eq{Sigma}), the bottom of the almost sure spectrum of $H_{\bom}$.  We will consider intervals at the bottom of the spectrum, more precisely,  intervals of the form $J=[E_0 , E_0 + A)$ with $A>0$.  We set $\tilde J= (E_0 -A, E_0 + A)$,
so $ J \cap \Sigma = \tilde J \cap \Sigma$,  call a box $(m,J)$-localizing if it is $(m,\tilde J)$-localizing as in Definition~\ref{defmIloc}, etc. 
We also set $J_L$ and $J^L$ so $\widetilde {J_L}= \tilde {J}_L$ and $\widetilde {J^L}= \tilde {J}^L$.

\subsection{Fixed disorder}

\begin{proposition}\label{PropLT} Let  $H_{\bom}$ be an Anderson model, and
set $E_0= \inf \Sigma$.
There exists a constant $C_{d,\mu}>0$ such that, given $\zeta \in (0,1)$, for sufficiently large $L$ we have 
\begin{align}\label{LT}
¥\inf_{x\in \R^d} \P\set{H_{\La_L(x)}  >E_0 + C_{d,\mu} L^{- \frac {2\zeta}d}   }\ge 1 - \e^{-L^{\zeta}}.
\end{align}¥
In particular,     for all   intervals $J_\zeta(L) = [E_0 , E_0 + C_{d,\mu} L^{- \frac {2\zeta}d} )$  and all $m>0$ we have
 \begin{align}\label{LT2}
\inf_{x\in \R^d} \P\set{\La_{L} (x) \sqtx{is}   (m,J_\zeta(L)) \text{-localizing for} \; H_{\bom}} \ge 1 -  \e^{-L^\zeta}.
\end{align}
\end{proposition}

The estimate \eq{LT}  follows from a Lifshitz tails estimate.  It can be derived from  \cite[Proof of Theorem 11.4]{Ki}.  Although the boxes in \cite{Ki} are all centered at points in $\Z^d$, the arguments, including the crucial \cite[Lemma~6.4]{Ki}, can be extended to boxes  centered at points in $\R^d$. Note that \eq{LT2} follows trivially from \eq{LT}.  Since the probability distribution $\mu$ is a continuous measure (see \eq{Holdercont}),   it follows from  \eq{Sigma} that $J_\zeta(L) \subset \Sigma$ for all sufficiently  large $L$.

We will now  combine Proposition~\ref{PropLT}  with Theorem~\ref{thmMSA}, taking  $I_0= \widetilde {J_\zeta(L_0) }$, i.e., $E=E_0$ and $A_0= C_{d,\mu} L_0^{- \frac {2\zeta}d}$ in Theorem~\ref{thmMSA}.   To satisfy \eq{upbm25552} for  $L$ large, we  take $m_0= \frac 1 {9d} C_{d,\mu} L^{- \frac {2\zeta}d}$, $m_-=\frac 1 {9d} C_{d,\mu}$ and $\kappa^\pr = \frac {2\zeta}d$.
To satisfy \eq{gamtzetabeta2}   we require $\frac {2\zeta}d < \tau - \gamma \beta$, and then choose  $0<\kappa <  \tau - \gamma \beta -\kappa^\pr$. Since for a fixed $\zeta$ we can take $\tau$ and $\gamma$  close to $1$ and $\beta$  close to $\zeta$, respecting \eq{ttauzeta0}, we find we can choose the parameters in  \eq{ttauzeta0}  as long as
\beq
\tfrac {2\zeta}d < 1- \zeta  \quad \iff \quad \zeta < \tfrac  d {d+2}.
\eeq
We obtain the following theorem.

\begin{theorem}\label{themlocbottom} Let  $H_{\bom}$ be an Anderson model,  and fix $ 0<\xi <\zeta <\frac  d {d+2}  $.    Then there exists $\gamma >1$ such that, if $L_0$ is sufficiently large, for all $ L\ge L_0^\gamma$ we have
  \begin{align} \label{MSALnok2234}
\inf_{x\in \R^d} \P\set{\La_{L} (x) \sqtx{is} ( m_{\zeta,\infty} , J_{\zeta,\infty}, J_{\zeta,\infty}^{L^\frac 1 \gamma})   \text{-localizing for} \; H_{\bom}} \ge 1 -  \e^{-L^\xi},
\end{align}
where 
 \begin{align}\label{Aminfty23}
 A_{\zeta,\infty}&=A_{\zeta,\infty}(L_0)= C_{d,\mu} L_0^{- \frac {2\zeta}d} \prod_{k=0}^\infty \pa{1- L_0^{-\kappa\gamma^k}} \ge  \tfrac 1 2 C_{d,\mu} L_0^{- \frac {2\zeta}d},  \\ \nn  J_{\zeta,\infty}&= [E_0, E_0+ A_{\zeta,\infty})\supset [E_0, E_0+ \tfrac 1 2 C_{d,\mu} L_0^{- \frac {2\zeta}d}), \\ \nn
  m_{\zeta,\infty}&=m_{\zeta,\infty} (L_0)=  \tfrac 1 {9d} C_{d,\mu} L^{- \frac {2\zeta}d} \prod_{k=0}^\infty \pa{1- C_{d,\frac 1 {9d} C_{d,\mu}} L_0^{-\vrho\gamma^k}} \ge  \tfrac 1 {18 d} C_{d,\mu} L^{- \frac {2\zeta}d} .
  \end{align}
  In particular, the conclusions of  Theorem~\ref{thmloc}  and Corollary~\ref{corloc} hold in the interval $J_{\zeta,\infty}$.
\end{theorem}
 
 \subsection{Fixed interval}
We may also use disorder to start the eigensystem multiscale analysis in a fixed interval at the bottom of the the spectrum. To do so we introduce a disorder parameter $g>0$, and
set $H_{g,\bom}=  -\Delta + gV_{\bom} $.  We assume
$\set{0}\in \supp \mu \subset [0,\infty)$, so it follows from \eq{Sigma} that $E_0= -2d$.
Then, given $B>0$ and $\zeta \in (0,1)$,
\begin{align}\label{HD}
¥&\inf_{x\in \R^d} \P\set{H_{g,\La_L(x)}  \ge -2d + B  }\ge \inf_{x\in \R^d} \P\set{g\omega_x \ge B \sqtx{for all} x\in \La_L(x)} \\ \nn &  \qquad \ge \pa{1 - \mu([0,g^{-1}B))}^{(L+1)^d}\ge  \pa{1 - K (g^{-1}B)^{\alpha}}^{(L+1)^d}\\ \nn &  \qquad \ge 1 - (L+1)^dK (g^{-1}B)^{\alpha}\ge 1 - L^{-\zeta}
\qtx{for}  g \ge g_\zeta(L).
\end{align}
It follows that, given $\zeta \in (0,1)$, for  $ g \ge g_\zeta(L)$ and all $m>0$ we have 
\begin{align}\label{HD2}
\inf_{x\in \R^d} \P\set{\La_{L} (x) \sqtx{is}   (m,{[-2d,-2d+B)}) \text{-localizing for} \; H_{g,\bom}} \ge 1 -  \e^{-L^\zeta}.
\end{align}
Combining with  Theorem~\ref{thmMSA} we obtain the following theorem.

\begin{theorem}\label{thmfixedint} Let  $H_{g,\bom}$ be an Anderson model with disorder as above, and choose exponents as in \eq{ttauzeta0}-\eq{vsdef}. Then,
given $B>0$, let $J(B)= [-2d,-2d+B)$ and  pick $ 0<m \le \frac 1 2\log (1 +\frac B {4d})$.   Then, if 
$L_0$ is sufficiently large, for all $ L\ge L_0^\gamma$ and $g \ge g_\zeta(L_0)$  we have
  \begin{align} \label{MSALnok9898}
\inf_{x\in \R^d} \P\set{\La_{L} (x) \sqtx{is} ( m_\infty , {J_\infty(B)}, \pa{J_\infty(B)}^{L^\frac 1 \gamma})  \text{-localizing for} \; H_{g,\bom}} \ge 1 -  \e^{-L^\xi},
\end{align}
where
\begin{align}\nn
 A_\infty&=A_\infty (L_0)= B \prod_{k=0}^\infty \pa{1- L_0^{-\kappa\gamma^k}}, \quad J_\infty= J_\infty(L_0)=[-2d,-2d+{A_\infty}), \\ \label{Aminfty22}
  m_\infty&=m_\infty (L_0)= m  \prod_{k=0}^\infty \pa{1- C_{d,m_-} L_0^{-\vrho\gamma^k}}.   \end{align}
In particular, the conclusions of  Theorem~\ref{thmloc}  and Corollary~\ref{corloc} hold in the interval $J_{\infty}$. Moreover,  
$\lim_{L_0\to \infty} A_\infty (L_0) = B$ and $\lim_{L_0\to \infty} m_\infty (L_0) = m $.

\end{theorem}

\section{Preamble to the eigensystem multiscale analysis}\label{secprep}

In the sections we introduce notation and prove lemmas that play an important role in the eigensystem multiscale analysis.  $H$ will always denote  a  discrete  Schr\"odinger operator
  $H=- \Delta +V$ on $\ell^2(\Z^d)$.

\subsection{Subsets, boundaries, etc.}

  Let  $\Phi \subset \Theta\subset \Z^d$.  We set the boundary, exterior boundary, and interior boundary of $\Phi$ relative to $\Theta$, respectively,  by
 \begin{align}\label{defbdry}
  \boldsymbol{ \partial}^{ \Theta} \Phi &=\set{(u,v) \in \Phi\times \pa{\Theta\setminus \Phi}; \  \abs{u-v}=1},
   \\
 \partial_{\mathrm{ex}}^{ \Theta} \Phi &=\set{v \in\pa{\Theta\setminus \Phi}; \ (u,v) \ \in  \boldsymbol{ \partial}^{ \Theta} \Phi\qtx{for some}u \in \Phi},\notag
    \\
 \partial^{ \Theta}_{\mathrm{in}}\Phi &=\set{u \in {\Phi}; \ (u,v) \ \in  \boldsymbol{ \partial}^{ \Theta} \Phi\qtx{for some}v \in \Theta\setminus \Phi}.\notag
   \end{align}
   We let  
\beq \label{Ry}
R_y^{\partial^{ \Theta}_{\mathrm{in}}\Phi} = \dist \pa{y, \partial^{ \Theta}_{\mathrm{in}}\Phi} \qtx{for} y \in \Phi.
\eeq  
  Given $t\ge 1$, we set
\begin{align}\label{defLatTh} 
 \Phi^{\Th,t}& = \set{y\in \Phi;   \; \dist \pa{y,{ \Th}\setminus \Phi }> \fl{t}}, \quad 
 {\partial}_{\mathrm{in}}^{\Th,t} \Phi   = \Phi \setminus  \Phi^{\Th,t},\\ \notag   
 {\partial}^{\Th,t} \Phi  & =  {\partial}_{\mathrm{in}}^{\Th,t} \Phi \cup \partial_{\mathrm{ex}}^{ \Th} \Phi.
 \end{align}
 If $\Th=\Z^d$ we omit it from the notation, i.e., $ \Phi^{t}= \Phi^{\Z^d,t}$.
 If $\Phi= \La_L(x)$,  we write   $\La_L^{\Th,t}(x)= \pa{\La_L(x)}^{\Th,t}$.

Consider a box   $\La_L\subset \Th \subset \Z^d$.  Given $v \in \Th$, we let   $\hat{v} \in   \partial_{\mathrm{in}}^\Th { \La_L} $ be the unique $u \in   \partial_{\mathrm{in}}^\Th { \La_L} $
 such that $(u,v)\in \boldsymbol{\partial}^\Th { \La_L} $ if $v\in \partial_{\mathrm{ex}}^\Th { \La_L} $,
 and set $\hat v=0$ otherwise.
  For $ L\ge 2$  
 we have  
  \beq\label{bdryest}
\abs{\partial_{\mathrm{in}}^{ \Th} \La_L }\le\abs{\partial_{\mathrm{ex}}^{ \Th} \La_L } =\abs{ \boldsymbol{ \partial}^{ \Th} \La_L }\le s_{d} L^{d-1}, \qtx{where} s_d= 2^{d} d.
\eeq

If $\Phi \subset \Theta\subset \Z^d$, \  
$
H_{ \Theta}= H_{ \Phi}\oplus H_{ \Theta\setminus  \Phi} + \Gamma_{ \boldsymbol{ \partial}^{ \Theta}  \Phi}$
on $\ell^2( \Theta)=\ell^2( \Phi)\oplus \ell^2( \Theta\setminus \Phi)$, 
\beq\label{Hdecomp1}
 \text{where}\quad \Gamma_{   \boldsymbol{ \partial}^{ \Theta}  \Phi}(u,v)=
\begin{cases}
-1 & \text{if either}\  (u,v) \sqtx{or}(v,u) \in   \boldsymbol{ \partial}^{ \Theta}  \Phi\\
\ \ 0 & \text{otherwise}
\end{cases}.
\eeq

\subsection{Lemmas for energy intervals}

\begin{lemma}\label{lementirefn}   Given $t>0$ and $\lambda \in \R$, let  $F_{t,\lambda}(z)$ be the entire function given by
  \beq\label{defanf}
  F_{t,\lambda} (z) =  \frac{1-\e^{-t(z^2-\lambda^2)}}{z-\lambda} \qtx{for} z \in \C\setminus \set{\lambda} \qtx{and} F_{t,\lambda}(\lambda)=2t\lambda.
  \eeq
Then, given  $\Phi \subset \Z^d$,   for all $x,y \in \Phi$ we have 
 \begin{align}\label{entkernel}
\abs{\scal{\delta_x,F_{t,\lambda}(H_\Phi)\delta_y}} \le  \inf_{\eta >0} \tfrac{70}{\sqrt{\eta^2+\lambda^2}} \e^{t(\eta^2+\lambda^2)}  \e^{-\pa{\log \pa{1 + \frac \eta {4d}}}\abs{x-y}} .
\end{align}
In particular, if $\lambda \in I=(E-A,E+A)$,  where  $A>0$ and   $E\in \R$,   and
\beq\label{mag}
0< m \le \tfrac 1 2 \log \pa{1 + \tfrac A {4d}},
\eeq
it follows that  for all  $x,y \in \Phi,\; x\ne y$, we  have 
\begin{align}
\label{FHestxy}
\abs{\scal{\delta_x,F_{\frac{m \abs{x-y}}{A^2},\lambda-E}(H_\Phi-E)\delta_y}}
\le 70 {A}^{-1} \e^{- m h_{I}(\lambda) \abs{x-y}}.
\end{align}
\end{lemma}

\begin{proof}  Given $t>0$ and $\lambda \in \R$,
the function $F_{t,\lambda}(z)$  defined in \eq{defanf} is clearly
 an entire function.  Moreover, given $\eta>0$, if $\abs{\Ima z} \le \eta$ and $c>0$ we have,  ,
\begin{align}
¥\abs{ F_{t,\lambda} (z)} 	\le  \begin{cases}  \frac{\e^{t(\eta^2+\lambda^2)}+1}{c\sqrt{\eta^2 + \lambda^2}}\le   \frac{2\e^{t(\eta^2+\lambda^2)}}{c\sqrt{\eta^2 + \lambda^2}}  & \text{if}  \;\abs{z- \lambda} \ge c\sqrt{\eta^2 + \lambda^2} \\
 \frac{ (c+2)\sqrt{\eta^2 + \lambda^2}\pa{\e^{t(\eta^2+\lambda^2)}-1}}{\eta^2+\lambda^2}\le \frac{ (c+2)\e^{t(\eta^2+\lambda^2)}}{\sqrt{\eta^2 + \lambda^2}} & \text{if}\;   \abs{z- \lambda} <c\sqrt{\eta^2 + \lambda^2} 
\end{cases},
\end{align}¥ 
 so we conclude that, taking $c=\sqrt{3}-1$,
 \begin{align}\label{Feta}
 ¥ F_{t,\lambda,\eta}=\sup_{\abs{\Ima z} \le \eta} \abs{ F_{t,\lambda} (z)} \le (\sqrt{3}+1) \frac{\e^{t(\eta^2+\lambda^2)}}{\sqrt{\eta^2 + \lambda^2}}.
  \end{align}¥
  
Given  $\Phi \subset \Z^d$,  it follows from  \cite[Theorem 3]{AG} (note that it applies also for $H_\Phi$ on $\ell^2(\Phi)$), that  for all  $x,y\in\Phi $ we have
  \begin{align}
\abs{\scal{\delta_x,F_{t,\lambda}(H_\Phi)\delta_y}}&\le 18\sqrt 2 F_{t,\lambda,\eta} \e^{-\pa{\log \pa{1 + \frac \eta {4d}}}\abs{x-y}}\\ \notag & \le   \tfrac{70}{\sqrt{\eta^2+\lambda^2}} \e^{t(\eta^2+\lambda^2)}  \e^{-\pa{\log \pa{1 + \frac \eta {4d}}}\abs{x-y}} \mqtx{for all} \eta >0.  
\end{align}

To prove \eq{FHestxy},  we take $E=0$ by replacing the potential $V$ by $V-E$, and note that \eq{entkernel} holds for any discrete Schr\"odinger operator $H$.   Now let $\lambda \in I=(-A,A)$, where $A>0$,  and $m$ as in \eq{mag}, and  fix $x,y \in \Phi$, $x\ne y$.
 Since
\begin{align}
 ¥\log \pa{1 + \tfrac \eta {4d}} -\tfrac  {m}{{A}^2}\pa{\eta^2+ \lambda^2}&= \log \pa{1 + \tfrac \eta {4d}} - {m}\pa{\tfrac{\eta^2}{{A}^2}+1}+ m h_{I}(\lambda),
  \end{align}¥
   choosing $\eta=A$, and using \eq{mag}, we obtain
   \begin{align}
 ¥\log \pa{1 + \tfrac A {4d}} -\tfrac  {m}{{A}^2}\pa{A^2+ \lambda^2}&= \log \pa{1 + \tfrac A{4d}} - 2{m}+ m h_{IA}(\lambda)\ge  m h_{I}(\lambda),
  \end{align}¥
so   \eq{FHestxy} follows from \eq{entkernel} by taking  $t=\frac{m \abs{x-y}}{A^2}$ and $\eta=A$.
\end{proof}

  \begin{lemma}\label{lemkey}   Let  $\Theta\subset \Z^d$, and  let $\psi\colon{\Theta} \to \C$ be a generalized eigenfunction for $H_{\Theta}$ with generalized eigenvalue $\lambda \in \R$.  Let $\Phi \subset \Theta$ be a finite set such that $\lambda \notin \sigma(H_\Phi)$. 
  Let $A>0$,  $E\in \R$, $I=(E-A,E+A)$.  The following holds for all $y\in \Phi$:
  \begin{enumerate}
\item

 For all $t>0$ we have
 \begin{align}\label{eq:2terms}
  {\psi(y)} &= \scal { \e^{-t \pa{\pa{H_{\Phi}-E}^2-(\lambda-E)^2 }}\delta_y,\psi}   -  \scal {F_{t,\lambda-E} (H_\Phi-E)\delta_y,\Gamma_{ \boldsymbol{ \partial}^{ \Theta}  \Phi} \psi} ,
  \end{align}
 where $\Gamma_{ \boldsymbol{ \partial}^{ \Theta}  \Phi}$ is defined in \eq{Hdecomp1} and
  $F_{t,\lambda}(z)$ is the function defined in  \eq{defanf}.  
 
\item  Let  $0<R \le R_y^{\partial^{ \Theta}_{\mathrm{in}}\Phi}$ and $m$ as in \eq{mag}. For  $\lambda \in I$ it follows that
\begin{align}
\label{FHestRy}
¥\abs{\scal {F_{\frac{m R}{A^2},\lambda-E} (H_\Phi-E)\delta_y,\Gamma_{ \boldsymbol{ \partial}^{ \Theta}  \Phi} \psi} }
\le 70 \abs{\boldsymbol{ \partial}^{ \Theta} \Phi }{A}^{-1} \e^{- m h_{I}(\lambda) R}\abs{\psi(v)},
\end{align}¥
for some $v \in \partial_{\mathrm{ex}}^{ \Theta} \Phi$.

\end{enumerate}
  \end{lemma}
  
  \begin{proof}We take $E=0$ by replacing the potential $V$ by $V-E$.
   By hypothesis we have $\lambda \notin \sigma(H_\Phi)$ and
 \beq
 \scal{(H_\Th -\lambda)\vphi, \psi}=0 \qtx{for all} \vphi \in \ell^2(\Phi),
  \eeq
  so
   \beq \label{pointeig33459} 
 \scal{(H_\Phi - H_\Theta) (H_\Phi -\lambda)^{-1}\vphi, \psi}=\scal{\vphi, \psi} \qtx{for all} \vphi \in \ell^2(\Phi).
  \eeq
  It follows that for all $y \in \Phi$ and $t>0$  we have
  \begin{align}
  ¥{\psi(y)} &= \scal {\delta_y,\psi}\\
\notag &=\scal { \e^{-t \pa{H_{\Phi}^2-\lambda^2 }}\delta_y,\psi}  +\scal {(H_\Phi - H_\Theta) (H_\Phi -\lambda)^{-1}\pa{1-  \e^{-t \pa{H_{\Phi}^2-\lambda^2}} }\delta_y,\psi}  \\ 
\notag &=\scal { \e^{-t \pa{H_{\Phi}^2-\lambda^2 }}\delta_y,\psi}   -  \scal {F_{t,\lambda} (H_\Phi)\delta_y,\Gamma_{ \boldsymbol{ \partial}^{ \Theta}  \Phi} \psi} ,
  \end{align}¥
  where $\Gamma_{ \boldsymbol{ \partial}^{ \Theta}  \Phi}$ is defined in \eq{Hdecomp1} and
  the function $F_{t,\lambda}(z)$ is defined in \eq{defanf}.

Let $0<R \le R_y^{\partial^{ \Theta}_{\mathrm{in}}\Phi}$,      $m$ as in \eq{mag}, and assume   $\lambda \in I=(-A,A)$, $A>0$.   Recalling \eq{Hdecomp1},  \eq{FHestRy} follows from \eq{FHestxy}.
    \end{proof}

\begin{lemma} \label{lemkey2}
 Let $\Phi \subset \Z^d$,  $I=(E-A,E+A)$,  where  $A>0$ and   $E\in \R$, and $\lambda \in I$. Then for all  $t>0$ we have
\begin{align}
\norm{ \e^{-t \pa{\pa{H_{\Phi}-E}^2-(\lambda-E)^2 }} \Chi_{\R \setminus I}(H_\Phi)} \le \e^{-t A^2 h_{I}(\lambda)}.
\end{align}
\end{lemma}

\begin{proof}We have
\begin{align}
¥\norm{ \e^{-t \pa{\pa{H_{\Phi}-E}^2-(\lambda-E)^2 }} \Chi_{\R \setminus I}(H_\Phi)} \le \e^{-t \pa{A^2 -(\lambda-E)^2}}= \e^{-t A^2 h_{I}(\lambda)}.
\end{align}¥
\end{proof}

\subsection{Lemmas for the multiscale analysis}

 Let  $I=(E-{A},E+{A})$ with $E\in \R$ and $A>0$, and  fix a constant $ m_->0$.  When we state that a box  $\La_\ell$ is   $(m,I)$-localizing  we always assume
\beq\label{upbm}
0<m_- \ell^{-\kappa^\pr} \le m   \le  \tfrac 1 2 \log \pa{1 + \tfrac {A}{4d}}.
\eeq

We also introduce the following notation:
\begin{itemize}
\item  Given $\Theta \subset \Z^d$ and $J \subset \R$, we set  $\sigma_J(H_\Theta)= \sigma(H_\Theta) \cap J$.

\item Let $\La_\ell \subset \Theta \subset \Z^d$ be an $(m,I)$-localizing box with an  $(m,I)$-localized eigensystem $\set{(\vphi_\nu, \nu)}_{\nu \in \sigma(H_{\La_\ell})}$, and let $t>0$.  Then, for $J\subset I$ we set
\beq
\sigma_{J}^{\Th,{t}}(H_{\La_\ell})= \set{\nu\in \sigma_{J} (H_{\La_\ell}); \; x_\nu \in    \La_\ell^{\Th,{t}}} .
\eeq

\end{itemize}

 The following lemmas plays an important role in our multiscale analysis. In particular,  the role of the modulating function $h_I$  becomes  transparent in the proof of Lemma~\ref{lemdecay2}.

\subsubsection{Localizing boxes}

  \begin{lemma}\label{lemdecay2}   Let $\psi\colon{\Theta}\subset \Z^d  \to \C$ be a generalized eigenfunction for $H_{\Theta}$ with generalized eigenvalue   $\lambda \in I_\ell$.
Consider a box  $ \La_\ell \subset{\Theta}$  such that $\La_\ell$ is   $(m,I)$-localizing   with   an $(m,I)$-localized eigensystem $\set{\vphi_\nu,\nu}_{\nu \in \sigma{(H_{\La_\ell})}}$.  Suppose
\beq\label{distpointeig}
\abs{\lambda - \nu}
\ge \tfrac 1 2\e^{-L^\beta}\qtx{for all} \nu \in \sigma_{I}^{\Th,\ell_\tau}(H_{\La_\ell}). 
\eeq  
 Then for  $\ell$ sufficiently large we have:
 
\begin{enumerate}
\item If  $y\in \La_\ell^{{\Theta},2\ell_\tau}$ we have
\beq\label{decayest00}
\abs{\psi(y)}\le \e^{- m_2 h_{ I}\pa{\lambda} \ell_\tau}  \abs{\psi(v)}\qtx{for some} v \in {\partial}^{\Th, 2\ell_\tau }  \Lambda_{\ell},
\eeq
where  
\beq\label{m414}
m_2=m_2(\ell) \ge  m\pa{1 - C_{d,m_-}\ell^{-(\tau - \gamma \beta -\kappa-\kappa^\pr)}}.
\eeq

\item If $y\in \La_\ell^{\Theta,\ell_{\ttau}}$,  we have
\beq\label{decayest12}
\abs{\psi(y)}\le  \e^{-m_3 h_{I}\pa{\lambda} {R_y^{\partial^{ \Theta}_{\mathrm{in}}\La_\ell}}}  \abs{\psi(v)}\qtx{for some} v \in {\partial}^{\Th, 2\ell_\tau }  \Lambda_{\ell},
\eeq
where
\beq\label{m4} 
m_3= m_3(\ell)\ge  m \pa{1 - C_{d,m_-}\ell^{-(\frac{1- \tau}2)}}.  
\eeq
\end{enumerate}
 
  \end{lemma}
 
Lemma~\ref{lemdecay2}  resembles \cite[Lemma~3.5]{EK}, but there are important differences.  
 The box  $ \La_\ell \subset{\Theta}$   is   $(m,I)$-localizing, and hence we only have decay for eigenfunctions with eigenvalues in $I$.  Thus we can only use \eq{distpointeig}   for $\nu \in \sigma_{I}^{\Th,\ell_\tau}(H_{\La_\ell})$. To compensate, we take $\lambda \in I_\ell$,
 and use  Lemmas~\ref{lemkey} and \ref{lemkey2}.

 \begin{proof}[Proof of  Lemma~\ref{lemdecay2}]

We take $E=0$ by replacing the potential $V$ by $V-E$.  Given $y \in \La$, we write $\psi(y)$ as in \eq{eq:2terms}.

Setting  $P_{I}= \Chi_{I}\pa{H_ {\Lambda_{\ell}}}$ and $\bar P_{I}= 1-P_{I}$, 
we have
\begin{align}\label{Jdecomp}
\scal {\e^{-t \pa{H_{\La}^2-\lambda^2}}\delta_y,\psi} =   \scal {\e^{-t \pa{H_{\La}^2-\lambda^2}}P_{I}\delta_y,\psi} +  \scal {\e^{-t \pa{H_{\La}^2-\lambda^2}}\bar P_{I}\delta_y,\psi}.
\end{align}
It follows from Lemma~\ref{lemkey2} that
\begin{align}\label{barJest}
&\abs{\scal {\e^{-t \pa{H_{\La}^2-\lambda^2}}\bar P_{I}\delta_y,\psi}}\le \norm{\Chi_\La \psi}\norm{\e^{-t \pa{H_{\La}^2-\lambda^2}}\bar P_{I}}\le (\ell+1)^{\frac d 2} \e^{-t {A}^2 h_{I}(\lambda) }\abs{\psi(v)},
\end{align}
for some $v \in \La$.

We have
 \begin{align}
\scal {\e^{-t \pa{H_{\La}^2-\lambda^2}}\ P_{I}\delta_y,\psi}&= \sum_{\mu \in \sigma_{I}(\H_\La)}\e^{-t (\mu^2-\lambda^2)}{\vphi_\mu(y)} \scal{ \vphi_\mu,\psi}.
\end{align}

Let $y \in \La_\ell^{{\Theta},2\ell_\tau}$.  For $\mu \in \sigma_{I}(H_\La)$ we have, as shown in \cite[Eqs. (3.37) and (3.39)]{EK}, 
 \beq\label{337}
 \abs{{\vphi_\mu(y)}\scal{ \vphi_\mu,\psi}}\le
 2\ell^{\frac {3d}2} \e^{L^\beta} \e^{-m h_I (\mu){\ell_\tau}}\abs{\psi(v_1)}\mqtx{for some} v_1 \in \La_\ell \cup\partial_{\mathrm{ex}}^{{\Theta}}  \Lambda_{\ell}.
\eeq
It follows that
\begin{align}
¥\e^{-t (\mu^2-\lambda^2)}\abs{{\vphi_\mu(y)}\scal{ \vphi_\mu,\psi}}\le  2\ell^{\frac {3d}2} \e^{L^\beta}  \e^{-t (\mu^2-\lambda^2)}\e^{-m h_I (\mu){\ell_\tau}}\abs{\psi(v_1)}.
\end{align}¥

We now take 
\beq
t= \tfrac{m \ell_\tau}{{A}^2}  \quad \Longrightarrow \quad \e^{-t (\mu^2-\lambda^2)}\e^{-m h_I (\mu){\ell_\tau}}= \e^{-m h_I (\lambda){\ell_\tau}}\qtx{for} \mu \in I,
\eeq
obtaining
\begin{align}\label{Jest}
 \abs{\scal {\e^{-\tfrac{m \ell_\tau}{{A}^2} \pa{H_{\La}^2-\lambda^2}} P_{I}\delta_y,\psi}}&\le   2\ell^{\frac {3d}2} (\ell+1)^d  \e^{L^\beta} \e^{-m h_I (\lambda){\ell_\tau}}\abs{\psi(v)}\\
 \nn & \le  \e^{2L^\beta} \e^{-m h_I (\lambda){\ell_\tau}}\abs{\psi(v)},
\end{align}¥
for some  $v \in \La_\ell \cup\partial_{\mathrm{ex}}^{{\Theta}}  \Lambda_{\ell}$.
Combining \eq{Jdecomp}, \eq{barJest} and \eq{Jest} yields
\begin{align}\label{combJ}
¥\abs{\scal {\e^{-\tfrac{m \ell_\tau}{{A}^2}\pa{H_{\La}^2-\lambda^2}} \delta_y,\psi}}\le 2 \e^{2L^\beta}\e^{-m h_{I}(\lambda) \ell_\tau}\abs{\psi(v)},
\end{align}¥
for some  $v \in \La_\ell \cup\partial_{\mathrm{ex}}^{{\Theta}}  \Lambda_{\ell}$.

Using \eq{FHestRy}, noting  $y \in \La_\ell^{{\Theta},2\ell_\tau}$ implies  ${R_y^{\partial^{ \Theta}_{\mathrm{in}}\La_\ell}} \ge 2\ell_\tau -1>\ell_\tau$, we get 
   \begin{align}\label{eq:anpar36}
¥\abs{\scal {F_{\frac {m  \ell_\tau}{{A}^2},\lambda} (H_\La)\delta_y,\Gamma_{ \boldsymbol{ \partial}^{ \Theta}  \La} \psi} }\le 70 s_d \ell^{d-1}A^{-1} \e^{-m h_{I}(\lambda)\ell_\tau }\abs{\psi(v)}, \end{align}¥
for some  $v \in \partial^{ \Theta}_{\mathrm{ex}}\La$.

 Combining \eq{combJ} and \eq{eq:anpar36}, and using \eq{upbm}, we conclude that \begin{align}
\abs{\psi(y)}\le C_{d,m_-}\ell^{\kappa^\pr} \e^{2L^\beta} e^{-m h_{I}(\lambda) \ell_\tau}  \abs{\psi(v)}\le      \e^{-m_2 h_{ I}(\lambda)  \ell_\tau}  \abs{\psi(v)},
\end{align}¥
for some $v \in \La_\ell \cup\partial_{\mathrm{ex}}^{{\Theta}}  \Lambda_{\ell}$
where, using $\lambda \in I_{\ell}$,
\beq
m_2\ge m\pa{1 - C_{d,m_-}\ell^{-\pa{\tau - \gamma \beta -\kappa-\kappa^\pr}}}.
\eeq  
By repeating the argument as many times a necessary we can get  $v \in  {\partial}^{\Th, 2\ell_\tau }  \Lambda_{\ell}$.  This proves  part (i).

To prove part (ii), let  $y\in \La_\ell^{\Theta,\ell_{\ttau}}$, so ${R_y^{\partial^{ \Theta}_{\mathrm{in}}\La_\ell}}\ge \ell_{\ttau}$. We proceed as before, but replace \eq{337} by the following estimate. For $\mu \in \sigma_{I}^{\Th,{t}}(H_{\La_\ell})$  and $v^\pr \in \partialin^{{\Theta}}  \Lambda_{\ell}$, we have, as in \cite[Eq. (3.41)]{EK}, 
\beq\label{mipr}
\abs{{\vphi_\mu(y)}\vphi_\mu(v^\pr)}\le  \e^{-m_1^\pr   h_I(\mu) {R_y^{\partial^{ \Theta}_{\mathrm{in}}\La_\ell}}} \qtx{with}  m^\pr_1\ge   m(1-  2  \ell^{\frac {\tau -1}2}),
\eeq
so, as in \cite[Eq. (3.44)]{EK},
\begin{align}
\abs{{\vphi_\mu(y)}\scal{ \vphi_\mu,\psi}}\le 2 \e^{L^\beta} s_d\ell^{d-1}  \e^{-m_1^\pr   h_I(\mu) {R_y^\Theta}}\abs{\psi(v_1)}\le \e^{2L^\beta}   \e^{-m_1^\pr   h_I(\mu) {R_y^{\partial^{ \Theta}_{\mathrm{in}}\La_\ell}}}\abs{\psi(v_1)},
\end{align}¥
for some $v_1 \in\partialex^{{\Theta}}  \Lambda_{\ell}$.  If $\mu \in \sigma_{I} (H_{\La_\ell})$ with  $x_\mu \in  \partialin^{\Th,\ell_\tau} \La_\ell$, we have 
\beq\norm{x_\mu-y}\ge {R_y^{\partial^{ \Theta}_{\mathrm{in}}\La_\ell}} - \ell_{\tau} \ge   {R_y^{\partial^{ \Theta}_{\mathrm{in}}\La_\ell}}\pa{1- 2 \ell^{\tau- \ttau}}={R_y^{\partial^{ \Theta}_{\mathrm{in}}\La_\ell}}\pa{1- 2 \ell^{\frac {\tau-1}2}},
\eeq
 so 
\begin{align}
&\abs{{\vphi_\nu(y)} \scal{ \vphi_\nu,\psi}} \le \e^{-m   h_I(\mu) \norm{x_\mu-y}}\norm{\Chi_\La \psi} \\ \notag & \quad
\le \e^{-m   h_I(\mu) {R_y^{\partial^{ \Theta}_{\mathrm{in}}\La_\ell}}\pa{1- 2 \ell^{\frac {\tau-1}2}}} (\ell+1)^{\frac d 2} \abs{\psi(v_2)}
\le  (\ell+1)^{\frac d 2} \e^{-m_1^\pr   h_I(\mu) {{R_y^{\partial^{ \Theta}_{\mathrm{in}}\La_\ell}}}}  \abs{\psi(v_2)} ,
\end{align}
for some  $v_2 \in \La$, where $m_1^\pr$ is given in \eq{mipr}.
 It follows that for all $\mu \in \sigma_{I}(H_\La)$ we have
\begin{align}
¥\e^{-t (\mu^2-\lambda^2)}\abs{{\vphi_\mu(y)}\scal{ \vphi_\mu,\psi}}\le  \e^{2L^\beta} 
\e^{-t (\mu^2-\lambda^2)}  \e^{-m_1^\pr   h_I(\mu) {R_y^{\partial^{ \Theta}_{\mathrm{in}}\La_\ell}}}\abs{\psi(v)} ,
\end{align}¥
for some $ v\in \La\cup \in\partialex^{{\Theta}}  \Lambda$.

We now take  
\beq
t= \tfrac{m_1^\pr {R_y^{\partial^{ \Theta}_{\mathrm{in}}\La_\ell}}}{{A}^2}  \quad \Longrightarrow \quad \e^{-t (\mu^2-\lambda^2)}\e^{-m_1^\pr h_I (\mu){\ell_\tau}}= \e^{-m_1^\pr h_I (\lambda){\ell_\tau}}\qtx{for} \mu \in I,
\eeq
obtaining
\begin{align}\label{Jest32}
 \abs{\scal {\e^{-\tfrac{m_1^\pr {R_y^{\partial^{ \Theta}_{\mathrm{in}}\La_\ell}}}{{A}^2}  \pa{H_{\La}^2-\lambda^2}} P_{I}\delta_y,\psi}}&\le   (\ell+1)^d \e^{2L^\beta} \e^{-m_1^\pr h_I (\lambda){{R_y^{\partial^{ \Theta}_{\mathrm{in}}\La_\ell}}}}\abs{\psi(v)}\\ \notag & \le  \e^{3L^\beta} \e^{-m_1^\pr h_I (\lambda){{R_y^{\partial^{ \Theta}_{\mathrm{in}}\La_\ell}}}}\abs{\psi(v)},
\end{align}¥
for some  $v \in \La_\ell \cup\partial_{\mathrm{ex}}^{{\Theta}}  \Lambda_{\ell}$.
Combining \eq{Jdecomp}, \eq{barJest} and \eq{Jest32} yields
\begin{align}\label{noJest}
¥\abs{\scal {\e^{-\tfrac{m_1^\pr {R_y^{\partial^{ \Theta}_{\mathrm{in}}\La_\ell}}}{{A}^2} \pa{H_{\La}^2-\lambda^2}} \delta_y,\psi}}&\le 2 \e^{3L^\beta}\e^{-m_1^\pr  h_{I}(\lambda) {R_y^{\partial^{ \Theta}_{\mathrm{in}}\La_\ell}}}\abs{\psi(v)},
\end{align}¥
for some  $v \in \La_\ell \cup\partial_{\mathrm{ex}}^{{\Theta}}  \Lambda_{\ell}$.

Using \eq{FHestRy} (with $m=m_1^\pr$), we get
\begin{align}\label{eq:anpar99}
&\abs{\scal {F_{\frac {m_1^\pr{R_y^{\partial^{ \Theta}_{\mathrm{in}}\La_\ell}}}{{A}^2},\lambda} (H_\La)\delta_y,\Gamma_{ \boldsymbol{ \partial}^{ \Theta}  \La} \psi} }\le     70 s_d \ell^{d-1}A^{-1}\e^{-m_1^\pr  h_I(\lambda)  {R_y^{\partial^{ \Theta}_{\mathrm{in}}\La_\ell}}}\abs{\psi(v)},
\end{align}
for some  $v \in {\partial}_{\mathrm{ex}}^{\Th} \La$.   
We conclude  from \eq{noJest}  and   \eq{eq:anpar99} that
\begin{align}\label{decayest9999} 
\abs{\psi(y)}\le C_{d,m_-} \ell^{\kappa^\pr}  \e^{3L^\beta}\e^{-m_1^\pr h_{I}(\lambda) {R_y^{\partial^{ \Theta}_{\mathrm{in}}\La_\ell}}}\abs{\psi(v)}\le      \e^{-m_3 h_{ I}(\lambda)  {R_y^{\partial^{ \Theta}_{\mathrm{in}}\La_\ell}}}  \abs{\psi(v)},
\end{align}¥
for some $v \in \La_\ell \cup\partial_{\mathrm{ex}}^{{\Theta}}  \Lambda_{\ell}$
where, using $h_{ I}(\lambda)  \ge \ell^{-\kappa}$ since  $\lambda \in I_{\ell}$, we have 
\begin{align}
m_3 &\ge  m \pa{1 - C_{d,m_-}\ell^{-\min \set{\ttau -  \gamma \beta -\kappa-\kappa^\pr, \frac{1- \tau}2}}} = m \pa{1 - C_{d,m_-}\ell^{-(\frac{1- \tau}2)}}.
\end{align}

If $v \notin {\partial}^{\Th, 2\ell_\tau }  \Lambda_{\ell}$, we can apply \eq{decayest00} repeatedly until we get \eq{decayest9999} with $v \in{\partial}^{\Th, 2\ell_\tau }  \Lambda_{\ell}$.
\end{proof}

 \begin{lemma}\label{lem:ident_eigensyst} Let the finite set  $\Theta\subset\Z^d$ be $L$-level spacing for $H$, and  let  $\set{(\psi_\lambda,\lambda)}_{\lambda \in \sigma(H_\Th)}$ be an eigensystem for $H_\Th$.  
 
 Then the following holds for sufficiently large $L$:

 \begin{enumerate}

\item   Let   $\Lambda_\ell\subset \Th$  be an  $(m,I)$-localizing box with an $(m,I)$-localized eigensystem $\set{(\vphi_{\lambda}, \lambda)}_{\lambda\in \sigma(H_{\La_\ell})}$.

\begin{enumerate}
\item

There exists an injection    
\beq\label{injection}
\lambda\in  \sigma_{I_{2\ell }}^{\Th,{\ell_\tau}}(H_{\La_\ell}) \mapsto  \wtilde{\lambda}\in \sigma(H_\Th), 
\eeq
such that for all $\lambda\in    \sigma_{I_{2\ell}}^{\Th,{\ell_\tau}}(H_{\La_\ell}) $ we have
\begin{align}
\abs{ \wtilde{\lambda} -\lambda} \le \e^{-m_1h_I(\lambda){\ell_\tau}},\sqtx{with}  m_1=m_1(\ell)\ge  m\pa{1- C_{d,m_-} \tfrac {\log \ell}{\ell^{\tau-\kappa-\kappa^{\pr}}}},  \label{tildedist132}
\end{align}
and, redefining  ${\vphi_{\lambda}}$ so $\scal{\psi_{\wtilde{\lambda}},\vphi_{\lambda}}>0$,
\beq\label{difeq86}
\norm{\psi_{\wtilde{\lambda}}-{\vphi_{\lambda}}}\le 2 \e^{-m_1h_I(\lambda){\ell_\tau}}\e^{L^\beta}.
\eeq
\item Let
\beq
 \sigma_{\set{\La_\ell}}(H_\Th):= \set{ \wtilde{\lambda};\ \lambda\in  \sigma_{I_{2\ell}}^{\Th,{\ell_\tau}}(H_{\La_\ell})}.
\eeq
Then for 
 $\nu \in  \sigma_{\set{\La_\ell}}(H_\Th)$  we have
\beq \label{psidecout}
\abs{\psi_\nu(y)}\le  2 \e^{-m_1{(2\ell)^{-\kappa}}{\ell_\tau}}\e^{L^\beta} \qtx{for all} y\in \Th \setminus \La_\ell.
\eeq

\item If  $\nu \in \sigma_{I_{\ell}}(H_\Th)\setminus  \sigma_{\set{\La_\ell}}(H_\Th)$, we have
 \beq\label{distpointeiga}
 \abs{\nu - \lambda}\ge  \tfrac 12 \e^{-\L^\beta} \qtx{for all} \lambda \in  \sigma_{I}^{\Th,{\ell_\tau}}(H_{\La_\ell}),
 \eeq
and 
\beq
\abs{\psi_\nu(y)}\le \e^{-m_2 h_{I}(\nu){{\ell_\tau}}} \sqtx{for} y\in \La_\ell^{{\Theta},2{{\ell_\tau}}}, \sqtx{with} m_2=m_2(\ell) \;\text{as in \eq{m414}}.   \label{psidecgood}
\eeq
Moreover, if  $y\in \La_\ell^{{\Theta},\ell_{\ttau}}$ we have 
\beq\label{psidecgoodpr}
\abs{\psi_\nu(y)}\le \e^{-m_3 h_{I}(\nu) {R_y^{\partial^{ \Theta}_{\mathrm{in}}\La_\ell}}}\abs{\psi_\lambda(y_1)} \qtx{for some} y_1 \in{\partial}^{\Th,2\ell_{\tau }}  \Lambda_{\ell},
\eeq
 with $m_3=m_3(\ell)$ as is  in \eq{m4}.

  \end{enumerate}

\item    Let  $\set{\Lambda_\ell(a)}_{a\in \cG} $, where $\cG \subset \R^d$ and  $\Lambda_\ell(a)\subset \Th$ for all $a \in \cG$,   be a collection of   $(m,I)$-localizing boxes with  $(m,I)$-localized eigensystems    \\ $\set{(\vphi_{\lambda\up{a}}, \lambda\up{a})}_{\lambda\up{a}\in \sigma(H_{\La_\ell}(a))}$, and set  
\begin{align}\label{defcE}
 \cE_\cG^\Th(\lambda)&=  \set{{\lambda\up{a}}; \;a\in \cG,\,\lambda\up{a}\in  \sigma_{I_{2\ell}}^{\Th,{\ell_\tau}}(H_{\La_\ell}(a)), \,  \wtilde{\lambda}\up{a}=\lambda}\sqtx{for} \lambda\in \sigma(H_\Th), \\ \notag
\sigma_\cG(H_\Th)&=\set{ \lambda\in \sigma(H_\Th); \, \cE_\cG^\Th(\lambda)\ne \emptyset} =\textstyle\bigcup_{a\in \cG} \sigma_{\set{\La_\ell (a)}}(H_\Th).
\end{align}

\begin{enumerate}
\item 
Let $a,b \in \cG$, $a\ne b$,  Then, for  $\lambda\up{a}\in   \sigma_{I_{2\ell}}^{\Th,{\ell_\tau}}(H_{\La_\ell}(a))$ and $\lambda\up{b}\in   \sigma_{I_{2\ell}}^{\Th,{\ell_\tau}}(H_{\La_\ell}(b))$, 
\beq \label{vphilambda}
 {\lambda\up{a}},{\lambda\up{b}}\in \cE_\cG^\Th(\lambda) \quad \Longrightarrow \quad  \norm{x_{\lambda\up{a}}-x_{\lambda\up{b}}} < 2 {{\ell_\tau}}.
\eeq
As a consequence,
\beq\label{sigmaab}
\Lambda_\ell(a) \cap \Lambda_\ell(b)=\emptyset   \quad \Longrightarrow \quad \sigma_{\set{\La_\ell(a)}}(H_\Th)\cap \sigma_{\set{\La_\ell(b)}}(H_\Th)=\emptyset  .
\eeq
 
  \item  If  $\lambda \in\sigma_\cG(H_\Th)$, we have
  \beq \label{psidecout63}
\abs{\psi_\lambda(y)}\le  2 \e^{-m_1{(2\ell)^{-\kappa}}{\ell_\tau}}\e^{L^\beta}   \sqtx{for all} y\in\Th \setminus {\Th}_{\cG}, \sqtx{where}  {\Th}_{\cG}:= \bigcup_{a\in \cG}\La_\ell(a).
\eeq

 \item  If  $\lambda \in \sigma_{I_{\ell}}(H_\Th)\setminus\sigma_\cG(H_\Th)$, we have
\beq\label{psidecgood63}
\abs{\psi_\lambda(y)}\le\e^{-m_2 h_{I}(\lambda){{\ell_\tau}}}   \qtx{for all} y\in{\Th}_{\cG,\tau}:= \bigcup_{a\in \cG}\La_\ell^{{\Theta},2{{\ell_\tau}}}(a).
\eeq
 \end{enumerate}
\end{enumerate} 

\end{lemma}

\begin{proof}  Let   $\Lambda_\ell\subset \Th$ be be an  $(m,I)$-localizing box with an $(m,I)$-localized eigensystem $\set{(\vphi_\lambda, \lambda)}_{\lambda\in \sigma(H_{\La_\ell})}$.
   Given $\lambda\in  \sigma_{I_{2\ell}}^{\Th,{\ell_\tau}}(H_{\La_\ell})$, it follows from \cite[Eq.~(3.10) in Lemma~3.2]{EK} that 
 \beq\label{appeig}
\dist\pa{\lambda, \sigma(H_{\Theta})}\le  \sqrt{ s_d }\, \ell^{\frac {d-1}2}  \e^{-m h_I(\lambda){\ell_\tau}},
\eeq
so the existence of $\wtilde{\lambda}\in \sigma(H_\Th)$ satisfying \eq{tildedist132} follows.
 Uniqueness follows from  the fact that $\Theta$ is $L$-level spacing  and $\gamma \beta <\tau$. In addition, note that $\wtilde{\lambda} \ne \wtilde{\nu}$ if $\lambda,\nu  \in   \sigma_{I_{2\ell}}^{\Th,{\ell_\tau}}(H_{\La_\ell})$, $\lambda\ne \nu$, because in this case we have
 \beq
\abs{ {\wtilde{\lambda}} -\wtilde{\nu}}\ge  \abs{ {\lambda} -\nu}- \abs{\wtilde{\lambda} - \lambda}  -  \abs{\wtilde{\nu} - \nu}\ge \e^{-\ell^\beta} - 2\e^{-m_1 {(2\ell)^{-\kappa}}{\ell_\tau}}\ge  \tfrac 12 \e^{-\ell^\beta},
\eeq
as $\La_\ell(a)$ is level spacing for $H$,  and $\kappa +\beta < \tau$.
Moreover, it follows from \cite[Lemma~3.3]{EK} that,  after multiplying  ${\vphi_{\lambda}}$ by a phase factor if necessary to get so $\scal{\psi_{\wtilde{\lambda}},\vphi_{\lambda}}>0$, we have \eq{difeq86}.

If $\nu \in \sigma_{\set{\La_\ell}}(H_\Th)$, we have $\nu=  \wtilde{\lambda}$ for some $\lambda\in  \sigma_{I_{2\ell}}^{\Th,{\ell_\tau}}(H_{\La_\ell})$, so \eq{psidecout} follows from \eq{difeq86} as
$\vphi_{\lambda}(y)=0$ for all  $y\in \Th \setminus \La_\ell(a)$.

 Let $\nu\in  \sigma_{I_{\ell}}(H_\Th)\setminus  \sigma_{\set{\La_\ell}}(H_\Th)$.  Then for all $\lambda\in  \sigma_{I_{2\ell}}^{\Th,{\ell_\tau}}(H_{\La_\ell})$ we have
 \begin{align}
\label{distpointsigeigab}
 \abs{\nu - \lambda}&\ge  \abs{\nu - \wtilde{\lambda}} -  \abs{\wtilde{\lambda} - \lambda}\ge   \e^{-\L^\beta} -\e^{-m_1h_I(\lambda){\ell_\tau}}\\  \nn & \ge  \e^{-\L^\beta} -\e^{-m_1 {(2\ell)^{-\kappa}}{\ell_\tau}}\  \ge \tfrac 12 \e^{-\L^\beta},
 \end{align}
since $\Theta$ is $L$-level spacing for $H$, we have \eq{tildedist132}, and $\kappa + \gamma \beta <\tau$. Thus
\beq
\abs{\nu - \lambda}\ge  \tfrac 12 \e^{-\L^\beta} \qtx{for all} \lambda \in  \sigma_{I_{2\ell}}^{\Th,{\ell_\tau}}(H_{\La_\ell}).
\eeq
Since $\nu \in I_\ell$, we actually have \eq{distpointeiga}. 
Thus  \eq{psidecgood} follows from Lemma~\ref{lemdecay2}(i) and $\norm{\psi_\nu}=1$, and  \eq{psidecgoodpr} follows from Lemma~\ref{lemdecay2}(ii).

Now let 
 $\set{\Lambda_\ell(a)}_{a\in \cG} $, where $\cG \subset \R^d$ and  $\Lambda_\ell(a)\subset \Th$ for all $a \in \cG$,  be a collection of   $(m,I)$-localizing boxes with  $(m,I)$-localized eigensystems    $\set{(\vphi_{\lambda\up{a}}, \lambda\up{a})}_{\lambda\up{a}\in \sigma(H_{\La_\ell}(a))}$.  Let  $\lambda\in \sigma(H_\Th)$, $a,b \in \cG$, $a\ne b$, $\lambda\up{a}\in   \sigma_{I_{2\ell}}^{\Th,{\ell_\tau}}(H_{\La_\ell}(a))$ and $\lambda\up{b}\in   \sigma_{I_{2\ell}}^{\Th,{\ell_\tau}}(H_{\La_\ell}(b))$. Suppose ${\lambda\up{a}},{\lambda\up{b}}\in \cE^\Th_\cG(\lambda)$, where $\cE_\cG^\Th(\lambda)$ is given in \eq{defcE}.
It then  follows from  \eq{difeq86} that
\beq \label{difeq}
\norm{{\vphi_{\lambda\up{a}}}- \vphi_{\lambda\up{b}}}\le  4 \e^{-m_1(2 \ell)^{-\kappa}{\ell_{\tau}}}\e^{L^\beta},
\eeq
so
\beq\label{disjs}
\abs{\scal{\vphi_{\lambda\up{a}},\vphi_{\lambda\up{b}}}}\ge  \Re \scal{{\vphi\up{a}_x},{\vphi\up{b}_y}}\ge 1- 8\e^{-2 m_1(2 \ell)^{-\kappa}{\ell_{\tau}}}\e^{2L^\beta} .
\eeq
On the other hand,  it follows from \eq{hypdec} that
\beq\label{disjs2}
\norm{x_{\lambda\up{a}}-x_{\lambda\up{b}}} \ge 2 {\ell_\tau} \quad \Longrightarrow  \quad \abs{\scal{{\vphi\up{a}_x},{\vphi\up{b}_y}}}\le   (\ell+1)^d  \e^{-m (2\ell)^{-\kappa}\ell_\tau    }.
\eeq
Combining \eq{disjs} and \eq{disjs2} we conclude that
\beq \label{vphiabxy}
{\lambda\up{a}},{\lambda\up{b}}\in \cE^\Th_\cG (\lambda) \quad \Longrightarrow \quad \norm{x_{\lambda\up{a}}-x_{\lambda\up{b}}} < 2 {\ell_{\tau}}.
\eeq

To prove \eq{sigmaab}, let $a,b \in \cG$, $a\ne b$. If $\Lambda_\ell(a) \cap \Lambda_\ell(b)=\emptyset $,  we have that
\beq
\lambda\up{a}\in \sigma_{I_{2\ell}}^{\Th,{\ell_\tau}}(H_{\La_\ell}(a))\sqtx{and}\lambda\up{b}\in  \sigma_{I_{2\ell}}^{\Th,{\ell_\tau}}(H_{\La_\ell}(b))\;\Longrightarrow \; \norm{x_{\lambda\up{a}}-x_{\lambda\up{b}}}  \ge 2\ell_{\tau},
\eeq
so it follows from  \eq{vphilambda}  that $\sigma_{\set{\La_\ell(a)}}(H_\Th)\cap \sigma_{\set{\La_\ell(b)}}(H_\Th)=\emptyset $.

Parts (ii)(b) and (ii)(c) are immediate consequence of parts (i)(b) and (i)(c), respectively. 
\end{proof}

\subsubsection{Buffered subsets}  In the multiscale analysis we will need to consider boxes $\La_\ell \subset \La_L $ that are  not $(m,I)$-localizing for $H$.  Instead of studying eigensystems for such boxes, we will surround them with a buffer of $(m,I)$-localizing boxes and study eigensystems for the augmented subset.

  \begin{definition}\label{defbuff} 
 We call $\Ups\subset \La_L$ an $(m,I)$-buffered subset of the box  $\La_L$   if 
 the following holds:
  
  \begin{enumerate}
\item  $\Ups $ is a connected set in $\Z^d$ of the form
\beq\label{defUpsinitial}
 \Ups= \bigcup_{j=1}^J \La_{R_j}(a_j)\cap \La_L,
 \eeq  
 where $J\in \N$, $a_1,a_2,\ldots, a_J \in \La^\R_L$, and $\ell \le R_j\le L$ for $j=1,2,\ldots,J$.
  
  \item ${\Upsilon}$  is $L$-level spacing for $H$.

  \item  There exists $\cG_\Ups \subset \La^\R_L$ such that:
  \begin{enumerate}
\item For all $a\in\cG_\Ups$ we have    $\La_\ell(a) \subset\Ups$, and    $\La_\ell(a)$ is  an  $(m,I)$-localizing box   for $H$.
\item  For all $y \in \partialin^{\La_L}\Ups$ there exists $a_y \in\cG_\Ups$ such that $y\in \La_\ell^{ {\Ups}, 2{\ell_\tau}}(a_y)$.
\end{enumerate}

  \end{enumerate}   
 In this case we set 
   \beq\label{defUpscheck}
\widecheck{\Upsilon} =\bigcup_{a \in \cG_\Ups}\La_\ell (a), \quad  \widecheck{\Upsilon}_{\tau} = \bigcup_{a \in \cG_\Ups}\La_\ell^{ {\Upsilon}, 2{\ell_\tau}}(a), \quad  \widehat{\Upsilon} = {\Upsilon} \setminus  \widecheck{\Upsilon},  \qtx{and} \widehat{\Upsilon}_\tau = {\Upsilon} \setminus  \widecheck{\Upsilon}_\tau .
\eeq
($\widecheck{\Upsilon} = \Upsilon_{\cG_\Ups}$ and $\widecheck{\Upsilon}_\tau = \Upsilon_{\cG_\Ups,\tau}$ in the notation of Lemma~\ref{lem:ident_eigensyst}.)

\end{definition}

  The set $\widecheck{\Upsilon}_{\tau}\supset \partialin^{\La_L}\Ups$ is a localizing buffer between  $ \widehat{\Upsilon}$ and $\La_L \setminus \Ups$, as shown in the following lemma.

 \begin{lemma}
 Let  $\Ups$  be an $(m,I)$-buffered subset of $\La_L$,  and let $\set{(\psi_\nu,\nu)}_{\nu \in \sigma(H_{\Upsilon})}$ be an eigensystem  for $H_{\Upsilon}$. 
   Let $\cG=\cG_\Ups$  and set 
 \beq\label{sigmabad}
\sigma_\cB(H_{{\Upsilon}})= \sigma_{I_{\ell}}(H_{{\Upsilon}})\setminus \sigma_\cG(H_{{\Upsilon}}),
\eeq 
where $\sigma_\cG(H_\Ups)$ is as in \eq{defcE}.  Then the following holds for sufficiently large $L$:

\begin{enumerate}
\item 
For all $\nu\in\sigma_\cB(H_{{\Upsilon}})$ we have 
\beq
\abs{\psi_\nu(y)} \le  e^{-m_2h_{I}(\nu){\ell_\tau}} \sqtx{for all} y \in \widecheck{\Upsilon}_{\tau}, \sqtx{with} m_2=m_2(\ell) \;\text{as in \eq{m414}}.  \label{vthetasmall}
\eeq.

\item Let $ \La_L$ be level spacing for $H$, and let $\set{(\phi_\lambda,\lambda)}_{\lambda \in \sigma(H_{\La_L})}$ be an eigensystem  for $H_{\La_L}$. 
There exists an injection    
\beq\label{injbad}
 \nu\in\sigma_\cB(H_{{\Upsilon}})  \mapsto \widetilde{\nu}\in \sigma(H_{\La_L})\setminus\sigma_\cG (H_{\La_L}),
\eeq
such that for  $\nu\in\sigma_\cB(H_{{\Upsilon}})$ we have 
\begin{gather}\label{tildedist}
\abs{\wtilde{\nu} -\nu} \le\e^{-m_4\ell^{\tau- \kappa}}, \sqtx{with} 
m_4=m_4(\ell) \ge  m\pa{1 - C_{d,m_-}\ell^{-\pa{\tau - \gamma \beta -\kappa-\kappa^\pr}}},
\end{gather}
and, redefining  ${\psi_\nu}$ so $\scal{\phi_{\wtilde{\nu}},\psi_\nu}>0$,
\beq \label{difeqr}
\norm{\phi_{\wtilde{\nu}}- \psi_\nu}\le 2\e^{-m_4\ell^{\tau- \kappa}}\e^{L^\beta}.
\eeq
\end{enumerate}
 \end{lemma}

\begin{proof} Part (i) follows immediately from Lemma~\ref{lem:ident_eigensyst}(ii)(c).

Now let $ \La_L$ be level spacing for $H$, and let $\set{(\phi_\lambda,\lambda)}_{\lambda \in \sigma(H_{\La_L})}$ be an eigensystem  for $H_{\La_L}$.  It follows from \cite[Eq.~(3.11) in Lemma~3.2]{EK} that for $ \nu\in\sigma_\cB(H_{{\Upsilon}})$ we have
 \begin{align}\nn
 \norm{\pa{H_{\La_L} -\nu}{\psi_\nu} }&\le (2d-1) \abs{ \partialex^{ \La_L}{\Ups}}^{\frac 12}\norm{\psi_\nu
 \Chi_{ \partialin^{ \La_L}{\Ups}}}_\infty\le  (2d-1)  L^{\frac {d}2} e^{-m_2h_{I}(\nu){\ell_\tau}}\\  &\le (2d-1)  L^{\frac {d}2} e^{-m_2\ell^{-\kappa}{\ell_\tau}} \le \e^{-m_4\ell^{\tau- \kappa}},  \label{appeig47549}
 \end{align}
where we used 
$\partialin^{ \La_L}{\Ups} \subset  \widecheck{\Upsilon}_{\tau}$ and \eq{vthetasmall}, and $m_4$ is given in \eq{tildedist}.
Since $ \La_L$  and $\Ups$ are $L$-level spacing for $H$, the map in \eq{injbad} is a well defined injection into $\sigma (H_{\La_L})$, and   \eq{difeqr} follows from \eq{tildedist} and \cite[Lemma~3.3]{EK}. 

To finish the proof we must show that $\wtilde\nu\notin \sigma_\cG(H_{\La_L})$ for all $\nu\in\sigma_\cB(H_{{\Upsilon}})$. 
Suppose  $\wtilde{\nu}\in  \sigma_\cG(H_{\La_L})$ for some $\nu\in\sigma_\cB(H_{{\Upsilon}}) $. Then there is $a\in \cG$ and $\lambda\up{a}\in   \sigma_{I_{2\ell}}^{\Th,{\ell_\tau}}(H_{\La_\ell}(a))$ such that
${\lambda\up{a}} \in\cE_\cG^{\La_L}(\wtilde{\nu})$.  On the other hand, it follows from   Lemma~\ref{lem:ident_eigensyst}(i)(a) that ${\lambda\up{a}} \in\cE_\cG^{\Ups}(\lambda_1) $ for some $\lambda_1 \in \sigma_\cG(H_{\Ups})$.  We conclude  from \eq{difeq86} and \eq{difeqr}  that
\begin{align}
\sqrt{2}= \norm{\psi_{\lambda_1} -\psi_{\nu}}&\le 
\norm{\psi_{\lambda_1}-\vphi_{\lambda\up{a}}  }+\norm{\vphi_{\lambda\up{a}} -{\phi}_{\wtilde{\nu}} }+ \norm{{\phi}_{\wtilde{\nu}} -\psi_{\nu} }\\
\notag & \le {4}\e^{-m_1{(2\ell)^{-\kappa}}{\ell_\tau}}\e^{L^\beta}+  {2}\e^{-m_4\ell^{\tau- \kappa}}\e^{L^\beta}< 1,
\end{align}
 a contradiction.
\end{proof}

 \begin{lemma}\label{lembad}   Let $\La_L=\La_L(x_0)$, $x_0 \in \R^d$. 
Let $ {\Upsilon}$ be    an $(m,I)$-buffered subset of  $\La_L$.  
 Let  $\cG=\cG_\Ups$, and for $\nu\in \sigma(H_{\Upsilon})$  set
\begin{align}\label{defcEUps}
\cE_\cG^{\La_L} (\nu)&=  \set{{\lambda\up{a}}; \; a\in \cG,\, \lambda\up{a}\in\sigma_{I_{2\ell}}^{\La_L,{\ell_\tau}}(H_{\La_\ell}(a) ), \,  \wtilde{\lambda}\up{a}=\nu}\subset \cE_\cG^{\Ups} (\nu), \notag \\ 
\sigma_\cG^{\La_L} (H_{{\Upsilon}})&=\set{\nu \in \sigma(H_{{\Upsilon}}); \; \cE_\cG^{\La_L} (\nu)\ne \emptyset} \subset \sigma_\cG (H_{{\Upsilon}}).
\end{align}
The following holds for sufficiently large $L$:
\begin{enumerate}
\item Let  $(\psi,\lambda)$ be an  eigenpair for $H_{\La_L }$ such that $\lambda \in I_{\ell}$ and
\beq\label{distpointeig1}
\abs{\lambda - \nu}
\ge \tfrac 1 2\e^{-L^\beta}\qtx{for all} \nu \in\sigma_\cG^{\La_L} (H_{{\Upsilon}})\cup \sigma_\cB  (H_{{\Upsilon}}).
\eeq
Then for all $y\in \Ups^{{\La_L},2{\ell_\tau} }$  we have
\beq\label{gggsum}
\abs{\psi(y)}\le   \e^{- m_5 h_{I_{\ell}}(\lambda)\ell_\tau }\abs{\psi(v)}\qtx{for some} v\in{\partial}^{\La_L, 2{\ell_\tau} }  {\Upsilon},
\eeq
where 
\beq\label{M21}
m_5=m_5(\ell)\ge  m\pa{1 - C_{d,m_-}\ell^{-\min \set{\kappa,\tau - \gamma \beta -\kappa-\kappa^\pr}}}.
\eeq
\item Let $\La_L$ be  level spacing for $H$,  let $\set{(\psi_\lambda,\lambda)}_{\lambda \in \sigma(H_{\La_L})}$ be an eigensystem  for $H_{\La_L}$, recall \eq{injbad},  and set
 \beq
\sigma_\Ups (H_{\La_L})= \set{\wtilde\nu; \,  \nu\in\sigma_\cB(H_{{\Upsilon}})}\subset  \sigma(H_{\La_L})\setminus\sigma_{\cG} (H_{\La_L}).
\eeq
  Then for all
   \[
\lambda \in \sigma_{ I_{\ell}} (H_{{\La_L}})\setminus\pa{  \sigma_\cG (H_{{\La_L}})\cup  \sigma_\Ups (H_{\La_L})},
\]  
the condition \eq{distpointeig1} is satisfied, and $\psi_\lambda$  satisfies  \eq{gggsum}.
\end{enumerate}
  \end{lemma}

\begin{proof} To prove part (i), 
we take $E=0$ by replacing the potential $V$ by $V-E$. 
Let  $(\psi,\lambda)$ be an  eigenpair for $H_{\La_L }$ satisfying \eq{distpointeig1}.  Given $y \in \Ups$, we write $\psi(y)$ as in \eq{eq:2terms}. We set
$P= \Chi_{I_{\ell }}\pa{H_ {\Ups}}$ and $\bar P= 1-P_{I_{\ell}}$. We use Lemma~\ref{lemkey2} with $\Phi=\Ups$ and  $J=I_{\ell}$. 

 To estimate  $\scal {\e^{-t \pa{H_{\Ups}^2-\lambda^2}}P\delta_y,\psi}$, let
 $\set{(\vtheta_\nu,\nu)}_{\nu \in \sigma(H_{\Upsilon})}$ be an eigensystem  for $H_{\Upsilon}$.    
For each  $\nu \in \sigma_\cG (H_{\Upsilon})$ we fix $\lambda\up{a_\nu}\in \cE_\cG^\Ups (\nu)$,  where 
  $a_\nu\in \cG,\, \lambda\up{a_\nu}\in  \sigma_{I_{2\ell}}^{\La_L,{\ell_\tau}}(H_{\La_\ell}(a_\nu) )$,  picking  $\lambda\up{a_\nu}  \in \cE_\cG^{\La_L} (\nu)$ if $\nu \in \sigma_\cG^{\La_L} (H_{\Upsilon})$, so   $ x_{\lambda\up{a_\nu}}\in \La_\ell^{{\La_L},\ell_{\tau}}(a_\nu )$.  If  $\nu \in\sigma_\cG (H_{{\Upsilon}})\setminus \sigma_\cG^{\La_L}(H_{{\Upsilon}})$ we have  $  x_{\lambda\up{a_\nu}}\in \La_\ell^{{\Upsilon},\ell_{\tau}}(a_\nu)\setminus \La_\ell^{{\La_L},\ell_{\tau}}(a_\nu)$.

 Given $J\subset \R$, we set
 $\sigma_{\cG,J} (H_{{\Upsilon}})=
 \sigma_\cG (H_{{\Upsilon}})\cap J$,   $\sigma_{\cG,J}^{\La_L} (H_{{\Upsilon}})=
 \sigma_\cG^{\La_L} (H_{{\Upsilon}})\cap J$.   We have
\begin{align}\label{sum000227}
\scal {\e^{-t \pa{H_{\Ups}^2-\lambda^2}}P\delta_y,\psi} &=\sum_{\nu \in\sigma_{I_{\ell}}( {\Upsilon})}\e^{-t \pa{\nu^2-\lambda^2}}{{\vtheta_\nu}(y)} \scal{ {\vtheta_\nu},\psi}\\   \notag  &  =  \sum_{\nu \in \sigma_{\cG,{I_{\ell}}}^{\La_L} (H_{{\Upsilon}})\cup \sigma_\cB  (H_{{\Upsilon}})} \e^{-t \pa{\nu^2-\lambda^2}}{{\vtheta_\nu}(y)} \scal{ {\vtheta_\nu},\psi}\\   \notag  & \qquad \quad +   \sum_{\nu \in\sigma_{\cG,{I_{\ell}}} (H_{{\Upsilon}})\setminus \sigma_{\cG,{I_{\ell}}}^{\La_L} (H_{{\Upsilon}})} \e^{-t \pa{\nu^2-\lambda^2}}{{\vtheta_\nu}(y)} \scal{ {\vtheta_\nu},\psi}.
\end{align}

If $\nu \in\sigma_\cG^{\La_L} (H_{{\Upsilon}})\cup \sigma_\cB  (H_{{\Upsilon}})$ we have
\beq
\scal{ {\vtheta_\nu},\psi}= \pa{\lambda- \nu }^{-1}\scal{ {\vtheta_\nu},  \pa{H_{{\La_L}}- \nu}\psi}= \pa{\lambda- \nu }^{-1}\scal{\pa{H_{{\La_L}}- \nu} {\vtheta_\nu}, \psi} .
\eeq
It follows from  \eq{distpointeig1} and \cite[Eq.~(3.10) in Lemma~3.2]{EK} that
\begin{align}\label{vphipsi39909}
&\abs{{{\vtheta_\nu}(y)} \scal{ {\vtheta_\nu},\psi}}\le 2\e^{L^\beta}   \abs{ \vtheta_\nu(y)} \sum_{  v \in    \partialex^{{{\La_L}}}  {\Upsilon}} \pa{ \textstyle\sum_{\substack{v^\pr \in \partial_{\mathrm{in}}^{ \La_L}{\Ups} \\ \abs{v^\pr-v}=1}}\abs{ \vtheta_\nu({v^\pr})}} \abs{\psi(v)}\\
&\qquad  \notag \le  2 L^{d}\e^{L^\beta} \set{2d \max_{  u \in   \partialin^{{{\La_L}}}  {\Upsilon}} {\abs{\vtheta_\nu(u)}} }  \abs{\psi(v_1)} \qtx{for some} v_1 \in\partialex^{{{\La_L}}} {\Upsilon}.
\end{align}
If $\nu \in \sigma_\cB  (H_{{\Upsilon}})$ it follows from \eq{vthetasmall} that 
\beq\max_{  u \in    \partialin^{{\La_L}}  {\Upsilon}} {\abs{\vtheta_\nu(u)}}   \le  e^{-m_2h_{I}(\nu){\ell_\tau}}.
\eeq
If  $\nu \in \sigma_{\cG,J}^{\La_L} (H_{{\Upsilon}})$, it follows from  \eq{difeq86} and   \eq{hypdec} that
\begin{align}
&\max_{  u \in   \partialin^{{{\La_L}}}  {\Upsilon}} {\abs{\vtheta_\nu(u)}}   \le 
\max_{  u \in   \partialin^{{{\La_L}}}  {\Upsilon}} \pa{{\abs{\vtheta_\nu(u)-\vphi\up{a_\nu}(u)}}+\abs{\vphi\up{a_\nu} (u)}}\\ 
&  \qquad  \notag \le 2 \e^{-m_1h_I(\lambda\up{a_\nu}){\ell_\tau}}\e^{L^\beta} +  e^{-mh_I(\lambda\up{a_\nu}){\ell_\tau}}\le 3\e^{-m_1h_I(\lambda\up{a_\nu}){\ell_\tau}}\e^{L^\beta}\\ \notag & \qquad \le 3 \e^{-m_1^\pr  h_I(\nu){\ell_\tau}}\e^{L^\beta}, \qtx{where} m_1^\pr \ge m_1(1 - \e^{-C_{m_-} \ell^{\tau-\kappa-\kappa^\pr})},
\end{align} 
where we used \eq{tildedist132}.
It follows that, with  
\beq
m_1^{\pr\pr}=\min \set{m_1^\pr, m_2}\ge  m\pa{1 - C_{d,m_-}\ell^{-\pa{\tau - \gamma \beta -\kappa-\kappa^\pr}}}, 
\eeq
for all $\nu \in \sigma_{\cG,J}^{\La_L} (H_{{\Upsilon}})\cup \sigma_\cB  (H_{{\Upsilon}})$ we have
\begin{align}
\abs{{{\vtheta_\nu}(y)} \scal{ {\vtheta_\nu},\psi}}\le \e^{3L^\beta}\e^{-m_1^{\pr\pr}  h_{I}(\nu){\ell_\tau}} \abs{\psi(v_1)} \sqtx{for some} v_1 \in\partialex^{{{\La_L}}} {\Upsilon}.
\end{align}¥
Picking $ t = \frac {m_1^{\pr\pr}\ell_\tau}{A^2 }$, for all $\nu \in \sigma_{\cG,J}^{\La_L} (H_{{\Upsilon}})\cup \sigma_\cB  (H_{{\Upsilon}})$ we get
\begin{align}
\e^{-t \pa{\nu^2-\lambda^2}}\abs{{{\vtheta_\nu}(y)} \scal{ {\vtheta_\nu},\psi}}\le \e^{3L^\beta}\e^{-m_1^{\pr\pr}  h_{I}(\lambda){\ell_\tau}} \abs{\psi(v_1)} \sqtx{for some} v_1 \in\partialex^{{{\La_L}}} {\Upsilon},
\end{align}¥
so
\begin{align}\label{gsum53}
¥\abs{ \sum_{\nu \in \sigma_{\cG,J}^{\La_L} (H_{{\Upsilon}})\cup \sigma_\cB  (H_{{\Upsilon}})} \e^{-t \pa{\nu^2-\lambda^2}}{{\vtheta_\nu}(y)} \scal{ {\vtheta_\nu},\psi}}\le (L+1)^d \e^{3L^\beta}\e^{-m_1^{\pr\pr}  h_{I}(\lambda){\ell_\tau}} \abs{\psi(v_2)}, 
\end{align}¥
for some $ v_2 \in\partialex^{{{\La_L}}} {\Upsilon}$.

Now let $\nu \in\sigma_{\cG,J} (H_{{\Upsilon}})\setminus \sigma_{\cG,J}^{\La_L} (H_{{\Upsilon}})$.  In this case  we have $  x_{\lambda\up{a_\nu}}\in \La_\ell^{{\Upsilon},\ell_{\tau}}(a_\nu)\setminus \La_\ell^{{\La_L},\ell_{\tau}}(a_\nu)$, so we have 
\beq
  \dist \pa{ x_{\lambda\up{a_\nu}}, \Ups \setminus \La_\ell(a_\nu)} > \ell_\tau \qtx{and}\dist \pa{ x_{\lambda\up{a_\nu}}, {\La_L}\setminus \La_\ell(a_\nu)} \le \ell_\tau,
 \eeq
 so there is $u_0 \in {\La_L} \setminus \Ups$ such that $\norm{ x_{\lambda\up{a_\nu}}- u_0  }\le  {\ell_\tau}$.
We now assume $y\in  \Ups^{{\La_L},2{\ell_\tau} }$, so we have  $\norm{y- u_0  }>2 \ell_\tau$.  We conclude that
\beq
\abs{ x_{\lambda\up{a_\nu}}-y}\ge  \norm{y- u_0  }- \norm{ x_{\lambda\up{a_\nu}}- u_0  } > 2\ell_\tau-\ell_\tau=\ell_\tau.
\eeq
Thus
\begin{align}\label{vthetabadg}
\abs{{\vtheta_\nu}(y)}&\le \abs{{\vtheta_\nu}(y) - \vphi_{\lambda\up{a_\nu}}(y)}+  \abs{  \vphi_{\lambda\up{a_\nu}}(y)}\\   \notag & \le {2}\e^{-m_1h_I (\lambda\up{a_\nu}) \ell_\tau}\e^{L^\beta}+ \e^{-   mh_I (\lambda\up{a_\nu}) \ell_\tau}\\   \notag & \le 3\e^{-m_1h_I (\lambda\up{a_\nu}) \ell_\tau}\e^{L^\beta}\le  
 3\e^{-m_1^\pr h_I (\nu) \ell_\tau}\e^{L^\beta},
 \end{align}
using  \eq{difeq86},  \eq{hypdec}, and  \eq{tildedist132}. It follows that
\beq\label{bsum40}
\abs{\sum_{\nu \in\sigma_{\cG,J} (H_{{\Upsilon}})\setminus \sigma_{\cG,J}^{\La_L} (H_{{\Upsilon}})}\e^{-t \pa{\nu^2-\lambda^2}}{{\vtheta_\nu}(y)} \scal{ {\vtheta_\nu},\psi}}
\le3 (L+1)^{\frac {3d}2}\e^{-m_1^{\pr\pr} h_I (\lambda)\ell_\tau}\e^{L^\beta}\abs{\psi(v_3)},
\eeq
for some  $v_3 \in {\Upsilon}$.

Combining \eq{sum000227}, \eq{gsum53} and \eq{bsum40}, we get for $y\in  \Ups^{{\La_L},2\fl{\ell^\tau} }$ that 
\beq\label{sum000227334}
\abs{\scal {\e^{-t \pa{H_{\Ups}^2-\lambda^2}}P\delta_y,\psi} }\le\e^{4L^\beta}\e^{-m_1^{\pr\pr} h_{I} (\lambda)\ell_\tau}\abs{\psi(v_4)} ,
\eeq
for some $v_4 \in {\Upsilon} \cup\partialex^{ {\La_L}}  {\Upsilon}$.

From  Lemma~\ref{lemkey2}, we get,
\begin{align}\label{aplemkey2}
¥\abs{\scal {\e^{- \frac {m_1^{\pr\pr}\ell_\tau}{A^2 }\pa{H_{\Ups}^2-\lambda^2}}\bar P\delta_y,\psi} }&\le (\ell+1)^{\frac d 2} \e^{m_1^{\pr\pr}(1- \ell^{-\kappa})^2h_{I_{\ell}}(\lambda)\ell_\tau  }\abs{\psi(v_5)},
\end{align}¥
for some $v_5\in \Ups$.

Combining \eq{sum000227334} and   \eq{aplemkey2},  we get 
\begin{align}\label{sum00022733426}
¥&\abs{\scal {\e^{-\frac {m_1^{\pr\pr}\ell_\tau}{A^2} \pa{H_{\Ups}^2-\lambda^2}}\delta_y,\psi} }\le 
2\e^{4L^\beta}\ \e^{- m_1^{\pr\pr} (1-\ell^{-\kappa})^2 h_{I_{\ell}}(\lambda)\ell_\tau }\abs{\psi(v_6)},
\end{align}¥
for some $v_6\in  {\Upsilon} \cup\partialex^{ {\La_L}}  {\Upsilon}$.

Using \eq{FHestRy} (with $ m= m_1^{\pr\pr}$,  $I=I_\ell$, $R=\ell_\tau$),     noting  $y\in \Ups^{{\La_L},2{\ell_\tau} }$ implies  $R_y^{\partial^{ \La_L}_{\mathrm{in}}\Ups}\ge 2\ell_\tau -1>\ell_\tau$, we get 
   \begin{align}\label{eq:anpar368888}
¥&\abs{\scal {F_{\frac {m_1^{\pr\pr}\ell_\tau}{A^2},\lambda} (H_\Ups)\delta_y,\Gamma_{ \boldsymbol{ \partial}^{ \Theta}  \Ups} \psi} }
\le 70 L^dA^{-1}  \e^{-m_1^{\pr\pr}h_{I}(\lambda)\ell_\tau }\abs{\psi(v)}, 
\end{align}¥
for some  $v \in \partial^{ \Theta}_{\mathrm{ex}}\Ups$.

 Combining \eq{sum00022733426} and \eq{eq:anpar368888}we get
 \begin{align}
 ¥\abs{\psi(y)}\le C_{d,m_-} \ell^{\kappa^\pr} \e^{4L^\beta}\ \e^{- m_1^{\pr\pr} (1-\ell^{-\kappa})^2 h_{I_{\ell}}(\lambda)\ell_\tau }\abs{\psi(v)}\le  \e^{- m_5 h_{I_{\ell}}(\lambda)\ell_\tau }\abs{\psi(v)},
\end{align}¥
 for some $v \in  {\Upsilon} \cup\partialex^{ {\La_L}}  {\Upsilon}$,
 where $m_5$ is as in \eq{M21}. Repeating the procedure as many times as needed, we can require $v \in \in{\partial}^{\La_L, 2{\ell_\tau} }  {\Upsilon}$.

Now  suppose $\La_L$ is level spacing for $H$, and let   $\lambda \in \sigma_{I_{\ell}} (H_{{\La_L}})\setminus\pa{  \sigma_\cG (H_{{\La_L}})\cup  \sigma_\Ups (H_{\La_L})}$.   If $\lambda \notin   \sigma_\cG (H_{{\La_L}})$, it follows from Lemma \ref{lem:ident_eigensyst}(i)(c) that \eq{distpointeiga} holds for all $a \in\cG$. If $\lambda \notin   \sigma_\Ups (H_{\La_L})$, the  argument in \eq{distpointsigeigab}, modified by the use of  \eq{tildedist} instead of  \eq{tildedist132}, using \eq{gamtzetabeta2}, gives $\abs{\lambda -\nu}\ge  \tfrac 12 \e^{-\L^\beta}$ for all $ \nu\in\sigma_\cB(H_{{\Upsilon}})$.  Thus  we have \eq{distpointeig1}, which implies \eq{gggsum}.
 \end{proof}

\subsection{Suitable covers of a box}\label{subsecsc}
To perform the multiscale analysis in an efficient way, it is convenient to use  a canonical  way to cover a box of side $L$ by boxes of side $\ell <L$.  We will  use the idea of
  suitable covers of a box as in \cite[Definition~3.12]{GKber}, adapted to the discrete case. 
   Since we will use  \eq{decayest12} to get decay of eigenfunctions in scale $L$ from decay in scale $\ell$, we will need to make sure $R_y^{\partial^{ \La_L}_{\mathrm{in}}\La_\ell} \approx \frac \ell 2$.  We will do so by ensuring that for all $y \in\La_L$ we can find a box $\La_\ell$ in the cover such that $y \in 
\La_\ell$ with $R_y^{\partial^{ \La_L}_{\mathrm{in}}\La_\ell} \approx \frac \ell 2\ge  \frac {\ell -\ell^\vs} 2 -1 $ for a fixed $\vs \in (0,1)$.   Later we will require $\vs$ as in \eq{vsdef} for convenience.

\begin{definition}\label{defcov} Fix  $\vs \in (0,1)$.  Let $\La_L=\La_L(x_0)$, $x_0 \in \R^d$ be  a box in $\Z^d$, and let $\ell < L$.
A suitable $\ell$-cover of $\La_L$ is
the collection of  boxes  
\begin{align}\label{standardcover}
{\mathcal C}_{L,\ell}={\mathcal C}_{L,\ell} \left(x_0 \right)= \set{ {\Lambda}_{\ell}(a)}_{a \in  \Xi_{L,\ell}},
\end{align}
where
\beq  \label{bbG}
 \Xi_{L,\ell}= \Xi_{L,\ell}(x_0):= \set{ x_0+ {\rho}\ell^\vs  \Z^{d}}\cap \La_L^\R 
\mqtx{with}   {\rho}\in  \br{\tfrac {1} {2},1}   \cap \set{\tfrac {L-\ell}{2 \ell^{\vs} k}; \, k \in \N }.
\eeq
We call ${\mathcal C}_{L,\ell} $ the suitable $\ell$-cover of $\La_L$  if   ${\rho} ={\rho}_{L,\ell}: =\max \br{\tfrac {1} {2},1}   \cap \set{\tfrac {L-\ell}{2 \ell^{\vs} k}; \, k \in \N }.
$
\end{definition}

We adapt \cite[Lemma~3.13]{GKber} to our context.

\begin{lemma}\label{lemcover } Let $\ell \le \frac  L   2$. Then for  every  box   $\La_L=\La_L(x_0)$, $x_0 \in \R^d$, a suitable $\ell$-cover
  ${\mathcal C}_{L,\ell}={\mathcal C}_{L,\ell}\left(x_0\right) $  satisfies
  \begin{align}\label{nestingproperty} 
&\La_L=\bigcup_{a \in  \Xi_{L,\ell}} {\Lambda}_{\ell}(a);\\ \label{covproperty}
&\text{for all}\; \; b \in\La_L  \sqtx{there is} {\Lambda}_{\ell}^{(b)} \in {\mathcal C}_{L,\ell}  \sqtx{such that}  
  b\in \pa{ {\Lambda}_{\ell}^{(b)}}^{\La_L,\frac {\ell -\ell^\vs}2},\\ \notag
 & \qtx{i.e.,} \La_L=\bigcup_{a \in  \Xi_{L,\ell}} {\Lambda}_{\ell}^{\La_L, \frac {\ell -\ell^\vs}2}(a);
    \\ \label{number}
&   \# \Xi_{L,\ell}= \pa{ \tfrac{L-\ell} {{\rho} \ell^\vs}+1}^{d }\le   \pa{\tfrac{2L} {\ell^\vs}}^{d}.  
\end{align}
Moreover,  given $a \in x_0 + {\rho}\ell^\vs  \Z^{d}$ and $k \in \N$, it follows that
\beq \label{nesting}
{\Lambda}_{(2  k {\rho}\ell^\vs  + \ell)}(a)= \bigcup_{b \in  \{ x_0 + {\rho}\ell^\vs  \Z^{d}\}\cap {\Lambda}^{\R}_{(2k {\rho} \ell^\vs + \ell)}(a) } {\Lambda}_{\ell}(b),
\eeq
and  $ \{ \Lambda_{\ell}(b)\}_{b \in  \{ x_0 + {\rho}\ell^\vs  \Z^{d}\}\cap {\Lambda}^{\R}_{(2k {\rho} \ell^\vs + \ell)}(a) }$ is a suitable $\ell$-cover of the box $\Lambda_{(2k {\rho} \ell^\vs + \ell)}(a)$. 
\end{lemma}

Note that $\Lambda_{\ell}^{(b)}$  does not denote a box  centered at $b$, just some box in ${\mathcal C}_{L,\ell} \left(x_0 \right)$ satisfying \eq{covproperty}.     By  $\Lambda_{\ell}^{(b)}$ we will always mean such a box.  We will use
\beq
R_b^{\partial^{\La_L}_{\mathrm{in}}\Lambda_{\ell}^{(b)}}  \ge  \tfrac {\ell -\ell^\vs} 2 -1 \qtx{for all} b \in \La_L.
\eeq
Note also  that  ${\rho} \le 1$ yields \eq{covproperty}.   We specified ${\rho}={\rho}_{L,\ell}$ in  for \emph{the} suitable $\ell$-cover for convenience, so there is no ambiguity  in the definition of ${\mathcal C}_{L,\ell} \left(x_0 \right) $.

Suitable covers are convenient for the construction of buffered subsets (see Definition~\ref{defbuff}) in the multiscale analysis.   We  will use the following observation:

\begin{remark} \label{remdisj}
 Let ${\mathcal C}_{L,\ell}$ be a suitable $\ell$-cover for the box $\La_L$, and set
\beq\label{defkell}
k_\ell= k_{L,\ell }=  \fl{\rho^{-1} \ell^{1-\vs}} +1.
\eeq
Then for all $a,b \in {\mathcal C}_{L,\ell}$ we have
\beq\label{disjbox}
\La_\ell^\R(a) \cap \La_\ell^\R(b) =\emptyset \quad \iff \quad \norm{a-b} \ge k_\ell \rho \ell^\vs .\eeq
\end{remark}

\subsection{Probability estimate for   level spacing}\label{secprobest}

The eigensystem multiscale analysis uses  a probability estimate of Klein and Molchanov  \cite[Lemma~2]{KlM},  which we state as in \cite[Lemma~2.1]{EK}.
If $J\subset \R$, we set $\diam J =\sup_{s,t \in J} \abs{s-t}$.

\begin{lemma}\label{lemSep}  Let  $H_{\bom}$ be an Anderson model as in Definition~\ref{defAnd}.
Let $\Th \subset \Z^d$ and $L>1$. Then
\begin{align}
\P\set{\Th \sqtx{is} L\text{-level spacing for}\;H_{\bom} }\ge 1 -Y_{\mu} \e^{-(2\alpha-1)L^\beta}\abs{\Th}^2.
\end{align}
where
\beq
 Y_{\mu} = 2^{2\alpha-1} \wtilde{K}^2 \pa{ \diam \supp \mu + 2d  +1},
\eeq
with $\wtilde{K}=K$ if $\alpha=1$ and $\wtilde{K}=8K$ if  $\alpha \in ( \frac 12,1)$.

In the special case of a box $\La_L$, we have
\begin{align}\label{probsep}
\P\set{\La_L \sqtx{is level spacing for} H }\ge 1 -Y_{\mu}\pa{L+1}^{2d} \e^{-(2\alpha-1)L^\beta}.
\end{align}
\end{lemma}

\section{Eigensystem multiscale analysis}\label{secEMSA}
In this section we fix an   Anderson model $H_{\bom}$ and prove Theorem~\ref{thmMSA}.
Note that $\vrho$ is given in  \eq{defvrho}.

 \begin{proposition}\label{propMSA}  Fix $m_- >0$.  There exists a 
 a finite scale   $\cL= \cL(d,m_-) $ with the following property:  Suppose for some scale 
$L_0 \ge \cL$ we have
  \begin{align}\label{initialconinduc}
\inf_{x\in \R^d} \P\set{\La_{L_0} (x) \sqtx{is}  (m_0,I_0) \text{-localizing for} \; H_{\bom}} \ge 1 -  \e^{-L_0^\zeta},
\end{align}
where  $I_0=(E-{A_0},E+{A_0})\subset \R$,   with $E\in \R$ and $A_0>0$, and  
 \beq\label{upbm2}
m_- L_0^{-\kappa^\pr} \le  m_0   \le  \tfrac 1 2 \log \pa{1 + \tfrac {A_0}{4d}}.
 \eeq
Set
 $L_{k+1}=L_k^\gamma$,  $A_{k+1}= A_k (1- L_k^{-\kappa})$, and $I_{k+1}= (E-A_{k+1}, E+A_{k+1})$,   for $k=0,1,\ldots$. 
 Then  for all $k=1,2,\ldots$ 
  we have
   \begin{align} \label{MSALk}
\inf_{x\in \R^d} \P\set{\La_{L_k} (x) \sqtx{is} ( m_k , I_k,I_{k-1})  \text{-localizing for} \; H_{\bom}} \ge 1 -  \e^{-L_k^\zeta} ,
\end{align}
where 
\begin{gather}\label{Minduc2}
m_-   L_k^{-\kappa^\pr}  <  m_{k-1}\pa{1- C_{d,m_-}  L_{k-1}^{-\vrho}}
 \le m_k < \tfrac 1 2 \log \pa{1 + \tfrac {A_k}{4d}}.  \end{gather}
 \end{proposition}

The proof of Proposition~\ref{propMSA} relies on the following lemma, the induction step for the multiscale analysis.

\begin{lemma}\label{lemInduction}
 Fix $m_->0$.  Let $I=(E-{A},E+{A})\subset \R$,   with $E\in \R$ and $A>0$, and $m>0$. Suppose for some scale $\ell$ we have
 \begin{align}\label{hypMSAlem}
\inf_{x\in \R^d} \P\set{\La_{\ell} (x) \sqtx{is}  (m,I) \text{-localizing for} \; H_{\bom}} \ge 1 -  \e^{-\ell^\zeta},
\end{align}
where
\beq\label{upbm3}
m_-  \ell^{-\kappa^\pr} \le m  \le  \tfrac 1 2 \log \pa{1 + \tfrac {A}{4d}}.
 \eeq
Then, if $\ell$ is sufficiently large, we have (recall $L=\ell^\gamma$)
 \begin{align}
\inf_{x\in \R^d} \P\set{\La_{L} (x) \sqtx{is}   (M,I_\ell,I) \text{-localizing for} \; H_{\bom}} \ge 1 -  \e^{-L^\zeta},
\end{align}
where
 \begin{align}\label{Minduc}
m_-  L^{-\kappa^\pr}  <  m\pa{1- C_{d,m_-} \ell^{-\vrho}}
 \le M < \tfrac 1 2 \log \pa{1 + \tfrac {A(1-\ell^{-\kappa})}{4d}}.
\end{align}
\end{lemma}

\begin{proof}
To prove the lemma we proceed as  in \cite[Proof of Lemma~4.5]{EK},  with some modifications.  The  crucial estimate  \eq{decayest12}  is a somewhat weaker statement than its counterpart  \cite[Eq.~(3.31)]{EK}.   For this reason we are forced to modify the definition
of an $\ell$-cover of a box, and use the version  given in Definition~\ref {defcov} with $\vs$ as in \eq{vsdef}, which differs from the version given in   \cite[Definition~3.10]{EK} which has $\vs=1$.  In particular, we have  \eq{disjbox},  while in \cite{EK} the corresponding  statement holds with the simpler $\norm{a-b} \ge 2 \rho \ell$.

 We assume \eq{hypMSAlem} \and \eq{upbm3} for a scale $\ell$. 
 We take
$\La_L=\La(x_0)$, where $x_0\in \R^d$,  and let
${\mathcal C}_{L,\ell}={\mathcal C}_{L,\ell} \left(x_0 \right)$ be the suitable $\ell$-cover of $\La_L$ (with $\vs$ as in \eq{vsdef}).  Given $a,b \in {\Xi}_{L,\ell}$, we will say that the boxes $\La_\ell(a)$ and $\La_\ell(b)$ are disjoint if and only if $ \La_\ell^\R(a) \cap \La_\ell^\R(b) =\emptyset$, that is, if and only if   $\norm{a-b} \ge k_\ell \rho \ell^\vs$  (see Remark~\ref{remdisj}).
  We take $N= N_\ell= \fl{\ell^{(\gamma-1)\tzeta}}$ (recall \eq{ttauzeta2}), and 
let  $\cB_N$ denote the event that there exist at most $N$ disjoint boxes in ${\mathcal C}_{L,\ell}$ that are not $(m,I)$-localizing for  $ H_{\bom}$. For sufficiently large $\ell$, we have, using \eq{number},  \eq{hypMSAlem}, and the fact that events on disjoint boxes are independent, that
\begin{align}\label{probBN}
\P\set{\cB_N^c}\le \pa{\tfrac{2L} {\ell^\vs}}^{(N+1)d} \e^{-(N+1)\ell^\zeta}= 2^{(N+1)d} \ell^{(\gamma-\vs)(N+1)d}\e^{-(N+1)\ell^\zeta} < \tfrac 12 \e^{-L^\zeta}.
\end{align}

We  now fix   $\bom \in\cB_N$.  There exists $\cA_N=\cA_N(\bom)\subset  \Xi_{L,\ell}=\Xi_{L,\ell} \left(x_0 \right)$ such that
$\abs{\cA_N}\le N$ and   $\norm{a-b} \ge k_\ell \rho \ell^\vs$ if $a,b \in \cA_N$ and $a\ne b$, with the following property:  if $a\in \Xi_{L,\ell} $ with
$\dist (a,\cA_N)\ge  k_\ell {\rho} \ell^\vs$, so $\La_\ell^\R(a)\cap \La_\ell^\R(b)=\emptyset$ for all $b\in \cA_N$,  the box $\La_\ell(a)$ is $(m,I)$-localizing for  $ H_{\bom}$.
 In other words,
\beq\label{implyloc}
a \in  \Xi_{L,\ell} \setminus \bigcup_{b\in \cA_N}\La^\R_{2(k_\ell-1) \rho \ell^\vs}(b) \quad \Longrightarrow  \quad \La_\ell(a) \sqtx{is} (m,I)\text{-localizing for} \quad  H_{\bom}.
\eeq

We  want to  embed the   boxes $\set{\La_\ell(b) }_{b \in \cA_{N}}$  into   $(m,I)$-buffered subsets of $\La_L$. To do so, we consider  graphs  $\G_i=\pa{ \Xi_{L,\ell}, \E_i}$, $i=1,2$,   
both having   $ \Xi_{L,\ell}$ as the set of vertices, with sets of edges given by
\begin{align}
 \E_1&=\set{\set{a,b}\in  \Xi_{L,\ell}^2;\;  \norm{a-b} \le (k_\ell-1)\rho \ell^\vs }\\  \notag
 & =\set{\set{a,b}\in  \Xi_{L,\ell}^2;\;  a\ne b \sqtx{and} \La^\R_\ell(a)\cap \La^\R_\ell(b)\ne \emptyset },\\
 \notag
 \E_2&=\set{\set{a,b}\in  \Xi_{L,\ell}^2;\; k_\ell\rho \ell^\vs\le  \norm{a-b}\le 3(k_\ell-1)\rho \ell^\vs}\\
 \notag & =\set{\set{a,b}\in  \Xi_{L,\ell}^2;\; \La^\R_\ell(a)\cap \La^\R_\ell(b)=\emptyset\sqtx{and} \ \norm{a-b}\le 3(k_\ell-1)\rho \ell^\vs}.
 \end{align}
 
Given $\Psi \subset \Xi_{L,\ell} $, we  let   $\overline{\Psi}= \Psi \cup \partial_{\mathrm{ex}}^{ \G_1} \Psi$, where $ \partial_{\mathrm{ex}}^{ \G_1} \Psi$,   the exterior boundary of $\Psi$ in the graph $\G_1 $, is defined by
\begin{align}
 \partial_{\mathrm{ex}}^{ \G_1} \Psi&= \set{a\in \Xi_{L,\ell}\setminus \Psi; \ \dist (a,\Psi)\le (k_\ell-1)\rho \ell^\vs }\\  \nn  
 &= \set{a\in \Xi_{L,\ell}\setminus \Psi; \ (b,a) \in \E_1 \sqtx{for some} b \in \Psi } .
 \end{align}

 Let  $\Phi\subset  \Xi_{L,\ell}$ be $\G_2$-connected, so $\diam \Phi \le 3\rho \ell \pa{\abs{\Phi}-1}$. Then
  \beq
\wtilde{\Phi}= \Xi_{L,\ell} \cap  \bigcup_{a\in \Phi} \La^\R_{(2{\rho} +1)\ell}(a)= \set{a\in \Xi_{L,\ell}; \; \dist (a,\Phi)\le {\rho} \ell  }
\eeq 
 is a  $\G_1$-connected subset of $ \Xi_{L,\ell}$ such that
 \beq\label{diamwphi}
 \diam \wtilde{\Phi} \le  \diam {\Phi} +2\rho \ell  \le  3(k_\ell-1)\rho \ell^\vs\pa{\abs{\Phi}-1}. 
 \eeq
We set 
\begin{align}\label{constUps}
\Ups_\Phi\up{0}  = \bigcup_{a \in \wtilde{\Phi}}  \La_\ell (a) \qtx{and}
\Ups_\Phi  = \Ups_\Phi \up{0} \cup 
\bigcup_{a \in \partial_{\mathrm{ex}}^{ \G_1}\wtilde{\Phi}} \La_\ell(a)       =     \bigcup_{a \in \overline{\wtilde{\Phi}}}\La_\ell(a).
\end{align}

Let $\set{\Phi_r}_{r=1}^R=\set{\Phi_r(\bom)}_{r=1}^R$ denote the $\G_2$-connected components of  $\cA_{N}$ (i.e., connected in the graph $\G_2$); we have 
$R \in \set{1,2,\ldots,N}$ and  $\sum_{r=1}^R\abs{\Phi_r}=\abs{\cA_N}\le N$.
We conclude  that   $\set{\wtilde{\Phi}_r}_{r=1}^R$ is a collection of disjoint, $\G_1$-connected subsets of $ \Xi_{L,\ell}$, such that
 \begin{gather}
  \dist (\wtilde{\Phi}_r,\wtilde{\Phi}_s)\ge  k_\ell\rho \ell^\vs \qtx{if} r\ne s. \label{distPhi}
 \end{gather}
  Moreover, it follows from \eq{implyloc} that
\beq
\label{implyloc2}
a \in \cG=\cG(\bom)=  \Xi_{L,\ell} \setminus \bigcup_{r=1}^R\wtilde{\Phi}_r\quad \Longrightarrow  \quad \La_\ell(a) \sqtx{is} (m,I)\text{-localizing for} \quad  H_{\bom}.
\eeq
In particular, we conclude that $\La_\ell(a)$ is $(m,I)$-localizing for $H_{\bom}$ for all $a \in \partial_{\mathrm{ex}}^{ \G_1}\wtilde{\Phi}_r$, $r=1,2,\ldots, R$.

Each $\Ups_r=\Ups_{\Phi_r}$, $r=1,2,\dots,R$,  clearly satisfies all the requirements to be an $(m,I)$-buffered subset of $\La_L$ with $\cG_{\Ups_r}=\partial_{\mathrm{ex}}^{ \G_1}\wtilde{\Phi}_r$  (see Definition~\ref{defbuff}), except that we do not know if $\Ups_r$ is   $L$-level spacing for $H_{\bom}$.
(Note that the sets $\{\Ups_r\up{0}\}_{r=1}^R$ are disjoint, but the sets $\set{\Ups_r}_{r=1}^R$ are not necessarily disjoint.)  Note also that it follows from \eq{diamwphi} that 
\beq\label{diamUps}
\diam \Ups_r \le \diam  \overline{\wtilde{\Phi}}_r  + \ell \le   (k_\ell-1)\rho \ell^\vs\pa{3\abs{\Phi_r}+1}+\ell\le 5\ell \abs{\Phi_r},
\eeq
so, using  \eq{gamtzetabeta}, we have
\beq\label{sumdiam}
\sum_{r=1}^R \diam \Ups_r\le 5\ell N  \le 5 \ell^{(\gamma-1)\tzeta  +1} \ll \ell^{\gamma \tau}=L^\tau.
\eeq

We can  arrange for $\set{\Ups_r}_{r=1}^R$ to be a collection of  $(m,I)$-buffered subsets of $\La_L$  as follows.  It follows from Lemma~\ref{lemSep} that 
for any $\Th\subset \La_L$ we have
\begin{align}\label{probspaced}
\P\set{\Th \sqtx{is} L\text{-level spacing for}\; H_{\bom} }\ge 1 -Y_{\mu} \e^{-(2\alpha-1)L^\beta}\pa{L+1}^{2d}.
\end{align}
 Let
\begin{align}
\cF_N=\bigcup_{r=1}^N  \cF(r), \sqtx{where} \cF(r)=\set{\Phi \subset \Xi_{L,\ell} ; \;  \Phi \sqtx{is}  \G_2\text{-connected}\sqtx{and} \abs{\Phi}=r}.
\end{align}
Setting $ \cF(r,a)=\set{\Phi \in \cF_r; \, a\in \Phi}$ for $a\in \Xi_{L,\ell}$, and noting that
 each vertex in the graph $\G_2$ has less than $\pa{ 6 k_\ell-5}^d\le \pa{14 \ell^{1-\vs}}^d$ nearest neighbors , we get
\begin{align}\nn
\abs{\cF(r,a)}\le  (r-1)! \pa{14 \ell^{1-\vs}}^{(r-1)d}\; & \Longrightarrow \; \abs{\cF(r)}\le (L+1)^d  (r-1)! \pa{14 \ell^{1-\vs}}^{(r-1)d}\\\label{cFN}
& \Longrightarrow \;\abs{\cF_N}\le (L+1)^d N! \pa{14 \ell^{1-\vs}}^{(N-1)d} .
\end{align}
Letting  $\cS_N$ denote that the event that the box $\La_L$ and the subsets  $\set{\Ups_\Phi}_{\Phi \in  \cF_N}$ are all  $L$-level spacing for  $ H_{\bom}$,   we get from \eq{probspaced} and \eq{cFN} that
\beq\label{probSN}
\P\set{\cS_N^c} \le Y_{\mu}\pa{1 +  (L+1)^d N! \pa{14 \ell^{1-\vs}}^{(N-1)d} }(L+1)^{2d} \e^{-(2\alpha-1)L^\beta} < \tfrac 12 \e^{-L^\zeta} 
\eeq
for sufficiently large $L$, since $(\gamma-1)\tzeta < (\gamma-1)\beta<\gamma \beta$ and $\zeta < \beta$.

We now define the event  $\cE_N= \cB_N \cap \cS_N$. It follows from \eq{probBN} and \eq{probSN} that
$
\P\set{\cE_N}> 1-  \e^{-L^\zeta}
$.
To finish the proof we need to show that for all  $\bom \in \cE_N$ the box $\La_L$ is 
$   (M,I_\ell,I)$-localizing for $ H_{\bom}$, where $M$ is given  in \eq{Minduc}.

Let us fix $\bom \in \cE_N$.   Then we have \eq{implyloc2}, $\La_L$ is level spacing for $ H_{\bom}$, and  the subsets $\set{\Ups_r}_{r=1}^R$ constructed in \eq{constUps} are $(m,I)$-buffered subsets of $\La_L$ for $ H_{\bom}$. It follows from \eq{covproperty} and Definition~\ref{defbuff}(iii) that
\beq\label{LadecompU}
\La_L=\set{ \bigcup_{a \in  \cG} {\Lambda}_{\ell}^{\La_L, \frac {\ell -\ell^\vs}2}(a)}\cup \set{\bigcup_{r=1}^R\Ups_r^{{\La_L},2\ell_\tau}}.
\eeq

Since $\bom$ is  fixed, we omit it from the notation.
 Let $\set{(\psi_\lambda,\lambda)}_{\lambda \in \sigma(H_{\La_L})}$ be an eigensystem for $H_{\La_L}$.  Given $a\in \cG$,
let
 $\set{(\vphi_{\lambda\up{a}}, \lambda\up{a})}_{\lambda\up{a}\in \sigma(H_ {\La_\ell (a)})}$ be an  $(m,I)$-localized eigensystem for $\La_\ell(a)
 $.  For $r=1,2,\ldots,R$,  let  $\set{(\phi_{\nu\up{r}},\nu\up{r})}_{\nu\up{r} \in \sigma(H_{\Ups_r})}$ be an eigensystem for $H_{\Ups_r}$, and set
 \beq\label{sigmanotsigma}
\sigma_{\Ups_r} (H_{\La_L})= \set{\wtilde{\nu}\up{r};  \nu\up{r}\in\sigma_\cB(H_{{\Upsilon_r}})}\subset  \sigma(H_{\La_L})\setminus\sigma_{\cG} (H_{\La_L}),
\eeq
where $\wtilde{\nu}\up{r}$ is given in \eq{injbad}, which  gives $\sigma_{\Ups_r} (H_{\La_L})\subset  \sigma(H_{\La_L})\setminus\sigma_{\cG_{\Ups_r}} (H_{\La_L})$, but the argument actually shows 
$\sigma_{\Ups_r} (H_{\La_L})\subset \sigma(H_{\La_L})\setminus\sigma_{\cG} (H_{\La_L})$.
We also set
\beq
\sigma_\cB(H_{\La_L}) =\bigcup_{r=1}^R \sigma_{\Ups_r} (H_{\La_L})\subset  \sigma(H_{\La_L})\setminus\sigma_{\cG} (H_{\La_L}).
\eeq

We claim
\beq\label{claimsp}
\sigma_{I_{\ell}}(H_{\La_L})\subset \sigma_{\cG} (H_{\La_L}) \cup \sigma_\cB (H_{\La_L}).
\eeq
To see this, suppose we have  $\lambda \in\sigma_{I_{\ell}}(H_{\La_L})\setminus \pa{ \sigma_{\cG} (H_{\La_L}) \cup \sigma_\cB (H_{\La_L})}$.  Since $\La_L$ is level spacing for $H$, it follows from Lemma~\ref{lem:ident_eigensyst}(ii)(c) that 
\beq\label{psidecgood6344}
\abs{\psi_\lambda(y)}\le\e^{-m_2 h_I(\lambda){{\ell_\tau}}}   \qtx{for all} y\in \bigcup_{a\in \cG}\La_\ell^{{\La_L},2\ell_\tau}(a),
\eeq
and it follows from Lemma~\ref{lembad}(ii) that
\beq
\abs{\psi_\lambda (y)} \le  \e^{-m_5 h_{I_{\ell}}(\lambda){\ell_{\tau}}} \qtx{for all}y\in \bigcup_{r=1}^R\Ups_r^{{\La_L},2\ell_\tau }.
\eeq
Using $\lambda \in I_\ell$, \eq{LadecompU},  \eq{upbm3}, and \eq{M21}	 we  conclude that (note $m_5\le m_2$)
\beq
1= \norm{\psi_\lambda} \le \e^{-m_5 h_{I_{\ell}}(\lambda){\ell_{\tau}}} \pa{L+1}^{\frac d 2}
\le  \e^{-\frac 1 2 \ell^{-(\kappa+\kappa^\pr)}{\ell_{\tau}}} \pa{L+1}^{\frac d 2}  <1,
\eeq 
a contradiction.  This establishes the claim.

To finish the proof we need to show  that $\set{(\psi_\lambda,\lambda)}_{\lambda \in \sigma(H_{\La_L})}$ is an $(M, I_\ell,I)$-localized eigensystem for $\La_L$, where $M$ is given in \eq{Minduc}.
We take $ \lambda \in \sigma_{I_{\ell}}(H_{\La_L})$, so  $h_{I_\ell} (\lambda) >0$.  In view of \eq{claimsp} we  consider several cases: 

\begin{enumerate}
\item Suppose $\lambda  \in \sigma_{\cG} (\La_L)$. In this case   
$\lambda \in \sigma_{\set{\La_\ell(a_\lambda)}}(H_{\La_L}) $ for some  $a_\lambda \in \cG$. We pick $x_\lambda \in \La_1(a_\lambda)$. 
In view of \eq{LadecompU} we consider two cases:

\begin{enumerate}
\item  
If  $y \in  {\Lambda}_{\ell}^{\La_L, \frac {\ell -\ell^\vs}2}(a)$ 
for some $a\in \cG$ and $\norm{y-x_\lambda} \ge 2\ell$, we must have  $\La_\ell (a_{\lambda})\cap  \La_\ell(a)=\emptyset$, so it follows from \eq{sigmaab} that
$\lambda\notin  \sigma_{\set{\La_\ell(a)}}(H_{\La_L})$, and, since ${R_y^{\partial^{ \La_L}_{\mathrm{in}}\La_\ell(a)}}\ge \fl{\frac {\ell -\ell^\vs}2}$,   \eq{psidecgoodpr} yields
\beq\label{psidecgoodpr2} 
\abs{\psi_\lambda(y)}\le \e^{-m_3 h_{I}(\lambda) \fl{\frac {\ell -\ell^\vs}2}}  \abs{\psi_\lambda (y_1)} \sqtx{for some} y_1 \in{\partial}^{\La_L,\ell_{\ttau }}  \Lambda_{\ell}(a).
\eeq
In particular,
\beq\label{psidecgoodpr299} 
\norm{y-y_1}\le \ell -  \fl{\tfrac {\ell -\ell^\vs}2} \le \tfrac {\ell +\ell^\vs}2 + 1\le  \tfrac {\ell +2\ell^\vs}2.
\eeq

\item If    $y\in \Ups_r^{{\La_L},2\ell_\tau}$ for some $r \in \set{1,2,\ldots,R}$, and $\norm{y-x_\lambda} \ge \ell+\diam \Ups_r$,  we must have $\La_\ell (a_{\lambda})\cap  \Ups_r=\emptyset$.   It follows from   \eq{sigmaab} that $\lambda \notin   \sigma_{\cG_{\Ups_r}} (H_{\La_L})$, and clearly $\lambda\notin    \sigma_{\Ups_r} (H_{\La_L})$ in view of \eq{sigmanotsigma}.   Thus Lemma~\ref{lembad}(ii)   gives 
\beq\label{gggsum356}
\abs{\psi_\lambda(y)}\le    \e^{-m_5 h_{I_{\ell}}(\lambda)\ell_\tau }\abs{\psi_\lambda(y_1)}\qtx{for some} y_1\in {\partial}^{\La_L, 2{\ell_\tau} }  {\Upsilon_r}.
\eeq
In particular,
\beq\label{gggsum35699}
\norm{y-y_1}\le \diam \Ups_r.
\eeq

\end{enumerate}
\item   Suppose $\lambda \notin \sigma_{\cG} (\La_L)$. Then it follows from \eq{claimsp} that we must have $\lambda_x \in    \sigma_{\Ups_s} (H_{\La_L})$ for some $s \in \set{1,2,\ldots,R}$. 
We pick $x_\lambda \in   \Ups_s^{{\La_L},2\ell_\tau}  $.   In view of \eq{LadecompU} we consider two possibilities:

\begin{enumerate}
 \item If   $y \in  {\Lambda}_{\ell}^{\La_L, \frac {\ell -\ell^\vs}2}(a)$
for some $a\in \cG$, and $\norm{y-x_\lambda} \ge \ell+\diam \Ups_s$, we must have  $\La_\ell (a)\cap  \Ups_s=\emptyset$,  and Lemma ~\ref{lem:ident_eigensyst}(i)(c) yields \eq{psidecgoodpr2}.

 \item If  $y\in \Ups_r^{{\La_L},2\ell_\tau}$ for some $r \in \set{1,2,\ldots,R}$, and $\norm{y-x_\lambda} \ge \diam \Ups_s+\diam \Ups_r$,  we must have $r\ne s$.   Thus   Lemma~\ref{lembad}(ii)  yields \eq{gggsum356}.
 
\end{enumerate}
\end{enumerate}

 Now consider  $y\in  \La_L $ such that $\norm{y-x_\lambda} \ge {L_\tau}$.
Suppose $\abs{\psi_\lambda(y)} >0$, since otherwise there is nothing to prove.  We estimate $\abs{\psi_\lambda(y)} $ using either \eq{psidecgoodpr2}  or  \eq{gggsum356} repeatedly, as appropriate, stopping when we get too close  to $x_\lambda$ so we are not in one the cases described above.  (Note that this must happen since $\abs{\psi_x(y)} >0$.) We accumulate decay only when we use \eq{psidecgoodpr2}, and just use $ \e^{-m_5 h_{I_{\ell}}(\lambda)\ell_\tau }< 1$   when using \eq{gggsum356}. In view of \eq{psidecgoodpr299}  and \eq{gggsum35699}, this can be done  using \eq{psidecgoodpr2} 
$S$ times, as long as 
\beq
   \tfrac {\ell +2\ell^\vs}2 S +\sum_{r=1}^R \diam \Ups_r + 2\ell  \le \norm{y-x_\lambda}.
\eeq
In view of  \eq{sumdiam}, this can be guaranteed by requiring
\begin{align}
¥  \tfrac {\ell +2\ell^\vs}2 S +  5 \ell^{(\gamma-1)\tzeta  +1} +2\ell  \le \norm{y-x_\lambda}.
\end{align}¥
We can thus have
\begin{align}
¥S& =  \fl{\tfrac 2{\ell +2\ell^\vs}\pa{\norm{y-x_\lambda}- 5 \ell^{(\gamma-1)\tzeta  +1}  -2\ell} }- 1
 \\ \nn  & \ge  \tfrac 2{\ell +2\ell^\vs}\pa{\norm{y-x_\lambda}- 5 \ell^{(\gamma-1)\tzeta  +1}  -2\ell} -2  \\ \nn  & \ge  \tfrac 2{\ell +2\ell^\vs}\pa{\norm{y-x_\lambda}- 5 \ell^{(\gamma-1)\tzeta  +1}  -3\ell-2\ell^\vs}  \ge  \tfrac 2{\ell +2\ell^\vs}\pa{\norm{y-x_\lambda}- 6 \ell^{(\gamma-1)\tzeta  +1}  } 
\end{align}¥

Thus we conclude that 
\begin{align}\label{repeateddecay}
\abs{\psi_\lambda(y)} &\le   \e^{-m_3 h_{I}(\lambda) \fl{\frac {\ell -\ell^\vs}2} { \tfrac 2{\ell +2\ell^\vs}\pa{\norm{y-x_\lambda}- 6 \ell^{(\gamma-1)\tzeta  +1}  }
}}  \le    \e^{-M h_{I}(\lambda)\norm{y-x_\lambda}}
\end{align}
where 
\begin{align}
M&\ge m_3\pa{1- C_{d,m_-} \ell^{-\min\set{{1- \vs }, \gamma \tau- (\gamma-1)\tzeta  -1}}}\\ \nn & = m_3\pa{1- C_{d,m_-} \ell^{-\min\set{ \gamma \tau- (\gamma-1)\tzeta  -1}}}
\\ \nn   &\ge 
m\pa{1- C_{d,m_-} \ell^{-\min\set{\kappa, \frac{1- \tau}2, \gamma \tau- (\gamma-1)\tzeta  -1}}}=m\pa{1- C_{d,m_-} \ell^{-\vrho}},
\end{align}
where we used \eq{vsdef}, \eq{m4}, and \eq{defvrho}.  In particular, $M$ satisfies \eq{Minduc} for sufficiently large $\ell$.

We conclude that $\set{(\psi_\lambda,\lambda)}_{\lambda\in \sigma(H_{ \La_L})}$ is an $(M,I_\ell,I)$-localized eigensystem for $\La_L$, where $M$ satisfies \eq{Minduc}, so the box  $\La_L$ is 
$ (M, I_\ell,I)$-localizing for $ H_{\bom}$.
\end{proof}

\begin{proof}[Proof of Proposition~\ref{propMSA}]
We assume \eq{initialconinduc} and \eq{upbm2} and set  $L_{k+1}=L_k^\gamma$,  $A_{k+1}= A_k (1- L_k^{-\kappa})$, and $I_{k+1}= (E-A_{k+1}, E+A_{k+1})$   for $k=0,1,\ldots$. 
Since if a box  $\La_{L}$ is  $(M, I_\ell,I)$-localizing for $ H_{\bom}$ it is also $(M, I_\ell)$-localizing, 
if $L_0$ is sufficiently large it follows from Lemma~\ref{lemInduction} by an induction argument that
for all $k=1,2, \ldots$ we have \eq{MSALk} and \eq{Minduc2}.
\end{proof}

 \begin{proposition}\label{propMSAnok} Fix $m_- >0$.  There exists a 
 a finite scale   $\cL= \cL(d,m_-) $ with the following property:  Suppose for some scale 
$L_0 \ge \cL$ we have
  \begin{align}\label{initialconinduc993}
\inf_{x\in \R^d} \P\set{\La_{L_0} (x) \sqtx{is}  (m_0,I_0) \text{-localizing for} \; H_{\bom}} \ge 1 -  \e^{-L_0^\zeta},
\end{align}
where  $I_0=(E-{A_0},E+{A_0})\subset \R$,   with $E\in \R$ and $A_0>0$, and
 \beq\label{upbm2555}
m_- L_0^{-\kappa^\pr} \le m_0   \le  \tfrac 1 2 \log \pa{1 + \tfrac {A_0}{4d}}.
 \eeq  
Set
 $L_{k+1}=L_k^\gamma$,  $A_{k+1}= A_k (1- L_k^{-\kappa})$, and $I_{k+1}= (E-A_{k+1}, E+A_{k+1})$,   for $k=0,1,\ldots$, 
 Then  for all $k=1,2,\ldots$   we have
   \begin{align} \label{MSALnok}
\inf_{x\in \R^d} \P\set{\La_{L} (x) \sqtx{is} ( m_k , I_{k},I_{k-1})  \text{-localizing for} \; H_{\bom}} \ge 1 -  \e^{-L^\xi}   \sqtx{for }     L\in [L_k, L_{k+1}),
\end{align}
where 
\begin{gather}\label{Minduc2333}
m_-   L_k^{-\kappa^\pr}  <  m_{k-1}\pa{1- C_{d,m_-}  L_{k-1}^{-\vrho}}
 \le m_k < \tfrac 1 2 \log \pa{1 + \tfrac {A_k}{4d}}, \end{gather}
 with  $C_{d,m_-}$ as in \eq{Minduc2}.
 \end{proposition}

\begin{proof}  We can apply  Proposition~\ref{propMSA}, so we have $\cL$.   Fix $L_0 \ge \cL$ ,
 so we have the conclusions of  Proposition~\ref{propMSA}.

Given a scale   $L\ge L_1$, let $k=k(L)\in \set{1,2,\ldots}$ be defined by
$L_k \le L <L_{k+1}$.  We have
 $L_k=L_{k-1}^\gamma\le L< L_{k+1}=L_{k-1}^{\gamma^2}$,  so $L=  L_{k-1}^{\gamma^\prime}$ with $\gamma \le \gamma^\prime<\gamma^2$. We proceed as in Lemma~\ref{lemInduction}. We take
$\La_L=\La_L(x_0)$, where $x_0\in \R^d$,  and let
${\mathcal C}_{L,L_{k-1}}={\mathcal C}_{L,L_{k-1}} \left(x_0 \right)$ be the suitable $L_{k-1}$-cover of $\La_L$. We let  $\cB_0$ denote the event that all  boxes in ${\mathcal C}_{L,L_{k-1}}$  are  $(m_{k-1},I_{k-1})$-localizing for  $ H_{\bom}$.  It follows from  \eq{number} and \eq{MSALk} that
\begin{align}\label{probB0}
\P\set{\cB_0^c}\le \pa{\tfrac{2L} {L_{k-1}^\vs}}^{d} \e^{-L_{k-1}^\zeta}= 2^{d} L_{k-1}^{(\gamma^\pr-\vs)d}\e^{-L_{k-1}^\zeta}\le     2^{d} L^{(1-\frac\vs { \gamma^\pr})d}\e^{-L^{\frac \zeta { \gamma^\pr}}}  < \tfrac 12 \e^{-L^\xi} ,
\end{align}
if $L_0$ is sufficiently large, since $\xi \gamma^2 < \zeta$.  Moreover, letting $\cS_0$ denote the event that  the box $\La_L$ is level spacing for  $H_{\bom}$, it follows from Lemma~\ref{lemSep} that 
\begin{align}\label{probspacedLL}
\P\set{\cS_0^c }\le Y_{\mu} \e^{-(2\alpha-1)L^\beta}\pa{L+1}^{2d}\le  \tfrac 12 \e^{-L^\xi},
\end{align}
if $L_0$ is sufficiently large, since $\xi < \beta$.  Thus, letting $\cE_0= \cB_0 \cap \cS_0$, we have 
\beq
\P\set{\cE_0 }\ge 1-  \e^{-L^\xi}.
\eeq

It only remains to prove that  $\La_L$ is  $(m_k,I_k,I_{k-1})$-localizing for  $ H_{\bom}$ for all $\bom \in \cE_0$.  To do so,  we fix $\bom \in \cE_0$ and  proceed as in the proof of Lemma~\ref{lemInduction}.   Since $\bom \in \cB_0$, we have $\cG=\cG(\bom)=\Xi_{L,L_{k-1}}$. Since  $\bom$ is now fixed, we omit them from the notation. As in the proof of  Lemma~\ref{lemInduction}, we get, noticing that   $(I_{k-1})_{L_{k-1}}= I_{k}$, 
\beq\label{claimsp22}
\sigma_{ I_{k}}(H_{\La_L})\subset\sigma_{\cG} (H_{\La_L}),
\eeq similarly to
\eq{claimsp}.   

 Let $\set{(\psi_\lambda,\lambda)}_{\lambda \in \sigma(H_{\La_L})}$ be an eigensystem for $H_{\La_L}$.  To finish the proof we need to show  that the eigensystem is    $  (m_k, I_k,I_{k-1})$-localized eigensystem for $\La_L$. Let  $ \lambda \in \sigma_{ I_{k}}(H_{\La_L})$, then by 
 \eq{claimsp22}   we have we have $\lambda \in \sigma_{\cG} (H_{\La_L})$, and hence $\lambda \in \sigma_{\set{\La_{L_{k-1}}(a_\lambda)}}(H_{\La_L}) $ for some  $a_\lambda \in \cG$. If $y \in \La_L$ and $\norm{y-x_\lambda} \ge 2L_{k-1}$, it follows from \eq{covproperty} that $y \in  {\Lambda}_{\ell}^{\La_L, \frac {L_{k-1} -L_{k-1}^\vs}2}(a)$  for some $a\in \cG$, and moreover   $\La_{L_{k-1}} (a_{\lambda})\cap  \La_{L_{k-1}}(a)=\emptyset$, so it follows from \eq{sigmaab} that
$\lambda\notin  \sigma_{\set{\La_{L_{k-1}}(a)}}(H_{\La_L})$, and, since ${R_y^{\partial^{ \La_L}_{\mathrm{in}}\La_{L_{k-1}}(a)}}\ge \fl{\frac {L_{k-1} -L_{k-1}^\vs}2}$, 
  \eq{psidecgoodpr} yields
\beq\label{psidecgoodpr256} 
\abs{\psi_\lambda(y)}\le \e^{-m_{k-1,3} h_{I_{k-1}}(\lambda) \fl{\frac {L_{k-1} -L_{k-1}^\vs}2}}  \abs{\psi_\lambda (y_1)} \sqtx{for some} y_1 \in{\partial}^{\La_L,\pa{L_{k-1}}_{\ttau }}  \Lambda_{L_{k-1}}(a),
\eeq
where we need
\beq\label{m4k}
m_{k-1,3}= m_{k-1,3} (L_{k-1})\ge  m_{k-1} \pa{1 - C_{d,m_-}L_{k-1}^{-( \frac{1- \tau}2)}},
\eeq
 and we have  
 \beq \label{psidecgoodpr29945}
 \norm{y-y_1}\le \tfrac {L_{k-1}+2L_{k-1}^\vs}2,
 \eeq as in
 \eq{psidecgoodpr299}.
 
 Now consider  $y\in  \La_L $ such that $\norm{y-x_\lambda} \ge {L_\tau}$.
Suppose $\abs{\psi_\lambda(y)} >0$, since otherwise there is nothing to prove.  We estimate $\abs{\psi_\lambda(y)} $ using either \eq{psidecgoodpr256} repeatedly, as appropriate, stopping when we get within $2 L_{k-1}$ of $x_\lambda$.  In view of \eq{psidecgoodpr29945}  , we can use \eq{psidecgoodpr256} 
$S$ times, as long as 
\beq
  \tfrac {L_{k-1}+2L_{k-1}^\vs}2 S + 2L_{k-1} \le \norm{y-x_\lambda}.
\eeq 
We can thus have
\begin{align}
¥S& =  \fl{\tfrac 2 {L_{k-1}+2L_{k-1}^\vs}\pa{\norm{y-x_\lambda}-2L_{k-1}} }- 1
 \\ \nn  & \ge  \fl{\tfrac 2 {L_{k-1}+2L_{k-1}^\vs}\pa{\norm{y-x_\lambda}-2L_{k-1}} }-2  \\ \nn  & \ge \tfrac 2 {L_{k-1}+2L_{k-1}^\vs}\pa{\norm{y-x_\lambda}-3L_{k-1}-2L_{k-1}^\vs}  \ge  \tfrac 2 {L_{k-1}+2L_{k-1}^\vs}\pa{\norm{y-x_\lambda}- 4L_{k-1} } .
\end{align}¥

Thus we conclude that 
\begin{align}\label{repeateddecay333}
\abs{\psi_\lambda(y)} &\le   \e^{-m_{k-1,3} h_{I_{k-1}}(\lambda) \fl{\frac {L_{k-1} -L_{k-1}^\vs}2}  {\tfrac 2 {L_{k-1}+2L_{k-1}^\vs}}\pa{\norm{y-x_\lambda}- 4L_{k-1} }
}\\ \nn &  \le    \e^{-m_k h_{I_{k-1}}(\lambda)\norm{y-x_\lambda}}
\end{align}
where  $m_k$ can be taken the same as  in \eq{Minduc2}.

We conclude that $\set{(\psi_\lambda,\lambda)}_{\lambda\in \sigma(H_{ \La_L})}$ is an $(m_k,I_k,I_{k-1})$-localized eigensystem for $\La_L$, where $m_k$ satisfies \eq{Minduc2}, so the box  $\La_L$ is 
$ (m_k, I_k,I_{k-1})$-localizing for $ H_{\bom}$.
\end{proof}

\begin{proof} [Proof of Theorem~\ref{thmMSA}]

 Let   $L_{k+1}=L_k^\gamma$,  $A_{k+1}= A_k (1- L_k^{-\kappa})$,  $I_{k+1}= (E-A_{k+1}, E+A_{k+1})$, and $m_{k+1}= m_{k}\pa{1- C_{d,m_-}  L_{k}^{-\vrho}}
  $   for $k=0,1,\ldots$.  Given $ L\ge L_0^\gamma =L_1$,  let $k=k(L)\in \set{1,2,\ldots}$ be defined by
$L_k \le L <L_{k+1}$.  Since
\beq
A_\infty \pa{1-L^{-\frac \kappa \gamma}}^{-1} < A_{k-1} \quad \Longrightarrow \quad I_\infty^{L^{\frac 1\gamma}} \subset I_{k-1},
\eeq
we conclude that
 \eq{MSALnok2} follows from \eq{MSALnok}.
\end{proof}

\section{Localization}\label{seclocproof}
 In this section we consider an   Anderson model $H_{\bom}$ and prove Theorem~\ref{thmloc}
and Corollary~\ref{corloc}.

 \begin{lemma}\label{lemWimp}  Fix $m_->0$, let $A>0$, and  $I=(E-{A},E+{A})$.  There exists a finite scale  $\cL_{d,\nu,m_-}$ such that  
 for all $L\ge \cL_{d,\nu,m_-}$, $a\in \Z^d$,  letting $L=\ell^\gamma$,   given    an $(m,I,I^{\ell})$-localizing box $\La_L(a)$  for the discrete Schr\"odinger operator  $H$,  where $m$ satisfies  \eq{upbm},
 then for all $\lambda \in I$,
 \beq
\max_{b\in \La_{\frac L 3}(a)} W\up{a}_{\lambda}(b)> \e^{-\frac 1 4 m h_{ I^{\ell}} (\lambda) L}\quad \Longrightarrow \quad  \min_{\theta\in \sigma_{ I^{\ell}}^{L_\tau}(H_{\La_L(a)})} \abs{\lambda -\theta}  < 
\tfrac 1 2 \e^{-L^\beta}.
\eeq
 \end{lemma}
  
 \begin{proof} Note that 
\beq  I \subset  I^{\ell}_L\subset I^{\ell}_{2L}\qtx{and}  \inf_{ \lambda \in I} h_{I^{\ell}_L} (\lambda) \ge  L^{-\frac \kappa \gamma}.
\eeq
 Now let $\lambda \in I \subset I^{\ell}_L$, and  suppose $ \abs{\lambda -\theta}  \ge\tfrac 1 2\e^{-L^\beta}$ for all $\theta\in \sigma_{I^\ell}^{L_\tau}(H_{\La_L(a)})$.  Let $\psi \in \cV(\lambda)$.  Then it follows from Lemma~\ref{lemdecay2}(ii) that  for large $L$ and $b\in \La_{\frac L 3}(a)$  we have
 \beq
 \abs{\psi(b)}\le \e^{-m_3 h_{I^\ell}(\lambda)\pa{\frac L 3 -1} }\norm{T_a^{-1} \psi}
  \scal{\tfrac L 2 + 1}^{\nu} \le \e^{-\frac 1 4 m h_{I^\ell}(\lambda)L} \norm{T_a^{-1} \psi}.
 \eeq
 \end{proof}
 
\begin{proof}[Proof of Theorem~\ref{thmloc}]  Assume  Theorem~\ref{thmMSA} holds for some $L_0$, 
and let $I=I_\infty$, $m=m_\infty$. 
Consider  $L_0^\gamma \le L \in 2\N$  and $a\in \Z^d$.  We have  
 \beq
  \La_{5L}(a) =\bigcup_{b\in \set{ a+ \frac 1 2L  \Z^d}, \ \norm{b-a}\le  2L} \La_{L}(b).
  \eeq
  Let $\cY_{L,a}$ denote the event that  $ \La_{5L}(a)$ is level spacing for $H_{\bom}$ and the boxes  $ \La_{L}(b)$ are  $(m,I,I^\ell)$-localizing for $H_{\bom}$ for all $b\in \set{ a+ \frac 1 2L  \Z^d}$ with  $\norm{b-a}\le 2 L$, where $L=\ell^\gamma$.  It follows from \eq{MSALnok2} and Lemma~\ref{lemSep}  that
  \beq
  \P\set{\cY_{L,a}^c}\le 5^d \e^{-L^\xi}  +  Y_{\mu}\pa{5L+1}^{2d} \e^{-(2\alpha-1)(5L)^\beta}\le C_{\mu} \e^{-L^\xi} .
  \eeq
  
  Suppose $\bom \in \cY_{ L,a}$, $\lambda \in I$,  and $\max_{b\in \La_{\frac L 3}(a)} W\up{a}_{\bom,\lambda}(b)> \e^{-\frac 1 4 m h_{I^\ell} (\lambda)L }$. It follows from Lemma~\ref{lemWimp} that $ \min_{\theta\in \sigma_{I^\ell}^{L_\tau}(H_{\La_L(a)})} \abs{\lambda -\theta}  < \tfrac 1 2\e^{-L^\beta}$.  Since $ \La_{5L}(a)$ is  level spacing for $H_{\bom}$, using Lemma~\ref{lem:ident_eigensyst}(i)(a)  we conclude that 
  \begin{align} 
   \min_{\theta\in \sigma_{I^\ell}^{L_\tau}(H_{\La_L(b)})} \abs{\lambda -\theta}& \ge \e^{-(5L)^\beta} - 
   2  \e^{-m_1h_{ I^\ell} (\lambda) {L_\tau}} - \tfrac 1 2\e^{-L^\beta}\\ \nn &
   \ge \e^{-(5L)^\beta} - 
   2  \e^{-m_1L^{-\frac \kappa \gamma} {L_\tau}} - \tfrac 1 2\e^{-L^\beta}\ge \tfrac 1 2 \e^{-L^\beta}
\end{align}
  for all $b\in \set{ a+ \frac 1 2L  \Z^d}$ with $L \le \norm{b-a}\le 2L$. Since
  \beq\label{Aell2}
 A_L(a)\subset \bigcup_{b\in \set{ a+ \frac 1 2L  \Z^d}, \ L \le \norm{b-a}\le 2L} \La_{L}^{\frac L 7}(b),
  \eeq
 it follows from Lemma~\ref{lemdecay2}(ii)  that for all  $y\in A_L(a)$  we have, given  $\psi \in \cV_{\bom}(\lambda)$,
 \begin{align}
 \abs{\psi(y)}&\le  \e^{-m_3  h_{I^\ell} (\lambda)\pa{\frac L 7 - 2}} \norm{T_a^{-1} \psi}\la \tfrac 5 2 L  +1 \ra^\nu\le \e^{-m h_{I^\ell} (\lambda) \frac L 8}\norm{T_a^{-1} \psi}&  \\  \nn
 & \le  \e^{- \frac 7 {132} m   h_{I^\ell} (\lambda)\norm{y-a} }\norm{T_a^{-1} \psi},
 \end{align}
so we get
\beq
W\up{a}_{\bom,\lambda}(y)\le \e^{-\frac 7 {132}        m  h_{I^\ell} (\lambda)\norm{y-a}} \qtx{for all} y\in A_L(a).
\eeq

Since we have \eq{boundGW}, we conclude that for $\bom \in \cY_{L,a}$ we always have
\begin{align} 
W\up{a}_{\bom,\lambda}(a)W\up{a}_{\bom,\lambda}(y) & \le 
\max \set{\e^{- \frac 7 {66}  m  h_{I^\ell} (\lambda)\norm{y-a}}\la  y-a\ra^\nu, \e^{- \frac 7 {132} m  h_{I^\ell} (\lambda)\ \norm{y-a}}}\\   \notag 
& \le \e^{- \frac 7 {132}  m  h_{I^\ell} (\lambda)\norm{y-a}} \qtx{for all} y\in A_L(a).
\end{align}
 \end{proof}

\begin{proof}[Proof of Corollary~\ref{corloc}]

Parts (i) and (ii) are proven in the same way as  \cite[Theorem~7.1(i)-(ii)]{GKber}, using 
$h_{I^L} \ge h_I  $ for all $L >1$. 

Part (iii) is proven similarly to  \cite[Corollary~1.8(iii)]{EK} and  \cite[Theorem~7.2(i)]{GKber}.   We use
the fact that  for any $L_0 \in 2\N$, setting $L_{k+1}=2L_k$ for $k=0,1,2,\ldots$, we have  (recall \eq{Aell})
   \beq
  \Z^d = \La_{3L_k}(a)\cup \bigcup_{j=k}^\infty  A_{L_j}(a) \qtx{for} k=0,1,2,\ldots.
   \eeq

Given  $k \in \N$, we set $L_{k}=2^{k}$,
 and consider the event  
\beq
\cY_{k}:=  \bigcap_{x \in \Z^{d}; \, \norm{x} \le  \e^{\frac 1 {2d}  L_k^\xi}}  \cY_{L_k,x} ,
\eeq
where $\cY_{\L_k,x} $ is the event given in Theorem~\ref{thmloc}.  It follows from \eq{cUdesiredint} that for sufficiently large $k$ we have
\beq
\P\set{\cY_{k}}\ge 1 - C\pa{2\e^{\frac 1 {2d}  L_k^\xi}+1}^d \e^{-L_k^\xi}\ge 1 - 3^d C \e^{-\frac 1 {2}  L_k^\xi},
\eeq
so we conclude from the  Borel-Cantelli Lemma  that 
\beq\label{cUinfty2}
\P \set{\cY_{\infty}}=1, \quad \text{where}\quad \cY_{\infty}=\liminf_{k\to \infty}\cY_{k}.
\eeq

We now fix $\bom \in  \cY_{\infty}$, so there exists $k_{\bom}\in \N$ such that
$\bom \in \cY_{L_k,x} $ for all $k_{\bom} \le k \in \N$ and $x \in \Z^d$ with  $\norm{x} \le  \e^{\frac 1 {2d}  L_k^\xi}$. We set $k^\pr_{\bom}=\max\set{k_{\bom},2}$.
Given $x \in \Z^d$, we define $k_x\in \N$  by
\beq\label{defkx}
 \e^{ \frac 1 {2d}  L_{k_x-1}^\xi} < \norm{x} \le  \e^{\frac 1 {2d}  L_{k_x}^\xi} \qtx{if} k_x \ge 2,
\eeq
and set $k_x=1$ otherwise. We set $ k_{\bom,x} =\max\set{k^\pr_{\bom},k_x}$

Let  $x \in \Z^d$.  If $y \in B_{\bom,x}=\bigcup_{k=k_{\bom,x}}^\infty  A_{L_k}(x)$, we have $y \in  A_{L_{k_1}}(x)$ for some $k_1 \ge  k_{\bom,x}$ and  $\bom \in \cY_{L_{k_1},x} $, so it follows from \eq{WW} that
\beq  \label{WW2}
W\up{x}_{\bom,\lambda}(x)W\up{x}_{\bom,\lambda}(y)\le 
 \e^{- \frac 7 {132} m  h_{I} (\lambda) \norm{y-x}} \qtx{for all} \lambda \in I.
\eeq
If $y \notin B_{\bom,x}$, we must have  $\norm{y-x}< \frac 8 7 L_{k_{\bom,x}}$, so for all $ \lambda \in \R$, using \eq{boundGW} and \eq{defkx},
\begin{align} \label{preSUDEC2}
& W\up{x}_{\bom,\lambda}(x)W\up{x}_{\bom,\lambda}(y) =W\up{x}_{\bom,\lambda}(x)W\up{x}_{\bom,\lambda}(y) \e^{\frac 7 {132} m  h_{I} (\lambda)\norm{y-x}}\e^{- \frac 7 {132} m  h_{I} (\lambda)\norm{y-x}}\\ \notag 
& \qquad
\le \la y-x \ra^\nu \e^{\frac 7 {132} m  h_{I} (\lambda) \norm{y-x}} \e^{- \frac 7 {132} m  h_{I} (\lambda)\norm{y-x}} \\ \notag & \qquad  \le \scal{\tfrac 8 7 L_{k_{\bom,x}} }^\nu  \e^{\frac 2 {33} m  h_{I} (\lambda) L_{k_{\bom,x}}} \e^{- \frac 7 {132} m  h_{I} (\lambda)\norm{y-x}} 
\\ & \qquad
\le
\begin{cases} 
\scal{ \tfrac {16} 7 \pa{\log \norm{x}^{2d}}^{\frac 1 \xi} }^\nu \e^{\frac 4 {33} m  h_{I} (\lambda) \pa{\log \norm{x}^{2d}}^{\frac 1 \xi}} 
 \e^{- \frac 7 {132} m  h_{I} (\lambda) \norm{y-x}} & \quad  \text{if} \quad k_{\bom,x}=k_x\\
\scal{\tfrac 8 7 L_{k^\pr_{\bom}} }^\nu  \e^{\frac 2 {33} m  h_{I} (\lambda)  L_{k^\pr_{\bom}}} 
 \e^{- \frac 7 {132} m  h_{I} (\lambda) \norm{y-x}} & \quad  \text{if} \quad k_{\bom,x}= k^\pr_{\bom}
\end{cases}   .\notag
\end{align}

 Combining \eq{WW2} and \eq{preSUDEC2},  noting $\norm{x}^{2d} >\e$ if $k_x\ge 2$,
and $ h_{I} (\lambda)\le 1$,  we conclude that for  all $ \lambda \in I$ with $h_I(\lambda)>0$ and  $x,y\in\Z^d$ we have 
\begin{align}
&W\up{x}_{\bom,\lambda}(x)W\up{x}_{\bom,\lambda}(y) \\ \nn  & \quad 
\le  C_{m,\bom,\nu}  \scal{  (2d\log \scal{x})^{\frac 1 \xi} }^\nu\e^{\frac 4 {33} m  h_{I} (\lambda)   (2d\log \scal{x})^{\frac 1 \xi}}  \e^{- \frac 7 {132} m  h_{I} (\lambda) \norm{y-x}} \\ \nn  & \quad 
\le  C_{m,\bom,\nu} \scal{  \pa{m  h_{I} (\lambda)}^{-1} }^\nu\e^{(\frac 4 {33} +\nu)m  h_{I} (\lambda)   (2d\log \scal{x})^{\frac 1 \xi}}  \e^{- \frac 7 {132} m  h_{I} (\lambda) \norm{y-x}} \\ \nn  & \quad 
\le  C^\pr_{m,\bom,\nu} \pa{h_{I} (\lambda)}^{-\nu}\e^{(\frac 4 {33} +\nu)m  h_{I} (\lambda)   (2d\log \scal{x})^{\frac 1 \xi}}  \e^{- \frac 7 {132} m  h_{I} (\lambda) \norm{y-x}} ,
\end{align}
which is \eq{eqWW}.

Part (iv) follows  from (iii), since \eq{eqWW} implies
\begin{align}\label{eqWW24444}
 &\abs{\psi(x)}\abs{\psi(y)}\\ \notag
& \;
\le  C_{m,\bom,\nu} \pa{h_{I} (\lambda)}^{-\nu}\, \norm{T_x^{-1} \psi}^2 \e^{(\frac 4 {33} +\nu)m  h_{I} (\lambda)   (2d\log \scal{x})^{\frac 1 \xi}}  \e^{- \frac 7 {132} m  h_{I} (\lambda) \norm{y-x}}  \\ \notag
&  \;\le  C_{m,\bom,\nu} \pa{h_{I} (\lambda)}^{-\nu}\, \norm{T_0^{-1} \psi}^2\la x\ra^{2\nu} \e^{(\frac 4 {33} +\nu)m  h_{I} (\lambda)   (2d\log \scal{x})^{\frac 1 \xi}}  \e^{- \frac 7 {132} m  h_{I} (\lambda) \norm{y-x}} ,
 \end{align}
for all $x,y \in \Z^d$, which is \eq{eqWW2}.

 Part   (v) similarly follows from (iii) using the discrete equivalent of   \cite[Eq.~(4.22)]{GKsudec}.
\end{proof}

\section{Connection with  the Green's functions multiscale analysis}\label{secGreen}
  Let $H_{\bom}$ be an Anderson model.   Given $\Th \subset \Z^d$ finite and 
$z \notin \sigma \left( H_{\Th}   \right)$, we set 
\beq  G_{\Th}(z)= (H_{\Th}  -z)^{-1}\mqtx{and} 
G_{\Th}(z;x,y) = \scal{ \delta_{x}, (H_{\Th}  -z)^{-1}\delta_{y}} \;\; \text{for} \;\;  x, \, y \in \Th.
\eeq
\begin{definition} Let $E\in \R$ and  $m > 0$. A box $\La_L$ is said to be  $(m, E)$-regular if  
$E \notin \sigma (H_{\La_L} )$  and
\beq\label{Gdecay}
 \abs{G_{\La_L}(E; x, y)} \leq e^{-m\norm{x -y}}
\;\; \text{for all} \;\;x,y \in\La_L \;\; \text{with} \; \norm{x -y} \geq \tfrac{L}{100} .
\eeq
\end{definition}

The following theorem is a typical result from the Green's function multiscale analysis.  \cite{FS,FMSS,DK,GKboot,Kle}.

\begin{theorem}\label{thmGMSA}
Let $J\subset \R$ be a  bounded open interval, $0<\xi<\zeta<1$, and $ m_0>0$.  Suppose for some scale $L_0$ we have 
 \beq\label{initialSingleMSA}
 \inf_{x\in \R^d}  \P \set{\La_{L_0} (x)\sqtx{is}  (m, \lambda)\text{-regular} }\ge 1 -  \e^{-L_0^{\zeta}}  \qtx{for all} \lambda\in J .
\eeq
Then, given $m \in (0, m_0)$,  if $L_0$ is sufficiently large, we have 
 \beq\label{concSingleMSA3}
 \inf_{x\in \R^d}  \P \set{\La_{L} (x)\sqtx{is}  (m, \lambda)\text{-regular} }\ge 1 -  \e^{-L^{\xi}}  \qtx{for all} \lambda\in J ,
   \eeq
and
 \begin{align}\label{concEGMSA}
\inf_{\substack{
x,y\in \R^d\\ \norm{x-y} > L
}} \P\set{\text{for all}\;  \lambda\in J\sqtx{either} \La_{L} (x) \sqtx{or}\La_{L} (y)\sqtx{is}  (m, \lambda)\text{-regular} } \ge 1 -  \e^{-L^\xi}.
\end{align}
\end{theorem}

Here \eq{concSingleMSA3} are the conclusions of the single energy  multiscale analysis, and \eq{concEGMSA} are the conclusions of the  energy interval   multiscale analysis.

 Given a  bounded open interval $J$ and $m>0$,  we call a box
$\La_L$   $(m,J)$-uniformly localizing for $H$ if $\La_L$ is level spacing for $H$, and  there exists an  eigensystem   $\set{(\vphi_\nu, \nu)}_{\nu \in \sigma(H_{\La_L})}$ for $H_{\La_L}$ such that for all $\nu \in \sigma_J (H_{\La_L})$ there is $x_\nu\in \La_L$ such that  $\vphi_\nu$ is $(x_\nu, m )$-localized.  Note that if $\La_L$ is  $(m,J)$-localizing for $H$ (as in  Definition~\ref{defmIloc}), it follows  from \eq{lowerbdh} that $\La_L$ is   $(mr^{-\kappa},J_{r})$-uniformly localizing for $H$ for all $r>1$.

\begin{proposition}\label{fromresMSAtoloc}  Let $J\subset \R$ be a  bounded open  interval, $0<\xi^\pr <\xi<1$, and $ m>0$.
Suppose there exists $\cL$ such that  the Anderson model $H_\bom$ satisfies \eq{concEGMSA} for all $L\ge \cL$.  Then, given $m^\pr \in (0, m)$,  for sufficiently large $L$ we have
 \begin{align}\label{unifloc}
\inf_{x\in \R^d} \P\set{\La_{L} (x) \sqtx{is}  (m^\pr, J) \text{-uniformly localizing for} \; H_{\bom}} \ge 1 -  \e^{-L^{\xi^\pr}}.
\end{align}
\end{proposition}

Proposition~\ref{fromresMSAtoloc} is proved  exactly as the analogous  result in  \cite[Proposition~6.4]{EK}.

We now show that the conclusions of  Theorem~\ref{thmMSA} imply a result similar to the the conclusions of Theorem~\ref{thmGMSA}.

\begin{lemma}\label{lemtoreg} Fix $m_->0$.  Let $I=(E-{A},E+{A})\subset \R$,   with $E\in \R$ and $A>0$, and $m>0$.   Suppose that $\La_L$ is  $(m,I)$-localizing for $H$, where
\beq\label{upbm37}
m_-  L^{-\kappa^\pr} \le m  \le  \tfrac 1 2 \log \pa{1 + \tfrac {A}{4d}}.
 \eeq
Then, for sufficiently large $L$, $\La_L$ is $(m^{\pr\pr}  h_{I}(\lambda), \lambda)$-regular
 for all $\lambda \in I_{L}$ with $\dist \set{\lambda, \sigma(H_{\La_L})}\ge \e^{-L^\beta}$, where
 \beq\label{mprpr3}
 m^{\pr\pr} \ge m \pa{1 - C_{d,m_-}L^{-(1- \tau)}}.
 \eeq
\end{lemma}

\begin{proof}

We take $E=0$ by replacing the potential $V$ by $V-E$.

Let  $\lambda \in I$ with $\dist \set{\lambda, \sigma(H_{\La_L})}\ge \e^{-L^\beta}$.  For all $t> 0$ we have
\beq\label{resolvdecomp}
G_{\La_L}(\lambda)= (H_{\La_L}  -\lambda)^{-1}= F_{t,\lambda}(H_{\La_L}) +  (H_{\La_L}  -\lambda)^{-1}\e^{-t \pa{H_{\La_L}^2-\lambda^2}}
\eeq
where the function $F_{t,\lambda}(z)$ is defined in \eq{defanf}.

Let $\set{(\vphi_\nu, \nu)}_{\nu \in \sigma(H_{\La_L})}$ be  an  $(m,I)$-localized eigensystem for $H_{\La_L}$.
 Let  $\nu \in \sigma_{I}(H_{\La_L})$ and  $x,y \in  {\Lambda}_{L}$ with $\norm{x-y}\ge \frac L {100}$. In this case either $ \norm{x-x_\nu}\ge L_\tau$ or $ \norm{y-x_\nu}\ge L_\tau$.  Say  $ \norm{x-x_\nu}\ge L_\tau$, then 
\beq
\abs{\vphi_\nu(x)\vphi_\nu(y)}\le \begin{cases} \e^{-m h_I(\nu)  \pa{\norm{x-x_\nu} +\norm{y-x_\nu}}}\le  \e^{-m h_I(\nu) \norm{x-y}} & \text{if}\quad  \norm{y-x_\nu}\ge L_\tau\\
 \e^{-mh_I(\nu) \norm{x-x_\nu}}\le \e^{-mh_I(\nu)  \pa{\norm{x-y}-L_\tau}}
 & \text{if}\quad  \norm{y-x_\nu}< L_\tau
\end{cases},
\eeq
so we conclude that
\beq\label{vphivphi}
\abs{\vphi_\nu(x)\vphi_\nu (y)}\le  \e^{-m^\pr h_I(\nu) \norm{x-y}}, \qtx{where} m^\pr \ge m(1-100L^{\tau -1}).
\eeq

Now let 
 $P_{I}= \Chi_{I}\pa{H_ {\Lambda_{L}}}$, $\bar P_{I}= 1-P_{I}$.
Since
\begin{align}
\scal{ \delta_{x},  (H_{\La_L}  -\lambda)^{-1}\e^{-t \pa{H_{\La_L}^2-\lambda^2}}P_{I}\delta_{y}} = \!\! \sum_{\quad \mu \in \sigma_{I}(H_{\La_L})}(\mu  -\lambda)^{-1}\e^{-t (\mu^2-\lambda^2)}\overline{\vphi_\mu (x)}{\vphi_\mu (y)},
\end{align}¥
it follows from \eq{vphivphi} that
\begin{align}\nn
¥& \abs{ \scal{ \delta_{x},  (H_{\La_L}  -\lambda)^{-1}\e^{-t \pa{H_{\La_L}^2-\lambda^2}}P_{I}\delta_{y}} }  \le 
\e^{L^\beta} \hskip-15pt \sum_{\quad \mu \in \sigma_{I}(H_{\La_L})}\e^{-t (\mu^2-\lambda^2)}\abs{\vphi_\mu (x)\vphi_\mu (y)} \\ & \hskip80pt  \le 
   \e^{L^\beta} \hskip-15pt\sum_{\quad \mu \in \sigma_{I}(H_{\La_L})}\e^{-t (\mu^2-\lambda^2)} \e^{-m^\pr h_I(\mu) \norm{x-y}}.
\end{align}¥

We now take 
\beq
t= \tfrac{m^\pr  \norm{x-y}}{{A}^2}  \; \Longrightarrow \; \e^{-t (\mu^2-\lambda^2)}\e^{-m^\pr h_I (\mu){\norm{x-y}}}= \e^{-m^\pr h_I (\lambda){\norm{x-y}}}\mqtx{for} \mu \in I,
\eeq
obtaining
\begin{align}
¥ \abs{ \scal{ \delta_{x},  (H_{\La_L}  -\lambda)^{-1}\e^{-t \pa{H_{\La_L}^2-\lambda^2}}P_{I}\delta_{y}} }\le (L+1)^d  \e^{L^\beta}\e^{-m^\pr h_I (\lambda){\norm{x-y}}}.
\end{align}¥

It follows from Lemma~\ref{lemkey2} that\begin{align}
& \abs{ \scal{ \delta_{x},  (H_{\La_L}  -\lambda)^{-1}\e^{-\tfrac{m^\pr \norm{x-y}}{{A}^2} \pa{H_{\La_L}^2-\lambda^2}}\bar P_{I}\delta_{y}} }\le  \e^{L^\beta} \e^{-m^\pr  h_{I}(\lambda) \norm{x-y}},
\end{align}
so
\begin{align}\label{Presolv}
& \abs{ \scal{ \delta_{x},  (H_{\La_L}  -\lambda)^{-1}\e^{-t \pa{H_{\La_L}^2-\lambda^2}}\delta_{y}} }\le  2 (L+1)^d  \e^{L^\beta}\e^{-m^\pr  h_{I}(\lambda) \norm{x-y}}.
\end{align}¥

It follows from \eq{FHestxy}, using \eq{upbm37},  that
\begin{align}
\label{FHestresolv}
\abs{\scal{\delta_x,F_{\frac{m^\pr \abs{x-y}}{A^2},\lambda}(H_{\La_L})\delta_y}}
\le 70 {A}^{-1} \e^{- m^\pr h_{I}(\lambda) \abs{x-y}}\le 70 m_-^{-1} L^{\kappa^\pr}\e^{- m^\pr h_{I}(\lambda) \abs{x-y}}.
\end{align}
Combining  \eq{resolvdecomp},      \eq{Presolv} and \eq{FHestresolv}, we get
\begin{align}
¥\abs{G_{\La_L}(\lambda;x,y)}  & \le \pa{ 70 m_-^{-1} L^{\kappa^\pr} +  2 (L+1)^d  \e^{L^\beta}} e^{- m^\pr h_{I}(\lambda) \abs{x-y}}.
\end{align}¥

We now require $\lambda \in I_L$, obtaining
\begin{align}
¥\abs{G_{\La_L}(\lambda;x,y)}   \le
 \e^{-m^{\pr\pr}  h_{I}(\lambda) \norm{x-y}},
\end{align}¥
where
\begin{align}
m^{\pr\pr} &\ge m^{\pr}  \pa{1 - C_{d,m_-}L^{-(1 - \beta -\kappa-\kappa^\pr)}}\\ \nn &
\ge  m \pa{1 - C_{d,m_-}L^{-\min \set{1- \tau,1- \beta -\kappa-\kappa^\pr}}}= m \pa{1 - C_{d,m_-}L^{-(1- \tau)}}.
\end{align}¥
\end{proof}

\begin{proposition} 
 Suppose the conclusions of  Theorem~\ref{thmMSA} hold  for an  Anderson model $H_{\bom}$, and  let $I=I_\infty$, $m=m_\infty$. 
   Then, given $0<\zeta^\pr <\xi$ ,  there exists a finite scale   $\cL_1$ such that  for all $L\ge \cL _1$  we have 
  \begin{align}
 \label{concSingleMSA}
\inf_{x\in \R^d} \P\set{\La_{L} (x) \sqtx{is} (m^{\pr\pr}  h_{I}(\lambda), \lambda) \text{-regular} } \ge 1 -  \e^{-L^{\zeta^\pr}}   \sqtx{for all}  \lambda \in    I_{L},
\end{align}
and
\begin{align}\label{concMSA4}
\inf_{\substack{
x,y\in \R^d\\ \norm{x-y} > L}}\!\!
\P\set{\text{for} \;  \lambda \in    I_{L} \sqtx{either} \La_{L} (x) \sqtx{or}\La_{L} (y)\sqtx{is}  (m^{\pr\pr}  h_{I}(\lambda), \lambda)\text{-regular} }\ge 1 -  \e^{-L^{\zeta^\pr}},
\end{align}
where $m^{\pr\pr}$ is given in \eq{mprpr3}.
\end{proposition}

\begin{proof}  Suppose the conclusions of  Theorem~\ref{thmMSA} hold  for an  Anderson model $H_{\bom}$, and  let $I=I_\infty$, $m=m_\infty$, and let $ L\ge L_0^\gamma$.
Since the Wegner estimate gives (see Lemma~\ref{lemSep} for the notation)
  \beq
  \P\set{ \norm{G_{\La_L}(\lambda)}\le \e^{L^\beta}}\ge  1 - \wtilde{K}2^\alpha \e^{-\alpha L^\beta}(L+1)^d\ge 1 - \tfrac 12  \e^{-L^{\zeta^\pr}}\sqtx{for all} \lambda \in \R,
  \eeq
  for large $L$, 
  it follows from \eq{MSALnok2} and Lemma~\ref{lemtoreg}  that for $L$ large we have \eq{concSingleMSA}.

 Now   
 consider two boxes $\La_L(x_1)$ and  $\La_L(x_2)$, where $x_1,x_2 \in \R^d$, $\norm{x_1-x_2} >L$.  Define the events
\begin{align}
\cA&= \set{\La(x_1)\sqtx{and}\La(x_2)\sqtx{are both} ( m, I)  \text{-localizing for} \; H_{\bom}},\\ \notag
\cB&= \set{\dist(\sigma(\La_L(x_1)), \sigma(\La_L(x_2)))\ge 2\e^{-L^{\beta}}}
\end{align}
Since  $\norm{x_1-x_2} >L$, the boxes are disjoint, so it follows from
 \eq{MSALnok2} that 
\beq
\P\set{\cA}\ge 1- 2\e^{-L^\xi}\ge 1 - \tfrac 12  \e^{-L^{\zeta^\pr}},
\eeq
    and  the  Wegner estimate between  boxes
gives
\begin{align} 
&\P\set{\cB} \ge 1 -\wtilde{K}4^\alpha \e^{-\alpha L^\beta}(L+1)^{2d}\ge 1 - \tfrac 12  \e^{-L^{\zeta^\pr}},
\end{align}
so we have
\beq
\P\set{\cA\cap\cB}\ge 1 -  \e^{-L^{\zeta^\pr}}.
\eeq
Moreover, for $\bom \in \cA\cap\cB$ and $\lambda\in \R$,  the boxes $\La(x_1)$ and $\La(x_2)$ are 
both $( m, I)$-localizing, and
we must have either $ \norm{G_{\La_L(x_1)}(\lambda)}\le \e^{L^\beta}$ or  $ \norm{G_{\La_L(x_2)}(\lambda)}\le \e^{L^\beta}$, so for $ \lambda \in    I_{L} $  the previous argument shows that either  $\La(x_1)$ or  $\La(x_2)$ is $(m^{\pr\pr}  h_{I}(\lambda), \lambda)$-regular for large $L$.   We proved \eq{concMSA4}.
\end{proof}

\end{document}